\documentclass[a4paper,oneside,11pt,reqno,english]{amsart}
\calclayout
\usepackage{geometry}
\usepackage{babel}
\usepackage{amssymb,amsmath,amsthm,amsfonts,mathtools,esint,dsfont}
\usepackage{appendix}
\usepackage{enumitem}
\usepackage{xspace} 

\usepackage[linktoc=page]{hyperref}

\newcommand{\R}{\mathbb{R}}
\newcommand{\eps}{\varepsilon}
\newcommand{\1}{\mathds{1}}
\DeclareMathOperator{\tr}{Tr}
\def\di{\mathop{}\!\mathrm{d}}
\def\dit{\di t}
\def\dix{\di x}
\def\diy{\di y}
\def\diz{\di z}
\newcommand{\NLS}{\relax\ifmmode \mathrm{NLS}\else\textup{NLS}\xspace\fi}

\newcommand{\wto}{\rightharpoonup}
\newcommand{\pscal}[1]{\ensuremath{\left\langle #1 \right\rangle}}
\newcommand{\pscalSM}[1]{\ensuremath{\langle #1 \rangle{\vphantom{\pscal{1}}}}}
\newcommand{\normDISPLAY}[1]{{\left\vert\kern-0.25ex\left\vert #1 \right\vert\kern-0.25ex\right\vert}}
\newcommand{\normINLINE}[1]{{\vert\kern-0.25ex\vert #1 \vert\kern-0.25ex\vert{\vphantom{\normDISPLAY{1}}}}}
\newcommand{\norm}[1]{{\mathchoice{\normDISPLAY{#1}}{\normINLINE{#1}}{\normINLINE{#1}}{\normINLINE{#1}}}}
\newcommand{\normSM}[1]{\normINLINE{#1}}
	
\numberwithin{equation}{section}

\theoremstyle{plain}
\newtheorem{theorem}{Theorem}[section]
\newtheorem{lemma}[theorem]{Lemma}

\newcommand\xqed[1]{%
	\leavevmode\unskip\penalty9999 \hbox{}\nobreak\hfill\quad\hbox{#1}%
}
\newcommand\remarkend{\xqed{$\triangle$}}
\makeatletter
	\def\@endtheorem{\remarkend\endtrivlist\@endpefalse }
\makeatother
\theoremstyle{remark}
\newtheorem{remark}[theorem]{Remark}
\newtheorem*{remark*}{Remark}
\makeatletter
	\def\@endtheorem{\endtrivlist\@endpefalse }
\makeatother

\makeatletter
	\@namedef{subjclassname@2020}{\textup{2020} Mathematics Subject Classification}
\makeatother

\allowdisplaybreaks
\begin{document}

\title[On 1D Bose gases with 2- and attractive 3-body interactions]{On one-dimensional Bose gases with two- \\ and (critical) attractive three-body interactions}

\author[D.-T. Nguyen]{Dinh-Thi Nguyen}
	\address{(D.-T. Nguyen) Department of Mathematics, Uppsala University, Box 480, 751 06 Uppsala, Sweden and Ecole Normale Sup{\'e}rieure de Lyon \& CNRS, UMPA (UMR 5669)}
	\email{\href{dinh_thi.nguyen@math.uu.se}{dinh\_thi.nguyen@math.uu.se}, \href{dinh.nguyen@ens-lyon.fr}{dinh.nguyen@ens-lyon.fr}}

\author[J. Ricaud]{Julien Ricaud}
	\address{(J. Ricaud) CMAP, CNRS, {\'E}cole polytechnique, Institut Polytechnique de Paris, 91120 Pa\-lai\-seau, France}
	\email{\href{julien.ricaud@polytechnique.edu}{julien.ricaud@polytechnique.edu}}


\begin{abstract}
	We consider a one-dimensional, trapped, focusing Bose gas where $N$~bosons interact with each other via both a two\nobreakdash-body interaction potential of the form $aN^{\alpha-1} U(N^\alpha(x-y))$ and an attractive three\nobreakdash-body interaction potential of the form $-bN^{2\beta-2} W(N^\beta(x-y,x-z))$, where $a\in\R$, $b,\alpha>0$, $0<\beta<1$, $U, W \geq 0$, and $\int_{\R}U(x)\dix = 1 = \iint_{\R^2} W(x,y) \dix\diy$. The system is stable either for any $a\in\R$ as long as $b<\mathfrak{b} := 3\pi^2/2$ (the critical strength of the 1D focusing quintic nonlinear Schr{\"o}\-dinger equation) or for $a \geq 0$ when $b=\mathfrak{b}$. In the former case, fixing $b \in (0,\mathfrak{b})$, we prove that in the mean\nobreakdash-field limit the many\nobreakdash-body system exhibits the Bose--Einstein condensation on the cubic-quintic \NLS ground states. When assuming $b=b_N \nearrow \mathfrak{b}$ and $a=a_N \to 0$ as $N \to\infty$, with the former convergence being slow enough and ``not faster'' than the latter, we prove that the ground state of the system is fully condensed on the (unique) solution to the quintic \NLS equation. In the latter case $b=\mathfrak{b}$ fixed, we obtain the convergence of many\nobreakdash-body energy for small $\beta$ when $a > 0$ is fixed. Finally, we analyze the behavior of the many\nobreakdash-body ground states when the convergence $b_N \nearrow \mathfrak{b}$ is ``faster'' than the slow enough convergence $0<a_N \searrow 0$.
\end{abstract}

\maketitle

\tableofcontents

\section{Introduction}
After its first observation in gases of alkali atoms in 1995~\cite{AndEnsMatWieCor-95,BraSacTolHul-95,DavMewAndvDrDurKurKet-95,CorWie-02,Ketterle-02}, the Bose--Einstein condensation (BEC) has been intensively studied during the last decades, in both the physics and mathematics communities. In particular, and remarkably, the stability of focusing Bose gas depends crucially on the interaction between particles, and BEC may collapse when the (two\nobreakdash-body) interaction is attractive and the number of particles or the scattering length of the microscopic interaction excesses a critical value, as observed in experiments~\cite{BraSacTolHul-95,KagMurShl-98,SacStoHul-98,GerStrProHul-00,DonClaCorRobCorWie-01}.

In the present paper, we prove and describe a possible explanation for such collapsings, which are still being studied, for a one-dimensional Bose gas trapped in a confining potential that interacts via both two- and attractive three\nobreakdash-body interactions. Note that such BEC with attractive three\nobreakdash-body interactions has recently been observed, e.g., in~\cite{HamLavBou-22}, where quasi-2D and quasi-1D regimes are considered, and that these two regimes are also studied in quantum many-body systems, e.g., see~\cite{CheHol-13,CheHol-17a,Shen-21} (see also~\cite{CheHol-16}). Because, as we will see, the kinetic and potential energies of such system behave the same under scaling after a mean\nobreakdash-field approximation, the system exhibits a critical mass above which it becomes unstable (this mass-critical unstability of such system is specific to dimension one).

Note that in the physics literature, three\nobreakdash-body interactions are often considered as non-conservative forces that put an end to the collapsing process by reducing the number of particles in the trap over time via a phenomenon called three\nobreakdash-body recombination (see, e.g., \cite{KagVarShl-81, deGvdBMulStoVerGlo-86, deGStoVerGlo-88, TieMoeVerSto-02, MoeBoeVer-96, FedReyShl-96, EsrGreBur-99, BedBraHam-00, WebHerMarNagGri-03a, WebHerMarNagGri-03b, LabOHaHucPhiRolPor-04, BraHamKanPla-08, HarKruDeiDreTieDen-13, MakKlaGolCorJin-14, ShoMacKokKha-14, DogGliHilEigCorSmiHad-19, ModShe-22}). This is not what our three\nobreakdash-body interaction is about. Here, we consider on the contrary conservative three\nobreakdash-body attractive interactions. Moreover, here, and not like physics, the three\nobreakdash-body interaction is assumed from the beginning to be not small, so that its effect is more observable for mathematical studies.

We prove that the ground state energy and the ground states of the system are effectively described by the cubic-quintic nonlinear Schr{\"o}\-dinger (\NLS) functional. In other words, the three\nobreakdash-$\R$\nobreakdash-body correlation between particles yields a leading order correction to the ground state energy, which is similar to the \NLS regime of the two\nobreakdash-$\R^2$\nobreakdash-body interaction studied in~\cite{Lewin-ICMP15,LewNamRou-16,LewNamRou-17}. Finally, we investigate two types of collapse regimes of the system (onto the Gagliardo--Nirenberg--Sobolev optimizer) when the number of particles tends to infinity:
\begin{itemize}[leftmargin=1em]
	\item \emph{critical} regimes: the strength of the attractive three\nobreakdash-body interaction tends (from below) to its critical value faster, in order of magnitude, than the strength of the repulsive two\nobreakdash-body interaction tends to zero;
	\item \emph{non-critical} regimes: the strength of the attractive three\nobreakdash-body interaction tends (from below) to its critical value at a speed ``proportional'' (see Theorem~\ref{thm:collapse_NLS} for the precise meaning) to the speed at which the strength of the (not necessarily repulsive) two\nobreakdash-body interaction tends to zero.
\end{itemize}

\medskip
\textbf{Notation.}
For shortness, we denote throughout this paper $\norm{\cdot}_p:=\norm{\cdot}_{L^p(\R^d)}$ when there is no possible confusion on $d$.
	
\subsection{Model}
We consider a system of $N$ identical bosons in~$\R$, described by the nonrelativistic Schr{\"o}\-dinger Hamiltonian
\begin{multline}\label{Hamiltonian}
	H^{N}_{a, b} := \sum_{i=1}^N\left(-\frac{\mathrm{d}^2}{\mathrm{d} x_i^2} + |x_i|^s\right) + \frac{a}{N-1} \sum_{1\leq i<j\leq N} N^{\alpha} U(N^\alpha(x_i-x_j)) \\
	- \frac{b}{(N-1)(N-2)} \sum_{1\leq i<j<k\leq N} N^{2\beta} W(N^\beta(x_i-x_j), N^\beta(x_i-x_k)) \,,
\end{multline}
acting on $\mathfrak{H}^N:=\bigotimes_{\textrm{sym}}^N \mathfrak{H}$, with $\mathfrak{H}:=L^2(\R)$. Here, $N\geq3$ is the number of particles, $s>0$ the power of the trapping potential, $\alpha > 0$ and $\beta>0$ are scaling parameters, while the factors $1/(N-1)$ and $1/((N-1)(N-2))$ of the coupling constants ensure that the kinetic and interaction energies are comparable in the limit $N\to\infty$. The two\nobreakdash-body interaction $aU$ can be either attractive ($a\leq0$) or repulsive ($a\geq0$) and satisfies
\begin{equation}\label{condition:potential_two_body}
	0 \leq U(x)=U(-x) \in L^1 \cap L^2(\R) \quad \text{and} \quad \int_\R U(x) \dix = 1 \,.
\end{equation}
Moreover, the three\nobreakdash-body interaction $bW$ is assumed to be attractive ($b\geq 0$) and to verify
\begin{equation}\label{condition:potential_three_body}
	0 \leq W\in L^1 \cap L^\infty(\R^2)\,, \quad \iint_{\R^2} W(x,y)\dix\diy = 1\,,
\end{equation}
and the symmetry conditions
\begin{equation}\label{condition:potential_three_body_symmetry}
	W(x,y)=W(y,x) \quad \textrm{ and } \quad W(x-y,x-z)=W(y-x,y-z)=W(z-y,z-x) \,.
\end{equation}
Therefore, $|a|$ and $b > 0$ describe the strengths of the two- and three\nobreakdash-body interactions, respectively.
Note that the trapping potential is taken as a power function in order to obtain explicit formulae for expansions and blow-up profiles, but our mean-field limit results would hold the same for a more general potential diverging to $+\infty$ at $\pm\infty$.
Finally, we define the short notation
\begin{equation}\label{Def_W_N}
	U_N(x) \equiv U_{N,\alpha}(x) := N^{\alpha}U(N^\alpha x) \quad \text{and} \quad W_N(x,y) \equiv W_{N,\beta}(x,y): = N^{2\beta}W(N^\beta x, N^\beta y) \,.
\end{equation}

We are interested in the large-$N$ behavior of the ground state energy per particle of~$H^{N}_{a, b}$. Namely,
\begin{equation}\label{energy:quantum}
	E_{a, b}^{\mathrm{Q}, N} := N^{-1} \inf \left\{\pscal{\Psi_N, H^{N}_{a, b} \Psi_N} : \Psi_N\in\mathfrak{H}^N, \norm{\Psi_N}_2=1\right\}
\end{equation}
and the corresponding ground state.	Roughly speaking, BEC occurs when almost all particles occupy the same quantum states. That is, in terms of wave functions, when
\begin{equation}\label{eq:BEC}
	\Psi_N(x_{1},\ldots,x_N) \approx u^{\otimes N}(x_{1},\ldots,x_N) := u(x_{1})\ldots u(x_N) \,.
\end{equation}
Inserting in the energy functional the trial state $u^{\otimes N}$, with the normalization condition $\norm{u}_2=1$, we obtain the Hartree energy functional $\mathcal{E}_{a, b}^{\mathrm{H}, N}(u) := N^{-1} \pscalSM{u^{\otimes N}, H^{N}_{a,b} u^{\otimes N}}$,
which expands as
\begin{multline}\label{functional:hartree}
	\mathcal{E}_{a, b}^{\mathrm{H}, N}(u) = \int_\R \left( |u'(x)|^2 + |x|^s |u(x)|^2 \right )\dix + \frac{a}{2}\iint_{\R^2} U_N(x-y)|u(x)|^2|u(y)|^2\dix\diy \\
		- \frac{b}{6}\iiint_{\R^3} W_N(x-y,x-z)|u(x)|^2|u(y)|^2|u(z)|^2\dix\diy\diz \,.
\end{multline}
This leads to the Hartree ground state energy
\begin{equation}\label{energy:hartree}
	E_{a, b}^{\mathrm{H}, N} := \inf\left\{\mathcal{E}_{a, b}^{\mathrm{H}, N}(u):u\in H^{1}(\R), \norm{u}_2=1\right\},
\end{equation}
which is thus an upper bound to the many\nobreakdash-body ground state energy $E_{a, b}^{\mathrm{Q}, N}$. When $N\to\infty$, since $U_N(x-y)$ and $W_N(x-y,x-z)$ converge respectively to the delta interactions $\delta_{x=y}$ and $\delta_{x=y=z}$, the Hartree functional formally boils down to the nonlinear Schr{\"o}\-dinger functional
\begin{equation}\label{functional:NLS}
	\mathcal{E}_{a,b}^{\NLS}(u) := \int_{\R}\left( |u'(x)|^2 + |x|^s|u(x)|^2 + \frac{a}{2} |u(x)|^4 - \frac{b}{6} |u(x)|^6 \right)\dix \,,
\end{equation}
with associated ground state energy
\begin{equation}\label{energy:NLS}
	E_{a,b}^{\NLS} := \inf\left\{\mathcal{E}^{\NLS}_{a,b}(u):u\in H^{1}(\R), \norm{u}_2=1\right\}.
\end{equation}
Note that, due to our hypotheses, the mass sub-critical term $a|u|^4$ in $\mathcal{E}_{a,b}^{\NLS}$ can be focusing ($a\leq0$) or defocusing ($a\geq0$), but the mass critical term $b|u|^6$ is focusing. Consequently, there exists a critical mass $\mathfrak{b}$ above which the \NLS functional is unstable. More precisely, for any $a\in \R$, the energy $E_{a,b}^{\NLS}$ is unbounded from below if $b > \mathfrak{b}$, where $\mathfrak{b} := 3\pi^2/2$ is the optimal constant in the Gagliardo--Nirenberg--Sobolev inequality
\begin{equation}\label{ineq:GNS}
	\frac{\mathfrak{b}}{6}\int_{\R}|u(x)|^6\dix \leq \left( \int_{\R}|u'(x)|^2\dix \right) \left( \int_{\R}|u(x)|^2\dix \right)^2, \quad \forall\, u\in H^1(\R)\,.
\end{equation}
It is well-known (see, e.g., \cite{BenLos-04,DolEstLapLos-14}) that~\eqref{ineq:GNS} has a positive radial optimizer $Q\in H^1(\R)$, which is the unique optimizer up to translations, multiplication by a complex factor, and scaling. After multiplication by a constant and after scaling, such optimizer solves the quintic \NLS equation
\[
	-Q'' + Q - |Q|^4Q = 0 \,.
\]
For simplicity, in this paper we choose one specific optimizer by fixing the translation, the factor, and the scaling as
\begin{equation}\label{Def_Q0}
	Q_0(x) := \frac{1}{\sqrt{\cosh \pi x}}\,.
\end{equation}
This optimizer~$Q_0$ is normalized, $\norm{Q_0}_2=1$, verifies $\norm{Q_0'}_2^2=\pi^2/8$ and $\norm{Q_0}_6^2=1/2$, and solves
\begin{equation}\label{qNLS_solution}
	-Q_0'' + \frac{\pi^2}4 Q_0 - \frac{3}{4} \pi^2 |Q_0|^4Q_0 = 0 \,.
\end{equation}

In the following, we always assume that $0<b\leq\mathfrak{b}$ and we define $\mathcal{M}^{\NLS}$ as the set of \NLS ground states. Below the critical mass $\mathfrak{b}$, the existence of \NLS ground states follows from standard methods in the calculus of variations (see section~\ref{Section_Proof_existence_NLS}). On the other hand, in the mass critical case $b=\mathfrak{b}$, \NLS ground states still exist if and only if $a>0$.
	
The effective expression~\eqref{functional:NLS} is analogous to the \NLS functional with two\nobreakdash-$\R^2$\nobreakdash-body interaction. The validity of the \NLS theory for the ground states and ground state energies with an attractive two\nobreakdash-$\R^2$\nobreakdash-body interaction potential has been proved in seminal\footnote{In the authors' opinions.} papers of Lewin, Nam, and Rougerie~\cite{Lewin-ICMP15,LewNamRou-16,LewNamRou-17,NamRou-20}, see also~\cite{CheHol-17b}. Furthermore, the collapse and condensation was studied in~\cite{LewNamRou-18-proc} (see also~\cite{GuoSei-14,Nguyen-20}). Our aim here is to extend these results to the three\nobreakdash-$\R$\nobreakdash-body interaction case. 

One might consider the system~\eqref{Hamiltonian} in higher dimension. In that case, the \NLS functional is mass super-critical with respect to the three\nobreakdash-body interaction. For the stability of the quantum problem, it is necessary for the three\nobreakdash-body term to be repulsive. In such defocusing case, the mean\nobreakdash-field approximation was obtained by Nam, Ricaud, and Triay~\cite{NamRicTri-23} in dimension three, while the derivation of the time-dependent defocusing quintic \NLS from many\nobreakdash-body quantum dynamics with repulsive three\nobreakdash-body interactions was investigated in~\cite{ChePav-11,Chen-12,Yuan-15,Xie-15,CheHol-19,NamSal-20,LiYao-21}.

Note however that, except for~\cite{Xie-15,LiYao-21}\footnote{And to some extent~\cite{ChePav-11} given that the authors mention (see p.~963) that their result can be extended to a combination of two- and three\nobreakdash-body interactions having the same scaling parameters $\beta$.}, which are about generic multi-body interactions, all these papers ---as well for instance as~\cite{FerMusTro-08} about one-dimensional \NLS ground states--- concerned with the defocusing case deal with only a three-body interaction. Therefore, together with the fact that we consider attractive three-body interactions, one of the original aspects of the present paper is to consider at the same time two- and three-body interactions. Finally, note that in~\cite{FerMalMusTro-09} the \NLS ground states are studied in a similar settings with a repulsive two- and an attractive three-body interactions, but with a periodic rather than a trapping potential.

\subsection{\NLS theory}
In the first part of the paper, we consider the minimization problem~\eqref{energy:NLS}. Our first result is a classification of the values of the parameters $a$ and $b$ for which there exist \NLS ground states.
\begin{theorem}[\textbf{Existence of the \NLS ground states}]\label{thm:existence_NLS}\leavevmode\\
	Let $a\in\R$, $b>0$, and $E_{a,b}^{\NLS}$ be given in~\eqref{energy:NLS}. We have the following
	\begin{enumerate}[label=(\roman*)]
		\item\label{NLS_GS_inexistence} If $b>\mathfrak{b}$ or ($b=\mathfrak{b}$ and $a<0$), then $E_{a,b}^{\NLS} = -\infty$.
		\item\label{NLS_GS_existence} If $b<\mathfrak{b}$ or ($b=\mathfrak{b}$ and $a>0$), then $E_{a,b}^{\NLS}$ has a ground state.
		\item\label{NLS_GS_critical} If $b=\mathfrak{b}$ and $a=0$, then $E_{a,b}^{\NLS} = 0$, but $E_{a,b}^{\NLS}$ has no ground states.
	\end{enumerate}
\end{theorem}

Note that we can restrict the minimization problem~\eqref{energy:NLS} to nonnegative functions $u$ since
$\mathcal{E}_{a,b}^{\NLS}(u) \geq \mathcal{E}_{a,b}^{\NLS}(|u|)$ for any $u \in H^1(\R)$. This follows from the fact that $\norm{\nabla u}_2 \geq \norm{\nabla |u|}_2$ (see, e.g.,~\cite[Theorem~7.8]{LieLos-01}). In particular, the ground state of~$E_{a,b}^{\NLS}$, when it exists, can be chosen to be nonnegative. The proof of the case $b=\mathfrak{b}$ is special: the compactness of the minimizing sequence for $E_{a, \mathfrak{b}}^{\NLS}$ cannot be obtained directly using~\eqref{ineq:GNS}. By refined arguments using the concentration-compactness lemma and the singularity of the repulsive cubic term, we will recover this property.

Our next result concerns the collapse of the \NLS ground states. By Theorem~\ref{thm:existence_NLS}, we see that the blow-up phenomenon of \NLS ground-states can only occur when $(a,b)\to(0,\mathfrak{b})$ in such a way that the condition~\emph{\ref{NLS_GS_existence}} of Theorem~\ref{thm:existence_NLS} stays fulfilled.
By the variational principle, we have
\begin{equation}\label{NLS_energy_upperbound}
	E_{a,b}^{\NLS} \leq \mathcal{E}_{a,b}^{\NLS}\left(\ell^{\frac{1}{2}}Q_0(\ell \cdot)\right) = \ell^2\frac{\mathfrak{b}-b}{12} + \ell\frac{a}{\pi} + \frac{\ell^{-s}}{s} \mathcal{Q}_s
\end{equation}
for all $\ell > 0$, where $Q_0$ is defined in~\eqref{Def_Q0} and where, from now on and for any $s>0$, we define
\begin{equation}\label{Def_Q_s}
	\mathcal{Q}_s:= s\int_{\R}|x|^s |Q_0(x)|^{2}\dix < \infty\,.
\end{equation}	
By optimizing the right hand side of~\eqref{NLS_energy_upperbound} over~$\ell > 0$, it gives an upper bound to the \NLS energy and allows to determine explicitly the blow-up profile of the \NLS ground states, depending on the collapse regime. We can actually also find the exact expansion of the corresponding energy. Those results are presented in the following theorem.

\begin{theorem}[\textbf{Collapse of the \NLS ground states}]\label{thm:collapse_NLS}\leavevmode\\
	Let $s>0$, $\zeta\geq0$, and $\mathcal{Q}_s$ as in~\eqref{Def_Q_s}. Let $\{a_n\}_n\subset \R$ and $\{b_n\}_n\subset (0,\mathfrak{b}]$ satisfy $a_n \to 0$ and $b_n \nearrow \mathfrak{b}$ as $n\to+\infty$, $\max\{0,a_n\} + (\mathfrak{b}-b_n)>0$, and
	\begin{equation}\label{thm_collapse_NLS_condition_ratio_sequences}
		\lim\limits_{n\to+\infty} a_n (\mathfrak{b} - b_n)^{ -\frac{s+1}{s+2} } = \pi \frac{6-\zeta}{6} \zeta^{-\frac{s+1}{s+2}} \mathcal{Q}_s^{\frac{1}{s+2}} \in (-\infty,+\infty]\,.
	\end{equation}
	Let $\{u_n\}_n$ be a sequence of (approximate) ground states of~$E_{a_n,b_n}^{\NLS}$ defined in~\eqref{energy:NLS} and $Q_0$ be given in~\eqref{Def_Q0}. Then,
	\begin{equation}\label{blowup_NLS_ground_state}
		\lim_{n\to\infty} \sqrt{\ell_n} u_n\left( \ell_n \cdot \right) = Q_0
	\end{equation}
	strongly in $H^1(\R)$, for
	\begin{equation}\label{thm_collapse_NLS_ell_n_choice}
		\ell_n =
		\left\{
		\begin{aligned}
			&\left(\frac{6a_n}{(6-\zeta)\pi \mathcal{Q}_s } \right)^{\frac{1}{s+1}} &\text{if } \zeta\neq6\,, \\
			&\left(\frac{\mathfrak{b} - b_n}{ \zeta \mathcal{Q}_s } \right)^{\frac{1}{s+2}} &\text{if } \zeta\neq0\,.
		\end{aligned}
		\right.
	\end{equation}
	Furthermore,
	\begin{equation}\label{blowup_NLS_energy}
		E_{a_n,b_n}^{\NLS} = \left( \frac{s+1}{s} - \frac{\zeta}{12} + o(1)\right) \mathcal{Q}_s \ell_n^s \,.
	\end{equation}
\end{theorem}
We recall that a sequence $\{u_n\}_n$ of \emph{approximate} ground states of~$E_{a_n,b_n}^{\NLS}$ is a sequence belonging to the minimizing domain and satisfying, as $n\to+\infty$, the property
\[
	E_{a_n,b_n}^{\NLS} \leq \mathcal{E}_{a_n,b_n}^{\NLS}(u_n) \leq E_{a_n,b_n}^{\NLS} + o(1)\,.
\]

\begin{remark}[Technicalities about Theorem~\ref{thm:collapse_NLS}]\label{Rmk_thm_collapse_nls}\leavevmode
	\begin{itemize}[leftmargin=1em]
		\item The purpose of the hypothesis $\max\{0,a_n\} + (\mathfrak{b}-b_n)>0 \Leftrightarrow (b_n<\mathfrak{b} \text{ or } (b_n=\mathfrak{b} \text{ and } a_n>0))$ is to guarantee, thanks to Theorem~\ref{thm:existence_NLS}, the existence of \NLS ground states \emph{for any $n$}. The existence for large $n$ being guaranteed by the limit condition.
		\item If $\zeta \neq 0$, then $\mathfrak{b}-b_n>0$ for $n$ large enough by~\eqref{thm_collapse_NLS_condition_ratio_sequences}. Thence, without loss of generality, we assume $\{b_n\}_n\subset (0,\mathfrak{b})$ if $\zeta \neq 0$.
		\item If $\zeta \neq 6$, then $a_n/(6-\zeta)>0$ for $n$ large enough by~\eqref{thm_collapse_NLS_condition_ratio_sequences}. Thence, without loss of generality, we assume $\left\{a_n/(6-\zeta)\right\}_n\subset (0,+\infty)$ if $\zeta \neq 6$. This guarantees for the upper formula for $\ell_n$ in~\eqref{thm_collapse_NLS_ell_n_choice} to make sense for all $n$.
		\item The complicated factor $(1-\zeta/6)\zeta^{-(s+1)/(s+2)}$ ---which is one-to-one decreasing as a function from $[0, +\infty)$ to $(-\infty,+\infty]$--- is here to avoid for $\zeta$ to be defined implicitly: if one replaces this factor by some $L\in(-\infty,+\infty]$, then the formulae in~\eqref{thm_collapse_NLS_ell_n_choice}--\eqref{blowup_NLS_energy} stay the same but with their $\zeta$ defined implicitly as the unique positive solution to $\zeta^{1/(s+2)}/6 + L = \zeta^{-(s+1)/(s+2)}$.
	\end{itemize}
\end{remark}
The reader will notice that our result in Theorem~\ref{thm:collapse_NLS} covers all possible ``proportionality'' in terms of the respective speeds of convergence of $a_n$ and $b_n$ ---that is, $\zeta\in(0,+\infty)$---, as well as the case where $b_n$ converges faster, in order of magnitude, than $a_n$ ---that is, $\zeta=0$---, but that it does not cover the converse case where $a_n$ converges faster, in order of magnitude, than $b_n$ ---that would be, $\zeta=+\infty$. The reason is that in the limit, and roughly speaking, the two former cases stay within the framework of \emph{\ref{NLS_GS_existence}} in Theorem~\ref{thm:existence_NLS}, where ground states exist, while the latter case $\zeta=+\infty$ tends, still roughly speaking, to the framework of \emph{\ref{NLS_GS_inexistence}} in Theorem~\ref{thm:existence_NLS} (more precisely to the framework $b=\mathfrak{b}$ and $a<0$) where the energy is not even bounded.

\subsection{Hartree theory}
The Hartree theory can be interpreted as an interpolation theory between the many\nobreakdash-body and \NLS theories. A feature of the Hartree theory is that it still describes correctly the two- and three\nobreakdash-body interactions while the Hartree functional has fewer variables and is nonlinear. The properties of the Hartree energy and its ground states are then easier to obtain. Before turning to our results on the many\nobreakdash-body theory, we state the results in the associated Hartree theory as they will be needed and because they are of their own interest. 
\begin{theorem}[\textbf{Condensation and collapse of the Hartree ground states}]\label{thm:collapse_Hartree}\leavevmode\\
	Let $\alpha, \beta>0$ and assume that $U$ and $W$ satisfy~\eqref{condition:potential_two_body}--\eqref{condition:potential_three_body_symmetry}.
	
	\begin{itemize}[leftmargin=2em]
		\item[(i)] Let $a \in \R$ and $0 < b \leq \mathfrak{b}$ be fixed with either $0 < b < \mathfrak{b}$ or ($b = \mathfrak{b}$ and $a>0$ and $\alpha > \beta$), and let $\{u_N\}_N$ be a sequence of (approximate) ground states of~$E_{a, b}^{\mathrm{H}, N}$ defined in~\eqref{energy:hartree}. Then, there exists a \NLS ground state $\phi$ of~\eqref{energy:NLS} such that
		\begin{equation}\label{cv:Hartree_to_NLS_ground_state}
			\lim_{N\to\infty}u_N = \phi
		\end{equation}
		strongly in $H^1(\R)$. Furthermore,
		\begin{equation}\label{cv:Hartree_to_NLS_energy}
			\lim_{N\to\infty}E_{a, b}^{\mathrm{H}, N} = E_{a,b}^{\NLS}.
		\end{equation}
		
		\item[(ii)] Assume $xU(x) \in L^1(\R)$ and $x W(x) \in L^1(\R^2)$, and let $\zeta$, $\mathcal{Q}_s$, $\{a_N\}_N$, $\{b_N\}_N$, and $\{\ell_N\}_N$ be as in Theorem~\ref{thm:collapse_NLS}, with $\ell_N \sim N^{-\eta}$ for
		\[
			0 < \eta < \min\left\{\frac{\beta}{s+3}, \alpha\right\}.
		\]
		Finally, assume $\alpha>\beta$ if $\zeta=0$. Let $\{u_N\}_N$ be a sequence of (approximate) ground states of~$E_{a_N, b_N}^{\mathrm{H}, N}$ and $Q_0$ be given in~\eqref{Def_Q0}. Then,
		\begin{equation}\label{blowup_Hartree_ground_state}
			\lim_{N\to\infty} \ell_N^{1/2} u_N ( \ell_N \cdot ) = Q_0
		\end{equation}
		strongly in $H^1(\R)$ and
		\[
			E_{a_N, b_N}^{\mathrm{H}, N} = E_{a_N, b_N}^{\NLS} (1+o(1)) = \left( \frac{s+1}{s} - \frac{\zeta}{12} + o(1) \right) \mathcal{Q}_s \ell_N^s\,.
		\]
	\end{itemize}
\end{theorem}
Similarly to the \NLS case, a sequence $\{u_N\}_N$ of \emph{approximate} ground states of~$E_{a_N, b_N}^{\mathrm{H}, N}$ is a sequence belonging to the minimizing domain and satisfying, as $N\to+\infty$, the property
\[
	E_{a_N, b_N}^{\mathrm{H}, N} \leq \mathcal{E}_{a_N, b_N}^{\mathrm{H}, N}(u_N) \leq E_{a_N, b_N}^{\mathrm{H}, N} + o(1)\,.
\]

Notice that the condition $\ell_N \sim N^{-\eta}$ is, by definition of $\ell_N$ in~\eqref{thm_collapse_NLS_ell_n_choice}, also a condition on $\mathfrak{b}-b_N$ or on $a_N$, depending on the value of $\zeta$. The technical assumptions $xU(x) \in L^1(\R)$ and $x W(x) \in L^1(\R^2)$ are used to determine the convergence rate of the Hartree energy to the \NLS energy in the limit $N \to \infty$. The condition $\eta<\beta/(s+3)$ is then used to ensure that the Hartree and \NLS ground state problems are close in the collapse regime. Moreover, the extra condition $\alpha>\beta$ applies in particular to the case where either $\{b_N\}_N$ equal to~$\mathfrak{b}$ or $b_N \nearrow \mathfrak{b}$ faster than $a_{N} \searrow 0$. This ensures the singularity of the two\nobreakdash-body interaction, and this corresponds to the case $\zeta = 0$ in Theorem~\ref{thm:collapse_NLS}.

\subsection{Many-body theory}
In this part of the paper, we turn to the $N$-particle Hamiltonian~\eqref{Hamiltonian} with two- and three\nobreakdash-body interaction potentials of the form~\eqref{Def_W_N}. In the stable regime, we verify the validity of the effective cubic-quintic \NLS~\eqref{functional:NLS}. As usual, the convergence of ground states is formulated using $k$-particles reduced density matrices, defined for any $\Psi_N\in\mathfrak{H}^N$ by the partial trace
\[
	\gamma_{\Psi_N}^{(k)} := \tr_{k+1\to N} | \Psi_N \rangle \langle \Psi_N | \,.
\]
Equivalently, $\gamma_{\Psi_N}^{(k)}$ is the trace class operator on $\mathfrak{H}^k$ with kernel
\[
	\gamma_{\Psi_N}^{(k)}(x_{1},\ldots,x_k;y_{1},\ldots,y_k) := \int_{\R^{N-k}}\overline{\Psi_N(x_{1},\ldots,x_k;Z)}{\Psi_N(y_{1},\ldots,y_k;Z)} \di Z \,.
\]
One of the main advantages of the reduced density matrices is that we can write
\[
	\frac{1}{N} \pscal{ \Psi_N, H_{a, b}^{N} \Psi_N } = \tr\left[h\gamma_{\Psi_N}^{(1)}\right] + \frac{a}{2}\tr\left[U_N\gamma_N^{(2)}\right] - \frac{b}{6}\tr\left[W_N\gamma_N^{(3)}\right] = \frac{1}{3}\tr\left[H_3\gamma_{\Psi_N}^{(3)}\right],
\]
where
\begin{equation}\label{eq:one_particle_operator}
	h_x := -\frac{\mathrm{d}^2}{\mathrm{d} x^2} + |x|^s
\end{equation}
is the one-particle operator and $H_3$ is the three-particles Hamiltonian
\[
	H_3 := h_{x_{1}} + h_{x_{2}} + h_{x_3} + \frac{a}{2} U_N(x_{1}-x_{2}) + \frac{a}{2} U_N(x_{2}-x_3) + \frac{a}{2} U_N(x_3-x_{1}) - \frac{b}{2} W_N(x_{1}-x_{2},x_{1}-x_3)\,.
\]
Furthermore, the Bose--Einstein condensation~\eqref{eq:BEC} is characterized properly by
\[
	\lim_{N\to\infty} \tr \left|\gamma_{\Psi_N}^{(k)} - |u^{\otimes k}\rangle \langle u^{\otimes k}|\right| = 0, \quad \forall\, k=1,2,\dots\,.
\]

With Theorem~\ref{thm:existence_NLS} in mind, we distinguish two cases: critical three\nobreakdash-body interactions, in the sense of regimes where $b$ converges so fast towards $\mathfrak{b}$ compared to $a$ towards $0$ that it is the sign of $a$ that determines if ground states exist, and the non-critical case.

\subsubsection{Critical three-body interactions}
We consider here the critical case for the three\nobreakdash-body interaction. For the mean\nobreakdash-field regime, this criticality corresponds to the case where $a$ and $b = \mathfrak{b}$ are fixed. In this context, the effective \NLS minimization problem\eqref{energy:NLS} does not make sense when $a<0$ for the two\nobreakdash-body interaction. In the complementary situation, where $a \geq 0$ (fixed), we derive an energy convergence for small $\beta>0$. For the collapse regime, this criticality corresponds to $\zeta=0$ in Theorem~\ref{thm:collapse_NLS}, and we prove in this situation that the many\nobreakdash-body ground states have a universal blow-up profile described by the (unique) solution of the quintic \NLS equation~\eqref{qNLS_solution}.

In the following result, we emphasize that~\emph{(ii)}, about collapse regimes, covers in particular the case of the strictly critical three\nobreakdash-body interaction where $b_N=\mathfrak{b}$ for all $N$, since it is a special case of $\zeta=0$.
\begin{theorem}[\textbf{Condensation and collapse of the many\nobreakdash-body ground states: critical regimes}]\label{thm:many_body_critical}
	Let $0<\beta<\min\left\{\alpha,s/(9s+6)\right\}$ and assume that $U$ and $W$ satisfy~\eqref{condition:potential_two_body}--\eqref{condition:potential_three_body_symmetry}.
	\begin{itemize}[leftmargin=2em]
		\item[(i)] Let $a > 0$ and $b = \mathfrak{b}$ be fixed. Then,
		\begin{align}\label{cv:energy_quantum_to_NLS_critical}
			\lim_{N\to \infty} E_{a, \mathfrak{b}}^{\mathrm{Q}, N} = E_{a, \mathfrak{b}}^{\NLS} >-\infty \,.
		\end{align}
		\item[(ii)] Assume $xU(x) \in L^1(\R)$ and $x W(x) \in L^1(\R^2)$, and let $\zeta$, $\mathcal{Q}_s$, $\{a_N\}_N$, $\{b_N\}_N$, and $\{\ell_N\}_N$ be as in Theorem~\ref{thm:collapse_NLS}, with $\zeta=0$ (hence $a_N>0$) and $\ell_N \sim N^{-\eta}$ for
		\begin{equation}\label{collapse:speed_critical}
			0<\eta<\min\left\{\frac{\beta}{s+3}, \frac{s-3\beta(3s+2)}{2s(2s+1)}\right\}\,.
		\end{equation}
		Then,
		\begin{equation}\label{asymptotic:quantum_energy_critical}
			E_{a_N, b_N}^{\mathrm{Q}, N} = E_{a_N, b_N}^{\NLS} \left( 1 + o(1) \right) = \frac{s+1}{s} \mathcal{Q}_s \ell_N^s ( 1 + o(1) ) \,.
		\end{equation}
		Moreover, if $\Phi_N = \ell_N^{N/2} \Psi_N(\ell_N\cdot)$ with $\Psi_N$ a ground state of~$E_{a_N, b_N}^{\mathrm{Q}, N}$, then
		\begin{equation}\label{asymptotic:many_body_ground_state_critical}
			\lim_{N\to\infty} \tr \left| \gamma_{\Phi_N}^{(k)} - | Q_0^{\otimes k} \rangle \langle Q_0^{\otimes k} | \right| = 0, \quad \forall\, k=1,2,\dots\,,
		\end{equation}
		where $Q_0$ is given in~\eqref{Def_Q0}.
	\end{itemize}
\end{theorem}

Our method is the quantitative quantum de~Finetti theorem developed by Lewin, Nam, and Rougerie \cite{LewNamRou-14,LewNamRou-16}. By using its information-theoretic version \cite{Rougerie-20a,Rougerie-20b}, one can obtain the convergence of energy~\eqref{cv:energy_quantum_to_NLS_critical} for $0<\beta<1/12$, which is a larger range than the $s$-dependence upper bound of $\beta$ in Theorem~\ref{thm:many_body_critical} only if $s<2$.

On one hand, we note that \eqref{cv:energy_quantum_to_NLS_critical} still holds true in the case $b=\mathfrak{b}$ and $a=0$. Indeed, for $0<\beta<s/(3(3s+2))$, we have
\[
	\lim_{N\to \infty} E_{0, \mathfrak{b}}^{\mathrm{Q}, N} = E_{0, \mathfrak{b}}^{\NLS} =0\,.
\]
However, the convergence of many\nobreakdash-body ground states is not expected in this case since $E_{0, \mathfrak{b}}^{\NLS}$ does not admit ground states. On the other hand, while the existence of the ground states of~$E_{a, \mathfrak{b}}^{\mathrm{Q}, N}$, as well as of~$E_{a, \mathfrak{b}}^{\NLS}$, still holds true in the case of fixed $a > 0$, the convergence of many\nobreakdash-body ground states in the mean\nobreakdash-field limit when $N \to \infty$ seems difficult to obtain. This is essentially due to the lack of compactness of the many\nobreakdash-body ground states and the linearity of the system~\eqref{Hamiltonian}.
However, under the assumption that the \NLS minimization problem $E_{a,\mathfrak{b}}^{\NLS}$ has a unique ground state $u_0$, we can prove that
\[
	\lim_{N\to\infty} \tr \left| \gamma_{\Psi_N}^{(k)} - | u_0^{\otimes k} \rangle \langle u_0^{\otimes k} | \right| = 0, \quad \forall\, k=1,2,\dots\,,
\]
for the whole sequence $\{\Psi_N\}_N$ of ground states of~$E_{a, \mathfrak{b}}^{\mathrm{Q}, N}$ with $a>0$ fixed. This is a consequence of a Feynman--Hellman-type argument which relies strongly on the energy convergence~\eqref{cv:energy_quantum_to_NLS_critical} and the uniqueness of the limiting profile. In this case, the proof is similar to the one of~\eqref{asymptotic:many_body_ground_state_critical}. However, the uniqueness of the \NLS ground states with the repulsive cubic term and the critical attractive quintic term seems to be open. To conclude about the mean\nobreakdash-field limit, and for comparison, note that we obtained in Theorem~\ref{thm:collapse_Hartree} the convergence~\eqref{cv:Hartree_to_NLS_ground_state}, in the mean\nobreakdash-field limit, of the Hartree ground states of the \NLS ground states. What make this possible is the nonlinearity of the Hartree functional $\mathcal{E}_{a, \mathfrak{b}}^{{\rm H}, N}$ in~\eqref{functional:hartree}.

Finally, the restriction $\beta<\alpha$ is only inherited from the comparison between Hartree and \NLS. That is, in order to compare the many\nobreakdash-body problem to the \NLS one, we compare the Hartree problem to the latter (see Theorem~\ref{thm:collapse_Hartree}) then the former to the Hartree problem, and we require the restriction $\beta<\alpha$ only in the first comparison, not in the second one.

\subsubsection{Non-critical three-body interactions}
Next, we consider the BEC in the non-critical case for the three\nobreakdash-body interaction. For the mean\nobreakdash-field regime, this non-criticality corresponds to the case where $0<b<\mathfrak{b}$ and $a \in \R$ are fixed. In this context, the compactness of the many\nobreakdash-body ground states is easily obtained and we are able to derive its condensation in the mean\nobreakdash-field limit $N\to\infty$ by convex analysis. For the collapse regime, this non-criticality corresponds to $\zeta>0$ in Theorem~\ref{thm:collapse_NLS}, which covers the situations where the convergence $b_N \nearrow \mathfrak{b}$ is either ``proportional'' to or slower (in order of magnitude) than the convergence $a_N \to 0$, as $N\to\infty$. As in the critical case $\zeta=0$, we prove that the many\nobreakdash-body ground states have a universal blow-up profile described by the (unique) solution of the quintic \NLS equation~\eqref{qNLS_solution}.

We recall that $\mathcal{M}^{\NLS}$ is the set of \NLS ground states. We have the following result.
\begin{theorem}[\textbf{Condensation and collapse of the many--body ground states: non-critical regimes}]\label{thm:many_body_noncritical}
	Let $0<\alpha,\beta<1$. Let $U$ and $W$ satisfy~\eqref{condition:potential_two_body}--\eqref{condition:potential_three_body_symmetry} with $\nabla_{1}W \in L^q(\R\times\R, \R)$ for any $q>1$.
	\begin{itemize}[leftmargin=2em]
		\item[(i)] Let $a\in \R$ and $0<b<\mathfrak{b}$ be fixed. Then,
			\begin{align}\label{cv:energy_quantum_to_NLS_noncritical}
				\lim_{N\to \infty} E_{a, b}^{\mathrm{Q}, N} = E_{a,b}^{\NLS} >-\infty \,.
			\end{align}
			Moreover, for any ground state $\Psi_N$ of $E_{a, b}^{\mathrm{Q}, N}$, there exists a Borel probability measure $\mu$ supported on $\mathcal{M}^{\NLS}$ such that, along a subsequence,
			\begin{equation}\label{cv:ground_state_quantum_to_NLS_noncritical}
				\lim_{N \to \infty}\tr \left| \gamma_{\Psi_N}^{(k)} - \int |u^{\otimes k} \rangle \langle u^{\otimes k}| \di\mu(u) \right| =0, \quad \forall\, k=1,2,\dots\,.
			\end{equation}
			If $\mathcal{M}^{\NLS}$ has a unique ground state $u_0$ (up to a phase), then for the whole sequence,
			\[
				\lim_{N \to \infty}\tr \left| \gamma_{\Psi_N}^{(k)} - |u_0^{\otimes k} \rangle \langle u_0^{\otimes k}| \right| =0, \quad \forall\, k=1,2,\dots\,.
			\]
		\item[(ii)] Assume $xU(x) \in L^1(\R)$ and $x W(x) \in L^1(\R^2)$, and let $\zeta$, $\mathcal{Q}_s$, $\{a_N\}_N$, $\{b_N\}_N$ and $\{\ell_N\}_N$ be as in Theorem~\ref{thm:collapse_NLS} with $\zeta \ne 0$ (hence $0<b_N<\mathfrak{b}$), $\zeta \ne 12(s+1)/s$, and $\ell_N \sim N^{-\eta}$ for
			\begin{equation}\label{collapse:speed_noncritical}
				0<\eta<\min\left\{\alpha, \frac{\beta}{s+3}, \frac{1}{10s}\right\}.
			\end{equation}
			Then,
			\begin{equation}\label{asymptotic:quantum_energy_noncritical}
				E_{a_N, b_N}^{\mathrm{Q}, N} = E_{a_N, b_N}^{\NLS} \left( 1 + o(1) \right) = \left( \frac{s+1}{s} - \frac{\zeta}{12} + o(1) \right) \mathcal{Q}_s \ell_N^s\,.
			\end{equation}
			Moreover, if additionally $\eta < 1/(10(s+2))$ and $\Phi_N = \ell_N^{N/2} \Psi_N(\ell_N\cdot)$ with $\Psi_N$ a ground state of~$E_{a_N, b_N}^{\mathrm{Q}, N}$, then
			\begin{equation}\label{asymptotic:many_body_ground_state_noncritical}
				\lim_{N\to\infty} \tr \left| \gamma_{\Phi_N}^{(k)} - | Q_0^{\otimes k} \rangle \langle Q_0^{\otimes k} | \right| = 0, \quad \forall\, k=1,2,\dots\,,
			\end{equation}
			where $Q_0$ is given in~\eqref{Def_Q0}.
	\end{itemize}
\end{theorem}
	
	By the assumption $0<\alpha, \beta<1$ in Theorem~\ref{thm:many_body_noncritical}, we consider mean\nobreakdash-field (by opposition to dilute) two\nobreakdash-$\R$\nobreakdash-body and three\nobreakdash-$\R$\nobreakdash-body interactions. For comparison, our assumptions in Theorem~\ref{thm:many_body_critical} was much weaker on $\alpha$ but stronger on $\beta$: we obtained there the convergence of the many\nobreakdash-body energy for any $\alpha>0$, but only for $\beta > 0$. It would be very interesting to extend the above result to the dilute limit. We expect the threshold for the three\nobreakdash-$\R$\nobreakdash-body interactions to be at $\beta = 1$, above which the particles meet rarely but interact strongly. In other words, the length scale of the interactions $N^{-\beta}$ is smaller than the average distance between particles $N^{-1}$ and the interaction strength $N^{2\beta-2}$ is large. This is the same as for two\nobreakdash-$\R^2$\nobreakdash-body interactions. As it happens, three\nobreakdash-body interactions in~1D behave similarly to two\nobreakdash-body interactions in~2D, which is related to the fact that in both cases the corresponding NLS term is $L^2$-critical. However, and contrarily to the case of the two\nobreakdash-$\R$\nobreakdash-body interaction (see, e.g.,~\cite[Lemma~3.7]{LewNamRou-16}), we have no immediate control on the three\nobreakdash-$\R$\nobreakdash-body interaction by a Sobolev inequality in 1D. This complicates our analysis in many places. Finally, note that the assumption $\zeta \ne 12(s+1)/s$ is related to the fact that the main order term in the \NLS energy (hence in the Hartree energy) is trivial at this value of $\zeta$.
	
	Our method of proof is the information-theoretic quantum de~Finetti theory \cite{Rougerie-20a,Rougerie-20b} and the second estimate as in~\cite{ErdSchYau-10}. Unfortunately with regards to the constraint on $\beta$, the arguments using moments estimates as in~\cite{LewNamRou-16,LewNamRou-17,NamRou-20} are limited to the case $0 < \alpha, \beta < 1$. It also requires the strict inequalities $b, b_N < \mathfrak{b}$ and our technical assumption $\nabla_{1}W \in L^q(\R^2)$ for any~$q>1$, which can actually be relaxed into $\nabla_{1}W \in L^q(\R^2)$ for any~$q\in(1, q_*)$ for some $q_*>1$, with no change in the proofs. Note that it can even be relaxed into $\nabla_{1}W \in L^q(\R^2)$ for some~$q\in(1,2)$ but at the price of the further restriction $0<\beta<q/(3q-2)$ on $\beta$ (see the discussion at the end of the proof of Lemma~\ref{Lemma_Moments_Estimates}).

	Finally, we emphasize that, within the (constrained) two-fold context of the collapse regimes $a_N \to 0$ and $b_N \nearrow \mathfrak{b}$ covered by~Theorem~\ref{thm:collapse_NLS} and of small enough $\alpha$'s, $\beta$'s, and $\eta$'s\footnote{That is, for $\alpha$'s, $\beta$'s, and $\eta$'s satisfying at the same times the constraints on them due to Theorems~\ref{thm:many_body_critical} and to~\ref{thm:many_body_noncritical}.}, Theorems~\ref{thm:many_body_critical} and~\ref{thm:many_body_noncritical} describe the blow-up phenomenon in the many\nobreakdash-body theory.

\medskip
\noindent\textbf{Organization of the paper.}
	The outline of this article is as follows. Section~\ref{sec:NLS} is devoted to the existence of \NLS ground states and its blow-up behavior. The condensation and collapse of the Hartree ground states are studied in Section ~\ref{sec:Hartree}. The mean\nobreakdash-field approximation and the collapse of the system is given in Section~\ref{sec:many_body}. The appendices contain various technical results and proofs.

\medskip
\noindent\textbf{Acknowledgments.} We thank P.~T.~Nam for drawing our attention to this problem. D.-T. Nguyen was supported through the Knut and Alice Wallenberg Foundation and the European Research Council (ERC) under the European Union’s Horizon 2020 Research and Innovation Programme (Grant agreement CORFRONMAT No 758620). J.~Ricaud acknowledges financial support from the French Agence Nationale de la Recherche (ANR) under Grant No. ANR-19-CE46-0007 (project ICCI). The authors thank the referee for the suggestions that helped improving the manuscript.

\section{Existence and collapse of the \NLS ground states}\label{sec:NLS}
	The goal of this section is to establish the existence of the ground states for the \NLS problem~\eqref{energy:NLS} and investigate its blow-up profile as well as the first order expansion of the ground state energy.

\subsection{Proof of Theorem~\ref{thm:existence_NLS}}\label{Section_Proof_existence_NLS}
Recall~\eqref{NLS_energy_upperbound} that
\[
	E_{a,b}^{\NLS} \leq \ell^2\frac{\mathfrak{b}-b}{12} + \ell\frac{a}{\pi} + \ell^{-s}\int_{\R}|x|^s|Q_0(x)|^2\dix\,, \quad \forall\, \ell>0\,.
\]

For~\emph{\ref{NLS_GS_inexistence}} ---($b>\mathfrak{b}$ and $a\in\R$) or ($b=\mathfrak{b}$ and $a<0$)---, taking $\ell \to \infty$ gives $E_{a,b}^{\NLS} = -\infty$.

For~\emph{\ref{NLS_GS_critical}} ---$b=\mathfrak{b}$ and $a=0$---, we have $E_{0,\mathfrak{b}}^{\NLS} \geq0$ by~\eqref{ineq:GNS} then we obtain $E_{0,\mathfrak{b}}^{\NLS} = 0$ the same way. Suppose now that there exists a ground state $u$ for $E_{0,\mathfrak{b}}^{\NLS}$. Then, \eqref{ineq:GNS} implies that $\norm{|\cdot|^{s/2}u}_{2}=0$, which contradicts $\norm{u}_{2}=1$.

The rest of the proof is dedicated to the most demanding case~\emph{\ref{NLS_GS_existence}} ---($0<b<\mathfrak{b}$ and $a\in \R$) or ($b=\mathfrak{b}$ and $a>0$)---, for which we prove the existence of a ground state. Let $\{u_n\}_n \subset H^1(\R)$ be a minimizing sequence for $E_{a,b}^{\NLS}$, i.e.,
\[
	E_{a,b}^{\NLS} = \lim_{n\to\infty}\mathcal{E}_{a,b}^{\NLS}(u_n) \quad \text{ with } \quad \norm{u_n}_{2}=1\,.
\]
By~\eqref{ineq:GNS}, we have
\[
	\mathcal{E}_{a,b}^{\NLS}(u_n) \geq \left(\frac{\mathfrak{b}-b}{\mathfrak{b}}\right) \norm{u_n'}_{2}^2 + \frac{a}{2} \norm{u_n}_{4}^4 + \norm{|\cdot|^{\frac{s}{2}}u_n}_{2}^2\,.
\]
Furthermore, using that~$\norm{f}_{\infty}^2\leq2\norm{f}_{2}\norm{f'}_{2}$, hence~$\norm{f}_{4}^4\leq 2 \norm{f}_{2}^3 \norm{f'}_{2}$, for any~$f\in H^1(\R)$, we have in the case ~$b<\mathfrak{b}$ (and $a\in\R$) that~$\{u_n\}_n$ is uniformly bounded in~$H^1(\R)$, thence that $\{|\cdot|^{s/2}u_n\}_n$ is uniformly bounded in~$L^2(\R)$. In the case~$b=\mathfrak{b}$ and~$a>0$, the latter immediately follows from the above lower bound on $\mathcal{E}_{a,b}^{\NLS}(u_n)$, and we are left with proving the $H^1(\R)$-boundedness\footnote{Note that the arguments presented here to prove $H^1(\R)$-boundedness also apply to the case~$0<b<\mathfrak{b}$ and~$a>0$.}. We assume on the contrary that $\norm{u_n'}_{2} \to \infty$ as~$n\to\infty$. Define~$v_n = \eps_n^{1/2} u_n(\eps_n \cdot)$ with~$\eps_n := \norm{u_n'}_{2}^{-1}$. Then, $\norm{v_n}_{2} = 1 =\norm{v_n'}_{2}$ and
\[
	\mathcal{E}_{a,b}^{\NLS}(u_n) = \eps_n^{-2}\mathcal{F}^{\NLS}_{b}(v_n) + \eps_n^{-1}\frac{a}{2} \norm{v_n}_{4}^4 + \eps_n^s \norm{|\cdot|^\frac{s}{2}v_n}_{2}^2\,.
\]
Here we used the shorthand notation
\begin{equation}\label{NLS_Def_mathcal_F}
	\mathcal{F}^{\NLS}_{b}(v_n) := \norm{v_n'}_{2}^2 - \frac{b}{6}\norm{v_n}_{6}^6
\end{equation}
Since~$\eps_n \to 0$ as~$n\to\infty$, the above implies that 
\begin{equation}\label{NLS_gound_state_existence_boundedness}
	\lim_{n\to\infty}\mathcal{F}^{\NLS}_{b}(v_n) = \lim_{n\to\infty}\norm{v_n}_{4} = 0.
\end{equation}
On the other hand, since $\{v_n\}_n$ is bounded uniformly in $H^1(\R)$ we have, up to a translation and a subsequence, $v_n\wto v$ weakly in $H^1(\R)$. We claim that $\norm{v}_{2} = 1$. Indeed, we have
\[
	0 \leq \norm{v}_{2} \leq \liminf_{n\to\infty} \norm{v_n}_{2} = 1\,.
\]
If $\norm{v}_{2} = 0$ then $v_n \to 0$ strongly in $L^r(\R)$ for any $2\leq r < \infty$ (see e.g.,~\cite[Lemma~9]{LenLew-11}). This contradicts~\eqref{NLS_gound_state_existence_boundedness} and the fact that $\norm{v_n'}_{2}=1$. If $0 < \norm{v}_{2} < 1$ then, by the weak convergence $v_n\wto v$ in $H^1(\R)$, we have the energy decomposition
\[
	0 = \lim_{n\to\infty}\mathcal{F}^{\NLS}_{b}(v_n) = \mathcal{F}^{\NLS}_{b}(v) + \lim_{n\to\infty}\mathcal{F}^{\NLS}_{b}(v_n-v)\,.
\]
This, however, is not possible because
\[
	\mathcal{F}^{\NLS}_{b}(v_n-v) \geq 0 \quad \text{and} \quad \mathcal{F}^{\NLS}_{b}(v) \geq \left(\frac{1}{\norm{v}_{2}^{4}}-1\right)\frac{\mathfrak{b}}{6}\norm{v}_{6}^{6} > 0\,,
\]
by~\eqref{ineq:GNS}. Therefore, we must have $\norm{v}_{2} = 1$. By the Brezis--Lieb lemma, we have $v_n\to v$ strongly in $L^2(\R)$. In fact, $v_n\to v$ strongly in $L^r(\R)$ for any $2\leq r < \infty$, because of the $H^1(\R)$ boundedness. In particular, $\norm{v_n}_{4} \to \norm{v}_{4}$ and $v$ is an optimizer for~\eqref{ineq:GNS} since
\[
	0 = \lim_{n\to\infty}\mathcal{F}^{\NLS}_{b}(v_n) \geq \mathcal{F}^{\NLS}_{\mathfrak{b}}(v) \geq 0\,,
\]
contradicting $\norm{v_n}_{4} \to 0$ as $n\to\infty$.

We have proved that $\{u_n\}_n$ is uniformly bounded in $H^1(\R)$ whence either $b=\mathfrak{b}$ and $a>0$, or $0<b<\mathfrak{b}$ and $a\in \R$. Then, up to a translation and a subsequence, $u_n\wto u$ weakly in $H^1(\R)$ and $u_n\to u$ strongly in $L^r_{\textrm{loc}}(\R)$ for any $2\leq r\leq+\infty$. This strong $L^r_{\textrm{loc}}(\R)$-convergence combined with the uniform boundedness of $\norm{|\cdot|^{s/2}u_n}_{2}$ gives that, up to a subsequence, $u_n\to u$ strongly in $L^r(\R)$ for $2\leq r<+\infty$ (see e.g.,~\cite{AdaFou-03}). In particular, $\norm{u}_{2}=1$. By the Brezis--Lieb lemma, we have $u_n\to u$ strongly in $L^2(\R)$. In fact, $u_n\to u$ strongly in $L^r(\R)$ for any $2\leq r < \infty$, because of the $H^1(\R)$ boundedness. Consequently, by weak lower semicontinuity we have
\[
	E_{a,b}^{\NLS}= \lim_{n\to\infty}\mathcal{E}_{a,b}^{\NLS}(u_n) \geq \mathcal{E}_{a,b}^{\NLS}(u) \geq E_{a,b}^{\NLS}\,.
\]
Hence, $u$ is a ground state of~$E_{a,b}^{\NLS}$.
\qed

\subsection{Proof of Theorem~\ref{thm:collapse_NLS}}
The proof for approximate ground states being the same, up to very few changes, as the one for ground states, we only write the latter.

For shortness, we define $B_n:=\mathfrak{b}-b_n$ and $A_n:=a_n/\pi$.
We start by inserting $\ell=\ell_n^{-1}$, defined in Theorem~\ref{thm:collapse_NLS}, into~\eqref{NLS_energy_upperbound}. Treating separately the cases $\zeta\neq0$ and $\zeta\neq6$, it yields the upper bound matching~\eqref{blowup_NLS_energy} in both cases:
\begin{equation}\label{NLS_upperbound}
	E_{a_n,b_n}^{\NLS} \leq \ell_n^s \left( \ell_n^{-(s+2)} \frac{B_n}{12} + A_n \ell_n^{-(s+1)} + \frac{\mathcal{Q}_s}{s} \right) = \left( \frac{s+1}{s} - \frac{\zeta}{12} + o(1) \right) \mathcal{Q}_s \ell_n^s\,.
\end{equation}

Let $u_n$ be a ground state of~$E_{a_n,b_n}^{\NLS}$ and $\tilde{u}_n:= \ell_n^{1/2} u_n(\ell_n\cdot)$ for any $n$. Then, $\norm{\tilde{u}_n}_{2} = \norm{u_n}_{2}=1$ and
\begin{equation}\label{NLS_approximate}
	E_{a_n,b_n}^{\NLS} = \ell_n^s \left[ \ell_n^{-(s+2)} \mathcal{F}^{\NLS}_{b_n}(\tilde{u}_n) + A_n \ell_n^{-(s+1)} \frac{\pi}{2} \norm{\tilde{u}_n}_{4}^4 + \norm{ |\cdot|^{\frac{s}{2}}\tilde{u}_n }_2 \right]\,.
\end{equation}
We now prove the claimed convergence~\eqref{blowup_NLS_ground_state}, which immediately implies~\eqref{blowup_NLS_energy}. We distinguish the two cases for $\zeta\geq0$: $\zeta\neq0$ and $\zeta\neq6$.

\medskip
\noindent\textbf{Case $\zeta\neq0$.} Recall that~$B_n>0$ (see Remark~\ref{Rmk_thm_collapse_nls}). Combining~\eqref{NLS_upperbound} with~\eqref{NLS_approximate} and the definition of $\ell_n$ depending on $B_n$ yields
\begin{equation}\label{NLS_control_nrj_for_gamma_not_0}
	\frac{s+1}{s} - \frac{\zeta}{12} + o(1) \geq \frac{\zeta}{B_n} \mathcal{F}^{\NLS}_{b_n}(\tilde{u}_n) + \left( \frac{6-\zeta}{6} \cdot \frac{\pi}{2} + o(1) \right) \norm{\tilde{u}_n}_{4}^4 + \frac{1}{\mathcal{Q}_s} \norm{ |\cdot|^{\frac{s}{2}}|\tilde{u}_n }_{2}^2\,.
\end{equation}
Using now~\eqref{ineq:GNS}, we obtain
\begin{equation}\label{NLS_control_nrj_L2_for_gamma_not_0}
	\frac{s+1}{s} - \frac{\zeta}{12} + o(1) \geq \frac{\zeta}{\mathfrak{b}} \norm{\tilde{u}_n'}_{2}^2 + \left( \frac{6-\zeta}{6} \cdot \frac{\pi}{2} + o(1) \right) \norm{\tilde{u}_n}_{4}^4 + \mathcal{Q}_s^{-1} \norm{|\cdot|^{\frac{s}{2}}\tilde{u}_n}_{2}^2\,,
\end{equation}
and using moreover~$\norm{\tilde{u}_n}_{4}^4\leq 2 \norm{\tilde{u}_n'}_{2}$, as in the proof of Theorem~\ref{thm:existence_NLS}, gives
\[
	\frac{s+1}{s} - \frac{\zeta}{12} + o(1) \geq \frac{\zeta}{\mathfrak{b}} \norm{\tilde{u}_n'}_{2}^2 - \left| \frac{6-\zeta}{6} \cdot \frac{\pi}{2} + o(1) \right| \norm{\tilde{u}_n'}_{2} + \mathcal{Q}_s^{-1} \norm{|\cdot|^{\frac{s}{2}}\tilde{u}_n}_{2}^2\,.
\]
This implies that~$\{\tilde{u}_n\}_n$ and~$\{ |\cdot|^{s/2}\tilde{u}_n \}_n$ are uniformly bounded in~$H^1(\R)$ and in~$L^2(\R)$, respectively. Notice that if~$0<\zeta<6$, this could have been deduced directly from~\eqref{NLS_control_nrj_for_gamma_not_0} and~\eqref{ineq:GNS}, using that the term in font of~$\norm{\tilde{u}_n}_{4}^4$ in~\eqref{NLS_control_nrj_for_gamma_not_0} is nonnegative for~$n$ large enough.

Thus, up to a subsequence,~$\tilde{u}_n \to u$ weakly in~$H^1(\R)$ and strongly in~$L^r(\R)$ for~$2\leq r < \infty$. In particular, $\norm{u}_{2}=1$ and, multiplying~\eqref{NLS_control_nrj_for_gamma_not_0} by~$B_n/\zeta$ then passing to the limit in~$n$, we obtain
\begin{equation}\label{NLS_limit_fct_is_optimizer_GNS}
	 \norm{u'}_{2}^2 \leq \lim\limits_{n\to+\infty} \norm{\tilde{u}_n'}_{2}^2 \leq \lim\limits_{n\to+\infty} \left(\frac{b_n}{6} \norm{\tilde{u}_n}_{6}^6 + O(B_n) \right) = \frac{\mathfrak{b}}{6} \norm{u}_{6}^6 \leq \norm{u'}_{2}^2,
\end{equation}
where the first inequality is Fatou's lemma and the last one is~\eqref{ineq:GNS}. On one hand, by the uniqueness (up to translations, multiplication by a complex factor, and scaling) of~$Q_0$, this implies that~$u$ is an optimizer of~\eqref{ineq:GNS} thence that, after a suitable rescaling, $u(x)=t^{1/2}Q_0(tx+x_0)$ for some $t>0$ and~$x_0\in\R$. On another hand, it implies that the convergence in~$H^1(\R)$ is actually strong since~$\norm{\tilde{u}_n'}_{2} \to \norm{u'}_{2}$ and~$\tilde{u}_n'\wto u'$ in~$L^2(\R)$.

We now show that~$u \equiv Q_0$, which proves (but only up to a subsequence) both~\eqref{blowup_NLS_ground_state} and~\eqref{blowup_NLS_energy}. That is, we show that~$t= 1$ and~$x_0=0$. Indeed, taking the limit~$n\to\infty$ in~\eqref{NLS_control_nrj_L2_for_gamma_not_0}, we get
\begin{align}\label{NLS_control_nrj_L4_for_gamma_not_0}
	\frac{s+1}{s} - \frac{\zeta}{12}
		\geq{} & \frac{\zeta}{\mathfrak{b}} \norm{u'}_{2}^2 + \frac{6-\zeta}{6} \cdot \frac{\pi}{2}\norm{u}_{4}^4 + \mathcal{Q}_s^{-1} \norm{|\cdot|^{\frac{s}{2}}u}_{2}^2 \nonumber \\
		={} & t^2 \frac{\zeta}{\mathfrak{b}} \norm{Q_0'}_{2}^2 + t \frac{6-\zeta}{6} \cdot \frac{\pi}{2}\norm{Q_0}_{4}^4 + t^{-s} \mathcal{Q}_s^{-1} \norm{|\cdot-x_0|^{\frac{s}{2}}Q_0}_{2}^2 \nonumber \\
		\geq{} & \frac{\zeta}{12} t^2 + \frac{6-\zeta}{6} t + \frac{t^{-s}}{s} =:g(t) \geq g(1) = \frac{s+1}{s} - \frac{\zeta}{12}\,.
\end{align}
For the one-to-last inequality, we used that~$Q_0$ is \emph{strictly symmetric-decreasing} and~$|\cdot|^s$ is \emph{strictly symmetric-increasing} to obtain (see Appendix~\ref{app:techn_NLS} for definitions and proof)
\begin{equation}\label{ineq_on_strict_symm_decreas_and_increasing_functions}
	\norm{|\cdot-x_0|^{\frac{s}{2}}Q_0}_{2}^2>\norm{|\cdot|^{\frac{s}{2}}Q_0}_{2}^2, \quad \forall\, x_0\neq0\,,
\end{equation}
that~$\norm{Q_0'}_{2}^2 = \pi^2/8 = \mathfrak{b}/12$ , and that~$\norm{Q_0}_{4}^4 = 2/\pi$. For the last inequality, we noticed that $g'(1)=0$ where~$g'(z) = \zeta z/6 + (6-\zeta)/6 - z^{-(s+1)}$ is strictly increasing on~$(0,+\infty)$, from~$-\infty$ to $+\infty$. Hence, $z=1$ is its unique zero, at which $g$ attains therefore its minimum. Consequently, equality holds in~\eqref{NLS_control_nrj_L4_for_gamma_not_0}, which implies~$t=1$, as well as in~\eqref{ineq_on_strict_symm_decreas_and_increasing_functions}, which in turn gives~$x_0=0$.

In order to complete the proof, we need to prove that the convergence holds for the whole sequence~$\{u_n\}_n$ and not only for a subsequence. This is a consequence of the \emph{uniqueness} of~$Q_0$ and of the fact that~\eqref{ineq:GNS} and~\eqref{NLS_control_nrj_for_gamma_not_0} hold for the whole sequence. Indeed, assuming the contrary, there exists a subsequence that does not converge in~$H^1(\R)$ to~$Q_0$. However, this subsequence is also bounded in~$H^1(\R)$, because the whole sequence is, and we can extract from it a subsubsequence that converges weakly in~$H^1(\R)$ to some~$v\not\equiv Q_0$. Now, since the reasoning leading to $u\equiv Q_0$ (after $u$ was obtained) relied only on inequalities that hold for the whole sequence~$\{u_n\}$, it also yields~$v\equiv Q_0$, which is a contradiction. The proof in the case~$\zeta\neq0$ is therefore complete.

\medskip
\noindent\textbf{Case $\zeta\neq6$.}
Recall that~$A_n/(6-\zeta)>0$ (see Remark~\ref{Rmk_thm_collapse_nls}). Combining~\eqref{NLS_upperbound} with~\eqref{NLS_approximate} and the definition of $\ell_n$ depending on $A_n$ yields
\begin{equation}\label{NLS_control_nrj_for_gamma_not_6}
	\frac{s+1}{s} - \frac{\zeta}{12} + o(1) \geq \frac{6-\zeta}{6 A_n\ell_n} \left[\mathcal{F}^{\NLS}_{b_n}(\tilde{u}_n) + \frac{6-\zeta }{6} \cdot \frac{\pi}{2} \norm{\tilde{u}_n}_{4}^4 + \frac{1}{\mathcal{Q}_s} \norm{ |\cdot|^{\frac{s}{2}}\tilde{u}_n }_{2}^2\right]\,.
\end{equation}
Using now~\eqref{ineq:GNS}, we obtain
\begin{equation}\label{NLS_control_nrj_L2_for_gamma_not_6}
	\frac{s+1}{s} - \frac{\zeta}{12} + o(1) \geq \frac{ \zeta + o(1) }{\mathfrak{b}} \norm{\tilde{u}_n'}_{2}^2 + \frac{6-\zeta}{6} \cdot \frac{\pi}{2} \norm{\tilde{u}_n}_{4}^4 + \mathcal{Q}_s^{-1} \norm{|\cdot|^{\frac{s}{2}}\tilde{u}_n}_{2}^2\,,
\end{equation}
with $\zeta + o(1) \geq0$ even if $\zeta=0$, as a consequence of $B_n\geq0$ and $A_n/(6-\zeta)>0$.

\begin{itemize}[leftmargin=1em]
	\item If~$\zeta>0$, using $\norm{\tilde{u}_n}_{4}^4\leq 2 \norm{\tilde{u}_n'}_{2}$ again, \eqref{NLS_control_nrj_L2_for_gamma_not_6} gives the uniform boundedness in~$H^1(\R)$ and, consequently, that~$\{ |\cdot|^{s/2}\tilde{u}_n \}_n$ is uniformly bounded in~$L^2(\R)$.
	\item If $\zeta=0$, we cannot recover the $H^1(\R)$-boundedness this way, but we have
		\[
			\frac{s+1}{s} - \frac{\zeta}{12} + o(1) \geq \frac{ o(1) }{\mathfrak{b}} \norm{\tilde{u}_n'}_{2}^2 + \frac{\pi}{2} \norm{\tilde{u}_n}_{4}^4 + \mathcal{Q}_s^{-1} \norm{|\cdot|^{\frac{s}{2}}\tilde{u}_n}_{2}^2
		\]
		where the factor in front of $\norm{\tilde{u}_n'}_{2}^2$ is nonnegative ---see the comment just below~\eqref{NLS_control_nrj_L2_for_gamma_not_6}. This implies that~$\{ |\cdot|^{s/2}\tilde{u}_n\}_n$ and~$\{\tilde{u}_n\}_n$ are uniformly bounded in~$L^2(\R)$ and in~$L^r(\R)$ for~$2\leq r \leq 4$, respectively. Assume now that $\norm{\tilde{u}_n'}_{2}$ is not uniformly bounded. Then, there exists a subsequence (still denoted the same) such that $\eps_n := \norm{\tilde{u}_n'}_{2}^{-1} \to0$. Let $v_n:=\eps_n^{1/2} \tilde{u}_n(\eps_n\cdot)$. It satisfies $\norm{v_n}_{2}=1=\norm{v_n'}_{2}$, $\norm{v_n}_{4}=\eps_n^{1/4} \norm{\tilde{u}_n}_{4}\to0$ when $n\to+\infty$ (since $\norm{\tilde{u}_n}_{4}$ is uniformly bounded), and
		\[
			E_{a_n,b_n}^{\NLS} = \eps_n^{-2}\ell_n^{-2} \left[\mathcal{F}^{\NLS}_{b_n}(\tilde{v}_n) + \eps_n\ell_n A_n \frac{\pi}{2} \norm{v_n}_{4}^4 + \eps_n^{s+2}\ell_n^{s+2} \norm{|\cdot|^{\frac{s}{2}}v_n}_{2}^2 \right]\,.
		\]
		This, combined with~\eqref{NLS_upperbound} and since each of the three term on the right hand side are nonnegative (using again~\eqref{ineq:GNS} and $b_n\leq\mathfrak{b}$ for the first term, and $A_n=6A_n/(6-\zeta)>0$ for the second), gives
		\[
			0\leq \mathcal{F}^{\NLS}_{b_n}(\tilde{v}_n) \leq \left( \frac{s+1}{s} - \frac{\zeta}{12} \right) \mathcal{Q}_s \ell_n^{s+2}\eps_n^2 (1+o(1)) \underset{n\to0}{\longrightarrow} 0\,,
		\]
		which in turn implies, by the same arguments as in the proof of Theorem~\ref{thm:existence_NLS}, that $v_n\to v$ strongly in~$L^r(\R)$ for $2 \leq r < \infty$, where $v$ is an optimizer of~\eqref{ineq:GNS}. In particular, $\norm{v_n}_{4}\to\norm{v}_{4}\neq0$, contradicting $\norm{v_n}_{4}\to0$. This concludes the proof that $\{\tilde{u}_n\}_n$ is uniformly bounded in~$H^1(\R)$.
	
		Therefore, for all $0\leq\zeta\neq6$, $\{\tilde{u}_n\}_n$ and~$\{ |\cdot|^{s/2}\tilde{u}_n \}_n$ are uniformly bounded in~$H^1(\R)$ and in~$L^2(\R)$, respectively. From here, we omit the detail of the proof because it follows strictly the one in the case $\zeta\neq0$. The only differences are that we multiply~\eqref{NLS_control_nrj_for_gamma_not_6} by~$A_n\ell_n/(6-\zeta)$ instead of~\eqref{NLS_control_nrj_for_gamma_not_0} by~$B_n/\zeta$ and that the $O(B_n)$ in~\eqref{NLS_limit_fct_is_optimizer_GNS} becomes a $O(A_n\ell_n/(6-\zeta))$. This concludes the proof of Theorem~\ref{thm:collapse_NLS}.
\end{itemize}

\section{Condensation and collapse of the Hartree ground states: proof of Theorem~\ref{thm:collapse_Hartree}}\label{sec:Hartree}
	The goal of this section is to prove Theorem~\ref{thm:collapse_Hartree}: to establish the condensation of the Hartree ground states in the limit $N\to\infty$ and to investigate its blow-up profile as well as the first order expansion of the ground state energy. To do so, we prove first an elementary lemma which will be heavily used in the proof of Theorem~\ref{thm:collapse_Hartree} but will also be useful in the rest of this work.
\begin{lemma}\label{Limit_and_rate_Hartree_interaction_energy}
	Let $u\in H^1(\R)$. Then,
	\begin{equation}\label{ineq_cv:Hartree_to_NLS_2body}
		0 \leq \int_{\R}|u(x)|^4\dix - \iint_{\R^2} U_N\left(x-y\right)|u(x)|^2|u(y)|^2\dix\diy \leq \norm{u}_{H^1}^4 o(1)_{N\to+\infty}
	\end{equation}
	and
	\begin{equation}\label{ineq_cv:Hartree_to_NLS_3body}
		0 \leq \int_{\R}|u(x)|^6\dix - \iiint_{\R^3} W_N\left(x-y,x-z\right)|u(x)|^2|u(y)|^2|u(z)|^2\dix\diy\diz \leq \norm{u}_{H^1}^6 o(1)_{N\to+\infty} \,,
	\end{equation}
	with the $o(1)$'s independent of $u$.
	
	Furthermore, if $x U(x) \in L^1(\R)$ and $xW(x,y) \in L^1(\R^2)$ then
	\begin{equation}\label{cv_rate:Hartree_to_NLS_2body}
		0 \leq \int_{\R}|u(x)|^4\dix - \iint_{\R^2} U_N\left(x-y\right)|u(x)|^2|u(y)|^2\dix\diy \leq 2 N^{-\alpha} \norm{x U(x)}_{1} \norm{u}_{6}^3 \normSM{u'}_{2}
	\end{equation}
	and
	\begin{align}\label{cv_rate:Hartree_to_NLS_3body}
		0 \leq{}& \int_{\R}|u(x)|^6 \dix - \iiint_{\R^3} W_N(x-y,x-z) |u(x)|^2 |u(y)|^2 |u(z)|^2 \dix \diy \diz \nonumber \\
		\leq{}& 4 N^{-\beta} \norm{x W(x,y)}_{1} \norm{u}_{10}^5 \normSM{u'}_{2}\,.
	\end{align}
\end{lemma}
We emphasize that~\eqref{ineq_cv:Hartree_to_NLS_2body} and~\eqref{ineq_cv:Hartree_to_NLS_3body} are not a consequence of~\eqref{cv_rate:Hartree_to_NLS_2body} and~\eqref{cv_rate:Hartree_to_NLS_3body}, as the former hold even if $\norm{x U(x)}_{1}$ and $\norm{x W(x,y)}_{1}$ are infinite. The point of~\eqref{ineq_cv:Hartree_to_NLS_2body} and~\eqref{ineq_cv:Hartree_to_NLS_3body} is precisely to hold in full generality, but without any explicit rate of convergence, while the point of~\eqref{cv_rate:Hartree_to_NLS_2body} and~\eqref{cv_rate:Hartree_to_NLS_3body} is precisely the rate of convergence, which requires these extra assumptions. Notice that, by the symmetry of $W$ in~\eqref{condition:potential_three_body_symmetry}, $\norm{x W(x,y)}_{1} < +\infty$ is equivalent to $\norm{(x,y)W(x,y)}_{1} < +\infty$.
\begin{proof}
	The nonnegativities in~\eqref{ineq_cv:Hartree_to_NLS_2body} and in~\eqref{ineq_cv:Hartree_to_NLS_3body} are a consequence of the Cauchy--Schwarz inequality on reals, combined respectively with $U_N\geq0$ and with $W_N\geq0$,
	\begin{align*}
		\iint_{\R^2} U_N(x-y)|u(x)|^2|u(y)|^2\dix\diy \leq \iint_{\R^2} U_N(x-y)\frac{|u(x)|^4+|u(y)|^4}{2}\dix\diy = \int_{\R}|u(x)|^4\dix
		\intertext{and}
		\begin{multlined}[t][\textwidth]
			\iiint_{\R^3} W_N(x-y,x-z)|u(x)|^2|u(y)|^2|u(z)|^2\dix\diy\diz \\
			\leq \iiint_{\R^3} W_N(x-y,x-z)\frac{|u(x)|^6+|u(y)|^6+|u(z)|^6}{3}\dix\diy\diz = \int_{\R}|u(x)|^6\dix,
		\end{multlined}
	\end{align*}
	where for the equalities we used $\norm{U_N}_{1}=\norm{U}_{1}=1$ for the former and the symmetries of~$W_N$ together with $\norm{W_N}_{1}=\norm{W}_{1}=1$ for the latter.
	
	The limit in~\eqref{ineq_cv:Hartree_to_NLS_2body} is the 1D analog of the limit in~\cite[Lemma~7]{LewNamRou-17} for the 2-body interaction in~2D.
	
	For the limit in~\eqref{ineq_cv:Hartree_to_NLS_3body}, we adapt the arguments therein to our 3-body interaction in 1D case. By the assumption $\norm{W}_{1}=1$, the changes of variables $y\to x - N^{-\beta}y$ and $z\to x - N^{-\beta}z$ and $W(y,z)=W(z,y)$, we have
	\begin{multline}\label{limit_Hartree_energy_bounds_to_NLS_first_step}
		\int_{\R}|u(x)|^6\dix - \iiint_{\R^3} W_N\left( x-y,x-z \right)|u(x)|^2|u(y)|^2|u(z)|^2\dix\diy\diz \\
		= \iiint_{\R^3} W(y,z) \left( |u(x)|^2 + |u(x-N^{-\beta}y)|^2 \right) |u(x)|^2 \left( |u(x)|^2 - |u(x-N^{-\beta}z)|^2 \right)\diz\diy\dix\,.
	\end{multline}
	As in~\cite{LewNamRou-17}, we split the integral over $z$ into two parts. On $|z|>L$, we use
	\[
		2|u(x)|^2|u(x-N^{-\beta}z)|^2\leq|u(x)|^4 + |u(x-N^{-\beta}z)|^4\,,
	\]
	then H{\"o}lder inequality (on the integral on $x$). On $|z|\leq L$, we use
	\begin{equation}\label{limit_Hartree_energy_bounds_to_NLS_integral_rewriting}
		|u(x)|^2 - |u(x-N^{-\beta}z)|^2 = N^{-\beta}z\int_0^1 \left( |u|^2 \right)'(x-tN^{-\beta}z) \dit\,,
	\end{equation}
	then the diamagnetic inequality $||u|'| \leq |u'|$, hence $|(|u|^2)'|=2|u|||u|'|\leq 2|u||u'|$. We obtain
	\begin{align*}
		\int_{\R}|u(x)|^6&\dix - \iiint_{\R^3} W_N\left( x-y,x-z \right)|u(x)|^2|u(y)|^2|u(z)|^2\dix\diy\diz \\
		\leq{}&\begin{multlined}[t]
			\frac12\int_{|\cdot|>L}\int_\R W(y,z)\int_{\R}\left( |u(x)|^2 + |u(x-N^{-\beta}y)|^2 \right) \left( 3|u(x)|^4 + |u(x-N^{-\beta}z)|^4 \right)\dix\diy\diz \\
			+ \frac{2L}{N^\beta}\int_0^1 \int_{|\cdot|\leq L} \int_{\R}W(y,z)\int_{\R} \left(|u(x)|^2 + |u(x-N^{-\beta}y)|^2 \right) \\
			\times |u(x)|^2 |u(x-tN^{-\beta})| |u'(x-tN^{-\beta}z)| \dix\diz\diy\dit 
		\end{multlined} \\
		\leq{}& 4\int_{|\cdot|>L}\int_\R W(y,z)\norm{u}_{6}^6\diy\diz + 4N^{-\beta}L\int_0^1 \int_{|\cdot|\leq L} \int_{\R}W(y,z) \norm{u}_{\infty}^4 \norm{u}_{2} \normSM{u'}_{2}\diz\diy\dit \\
		\leq{}& 4C\norm{u}_{H^1}^6\left( \int_{|\cdot|>L}\int_\R W(y,z)\diy\diz+4N^{-\beta}L \norm{W}_{1} \right)\underset{\substack{L\to\infty\\LN^{-\beta}\to0}}{\longrightarrow}0\,,
	\end{align*}
	where we used Sobolev inequalities for the last line.
	
	The convergent rate~\eqref{cv_rate:Hartree_to_NLS_2body} is~\cite[Lemma~4.1]{LewNamRou-16}. See the details of its proof for our exact formulation (recalling that $\norm{\nabla |u|}_{2}\leq\norm{\nabla u}_{2}$). For the convergent rate~\eqref{cv_rate:Hartree_to_NLS_3body}, we adapt the arguments therein to our 3-body interaction in 1D case. We start by applying~\eqref{limit_Hartree_energy_bounds_to_NLS_integral_rewriting} to~\eqref{limit_Hartree_energy_bounds_to_NLS_first_step} without splitting the integral on $z$, which leads to
	\[
	\begin{multlined}
		\iiint_{\R^3} W_N\left( x-y,x-z \right)|u(x)|^2|u(y)|^2|u(z)|^2\dix\diy\diz -\int_{\R}|u(x)|^6\dix \\
		= -N^{-\beta}\int_0^1 \iint_{\R^2} z W(y,z)\int_{\R}|u(x)|^2 \left( |u(x-N^{-\beta}y)|^2 + |u(x)|^2 \right)\left( |u|^2 \right)'(x-tN^{-\beta}z) \dix\diz\diy\dit\,.
	\end{multlined}
	\]
	By the diamagnetic inequality $|(|u|^2)'| \leq 2|u||u'|$ and H\"older's inequality we have
	\[
		\int_{\R}|u(x)|^2|u(x-N^{-\beta}\tilde{z})|^2(|u|^2)'(x-tN^{-\beta}\tilde{y}) \dix \leq 2\norm{u}_{10}^5 \normSM{u'}_{2}\,.
	\]
	Similarly, we have
	\[
		\int_{\R}|u(x)|^{4}(|u|^2)'(x-sN^{-\beta}\tilde{y}) \dix \leq 2\norm{u}_{10}^5 \normSM{u'}_{2}\,.
	\]
	Thus, by triangle ineqality we obtain~\eqref{cv_rate:Hartree_to_NLS_3body}, concluding the proof of Lemma~\ref{Limit_and_rate_Hartree_interaction_energy}.
\end{proof}

With this lemma at hand, we can now prove Theorem~\ref{thm:collapse_Hartree}.
\begin{proof}[Proof of Theorem~\ref{thm:collapse_Hartree}]
	The proof for approximate ground states being the same, up to very few changes, as the one for ground states, we only write the latter.
	
	We first consider the limit $N\to\infty$ and prove~\eqref{cv:Hartree_to_NLS_ground_state} and~\eqref{cv:Hartree_to_NLS_energy}. When $a \leq 0$ and $0 < b < \mathfrak{b}$, it is immediate from the nonnegativity in~\eqref{ineq_cv:Hartree_to_NLS_2body} and~\eqref{ineq_cv:Hartree_to_NLS_3body} that $E_{a, b}^{\mathrm{H}, N} \geq E_{a,b}^{\NLS}$.
	On the other hand, choosing a ground state for $E_{a,b}^{\NLS}$ as a trial state for $E_{a, b}^{\mathrm{H}, N}$, the variational principle together with~\eqref{ineq_cv:Hartree_to_NLS_2body} and~\eqref{ineq_cv:Hartree_to_NLS_3body} imply that
	\begin{equation}\label{upperbound:Hartree_NLS}
		o(1) + E_{a,b}^{\NLS} \geq E_{a, b}^{\mathrm{H}, N}\,.
	\end{equation}
	Next, we consider the case where $a > 0$ and $0 < b \leq \mathfrak{b}$. We use again Lemma~\ref{Limit_and_rate_Hartree_interaction_energy} to obtain~\eqref{upperbound:Hartree_NLS}. For the energy lower bound, we process as follows. Let $u_N$ be a ground state for $E_{a, b}^{\mathrm{H}, N}$. On one hand, when $0< b < \mathfrak{b}$, it follows from~\eqref{upperbound:Hartree_NLS}, the nonnegativity in~\eqref{ineq_cv:Hartree_to_NLS_3body}, and~\eqref{ineq:GNS} that $\normSM{u_N'}_{2}$ is bounded uniformly in $N$. Then, the variational principle together with Lemma~\ref{Limit_and_rate_Hartree_interaction_energy} imply that
	\[
		E_{a, b}^{\mathrm{H}, N} \geq E_{a,b}^{\NLS} + o(1).
	\]
	On another hand, when $b = \mathfrak{b}$, we prove that $\normSM{u_N'}_{2}$ is bounded uniformly in $N$ under the assumption $\alpha>\beta>0$. We assume on the contrary that $\norm{u_N'}_{2} \to \infty$ as $N\to\infty$. Define $v_N = \eps_N^{1/2} u_N(\eps_N \cdot)$ with $\eps_N := \norm{u_N'}_{2}^{-1}$. Then, we have $\norm{v_N}_{2} = 1 =\norm{v_N'}_{2}$ and
	\[
		\mathcal{E}_{a, \mathfrak{b}}^{\mathrm{H}, N}(u_N) \geq \eps_N^{-2}\mathcal{F}^{\NLS}_{\mathfrak{b}}(v_N) + \frac{a}{2}\iint_{\R^2} U_N(\eps_N(x-y))|v_N(x)|^2|v_N(y)|^2\dix\diy.
	\]
	Here, we have used~\eqref{ineq_cv:Hartree_to_NLS_3body} on $v_N$ and the nonnegativity of the external potential. Since $\eps_N \to 0$ as $N\to\infty$, the above and \eqref{upperbound:Hartree_NLS} imply that
	\[
		\mathcal{F}^{\NLS}_{\mathfrak{b}}(v_N) \to 0 \quad \text{and} \quad \iint_{\R^2}\eps_NU_N(\eps_N(x-y))|v_N(x)|^2|v_N(y)|^2\dix\diy \to 0 
	\]
	as $N\to\infty$, where we recall that $\mathcal{F}^{\NLS}_{\mathfrak{b}}$ is defined in~\eqref{NLS_Def_mathcal_F}. As usual, the former implies that $v_N\to v$ strongly in $L^r(\R)$ for any $2\leq r < \infty$, where $v$ is an optimizer for~\eqref{ineq:GNS}. On the other hand, we observe that $\eps_N \geq CN^{-\beta}$. This follows from~\eqref{upperbound:Hartree_NLS}, which implies
	\[
		\eps_N^{-2} = \norm{u_N'}_{2}^2 \leq \frac{\mathfrak{b}}{6}\iiint_{\R^3} W_N(x-y,x-z) |u_N(x)|^2 |u_N(y)|^2 |u_N(z)|^2\dix\diy\diz + E_{a,b}^{\NLS} + o(1)\,,
	\]
	and $\norm{W_N}_{\infty}=N^{2\beta}\norm{W}_{\infty}$. When $\alpha > \beta$, we have that $N^{\alpha}\eps_N \geq CN^{\alpha-\beta} \to \infty$ as $N\to\infty$. Then, the scaled two\nobreakdash-body interaction $\eps_NU_N(\eps_N\cdot)$ must converge to the delta interaction $\delta_{x=y}$.
	More precisely, by~\eqref{ineq_cv:Hartree_to_NLS_2body} ---which applies the same to $\eps_N U_N(\eps_N\cdot)$ as long as $\eps_N N^\alpha\to+\infty$---, we have
	\[
		0 \leq \int_{\R}|v_N(x)|^4\dix - \iint_{\R^2} \eps_N U_N(\eps_N(x-y))|v_N(x)|^2|v_N(y)|^2\dix\diy \leq \norm{v_N}_{H^1}^4 o(1) = o(1)\,,
	\]
	and we now use crucially that this $o(1)$ does not depend on $v_N$, in order to obtain the contradiction
	\begin{align*}
		\int_{\R}|v(x)|^4\dix &= \lim_{N\to\infty} \int_{\R}|v(x)|^4\dix - \iint_{\R^2} \eps_N U_N(\eps_N(x-y))|v_N(x)|^2|v_N(y)|^2\dix\diy \\
			&= \lim_{N\to\infty} \int_{\R}|v_N(x)|^4\dix - \iint_{\R^2} \eps_N U_N(\eps_N(x-y))|v_N(x)|^2|v_N(y)|^2\dix\diy = 0\,.
	\end{align*}
	Therefore, we must have desired convergence of energy~\eqref{cv:Hartree_to_NLS_energy}.
	
	We now turn to the proof of~\eqref{cv:Hartree_to_NLS_ground_state}. On one hand, we have proved in all cases that the sequence of Hartree ground states $\{u_N\}_{N}$ is uniformly bounded in $H^1(\R)$. On another hand, \eqref{upperbound:Hartree_NLS}, \eqref{ineq_cv:Hartree_to_NLS_2body}, \eqref{ineq_cv:Hartree_to_NLS_3body}, and~\eqref{ineq:GNS} imply the uniform boundedness of $\norm{|\cdot|^{s/2}u_N}_{2}$. Therefore, up to a subsequence, $u_N\to \phi$ weakly in $H^1(\R)$ and strongly in $L^r(\R)$ for $2\leq r<+\infty$. By~\eqref{ineq_cv:Hartree_to_NLS_2body}, \eqref{ineq_cv:Hartree_to_NLS_3body}, and the weak lower semicontinuity, we have
	\[
		\lim_{N\to\infty}\mathcal{E}_{a,b}^{{\rm H}, N}(u_N) \geq \mathcal{E}_{a,b}^{\NLS}(\phi) \geq E_{a,b}^{\NLS}\,.
	\]
	Together with~\eqref{cv:Hartree_to_NLS_energy}, the above implies that $\phi$ is a ground state of~$E_{a,b}^{\NLS}$. We also obtain
	\[
		\lim_{N\to\infty}\norm{u_N'}_{2}^2 = \norm{\phi'}_{2}^2\,.
	\]
	Hence, $u_N\to \phi$ strongly in $H^1(\R)$. This completes the proof of~\eqref{cv:Hartree_to_NLS_ground_state}.
	
	As for the proof of collapse for Hartree ground states in~\eqref{blowup_Hartree_ground_state}, we start with the upper bound matching the ground state energy expansion. We rewrite
	\begin{multline*}
		\mathcal{E}_{a_N, b_N}^{\mathrm{H}, N}(u) = \mathcal{E}_{a_N, b_N}^{\NLS}(u) + \frac{a_N}{2}\iint_{\R^2} U_N(x-y)|u(x)|^2|u(y)|^2\dix\diy - \frac{a_N}{2}\iint_{\R^2}|u(x)|^4\dix \\
		- \frac{b_N}{6}\iiint_{\R^3} W_N(x-y,x-z)|u(x)|^2|u(y)|^2|u(z)|^2\dix\diy\diz + \frac{b_N}{6}\iiint_{\R^3} |u(x)|^6 \dix
	\end{multline*}
	for any $u\in H^1(\R)$. Choosing~$Q_N := \ell_N^{-1/2}Q_{0}(\ell_N^{-1}\cdot)$ and using~\eqref{cv_rate:Hartree_to_NLS_2body} and~\eqref{cv_rate:Hartree_to_NLS_3body}, we obtain
	\[
		E_{a_N, b_N}^{\mathrm{H}, N} \leq \mathcal{E}_{a_N, b_N}^{\mathrm{H}, N}(Q_N) \leq \mathcal{E}_{a_N, b_N}^{\NLS}(Q_N) - C \min\{0,a_N\} N^{-\alpha} \ell_N^{-2} + C b_N N^{-\beta} \ell_N^{-3}\,.
	\]
	By Theorem~\ref{thm:collapse_NLS} and since $\ell_N=N^{-\eta}$, it yields
	\[
		E_{a_N, b_N}^{\mathrm{H}, N} \leq E_{a_N, b_N}^{\NLS} \left( 1 + o(1) - C \min\{0,a_N\} \ell_N^{-(s+1)} N^{\eta-\alpha} + C N^{\eta(s+3)-\beta} \right)
	\]
	and the error term is negligible since~$\eta<\min\left\{\beta/(s+3), \alpha\right\}$ and~$a_N \ell_N^{-(s+1)}=O(1)$ for all $\zeta\geq0$.

	Let $u_N$ be a ground state of~$E_{a_N, b_N}^{\mathrm{H}, N}$ and $\tilde{u}_N:= \ell_N^{1/2} u_N(\ell_N\cdot)$. Then, $\norm{\tilde{u}_N}_{2} = \norm{u_N}_{2}=1$~and
	\begin{equation}\label{Hartree_general_lowerbound}
		E_{a_N, b_N}^{\mathrm{H}, N} \geq \ell_N^{-2} \mathcal{F}^{\NLS}_{b_N}(\tilde{u}_N) + \ell_N^s \norm{ |\cdot|^{\frac{s}{2}}\tilde{u}_n }_{2}^2 + \frac{a_N}{2}\iint_{\R^2} U_N(\ell_N(x-y))|\tilde{u}_N(x)|^2|\tilde{u}_N(y)|^2\dix\diy \,,
	\end{equation}
	where we used the nonnegativities in~\eqref{ineq_cv:Hartree_to_NLS_3body} and of~$b_N$. We now prove the claimed convergence~\eqref{blowup_Hartree_ground_state}, which immediately implies the ground state energy expansion. Again, we distinguish the two cases $\zeta\neq0$ and $\zeta\neq6$.
	
	\medskip
	\noindent\textbf{Case $\zeta\neq0$.}
	Recall that~$B_N = \mathfrak{b} - b_N > 0$. Combining the upper bound on~$E_{a_N, b_N}^{\mathrm{H}, N}$ with~\eqref{Hartree_general_lowerbound}, the definition of $\ell_N$ depending on $B_N$, and~\eqref{ineq:GNS}, we obtain
	\begin{multline}\label{Hartree_control_nrj_L2_for_gamma_not_0}
		\frac{s+1}{s} - \frac{\zeta}{12} + o(1) \geq \frac{\zeta}{\mathfrak{b}} \norm{\tilde{u}_N'}_{2}^2 + \mathcal{Q}_s^{-1} \norm{|\cdot|^{\frac{s}{2}}\tilde{u}_N}_{2}^2 \\
		+ \left( \frac{6-\zeta}{6} \cdot \frac{\pi}{2} + o(1) \right) \iint_{\R^2} \ell_N U_N(\ell_N(x-y))|\tilde{u}_N(x)|^2|\tilde{u}_N(y)|^2\dix\diy \,.
	\end{multline}
	Thus, by~\eqref{ineq_cv:Hartree_to_NLS_2body} and~$\norm{\tilde{u}_N}_{4}^4\leq 2 \norm{\tilde{u}_N'}_{2}$, we have
	\[
		\frac{s+1}{s} - \frac{\zeta}{12} + o(1) \geq \frac{\zeta}{\mathfrak{b}} \norm{\tilde{u}_N'}_{2}^2 - \left| \frac{6-\zeta}{6} \cdot \frac{\pi}{2} + o(1) \right| \norm{\tilde{u}_N'}_{2} + \mathcal{Q}_s^{-1} \norm{|\cdot|^{\frac{s}{2}}\tilde{u}_N}_{2}^2 \,,
	\]
	implying that~$\{\tilde{u}_N\}_N$ and~$\left\{|\cdot|^{s/2}\tilde{u}_N \right\}_N$ are uniformly bounded in~$H^1(\R)$ and in~$L^2(\R)$, respectively.
	Hence, up to a subsequence, $\tilde{u}_N$ converges to $u$ weakly in $H^1(\R)$ and strongly in $L^r(\R)$ for $2 \leq r < \infty$. Therefore, by~\eqref{ineq_cv:Hartree_to_NLS_2body} ---which applies the same to $\ell_N U_N(\ell_N\cdot)$ as long as $\ell_N N^\alpha\to+\infty$, which is ensured by $\eta < \min\left\{\beta/(s+3), \alpha\right\}$---, we have
	\[
		0 \leq \int_{\R}|\tilde{u}_N(x)|^4\dix - \iint_{\R^2} \ell_N U_N(\ell_N(x-y))|\tilde{u}_N(x)|^2|\tilde{u}_N(y)|^2\dix\diy \leq \norm{\tilde{u}_N}_{H^1}^6 o(1) = o(1)\,,
	\]
	and we now use crucially that this $o(1)$ does not depend on $\tilde{u}_N$ in order to obtain
	\[
		\lim_{N\to\infty} \iint_{\R^2} \ell_N U_N(\ell_N(x-y))|\tilde{u}_N(x)|^2|\tilde{u}_N(y)|^2\dix\diy = \lim_{N\to\infty} \norm{\tilde{u}_n}_{4}^4\dix = \norm{u}_{4}^4\,.
	\]
	Inserting the above into~\eqref{Hartree_control_nrj_L2_for_gamma_not_0} and adapting the arguments in the proof of Theorem~\ref{thm:collapse_NLS}, it concludes the proof in the case $\zeta\neq0$.
	
	\medskip
	\noindent\textbf{Case $\zeta\neq6$.}
	Recall that~$A_N/(6-\zeta)>0$ with $A_N = a_N/\pi$. Combining the upper bound on~$E_{a_N, b_N}^{\mathrm{H}, N}$ with~\eqref{Hartree_general_lowerbound}, with the definition of $\ell_N$ depending on $A_N$, and with~\eqref{ineq:GNS}, we obtain
	\begin{multline*}
		\frac{s+1}{s} - \frac{\zeta}{12} + o(1) \geq \frac{ \zeta + o(1) }{\mathfrak{b}} \norm{\tilde{u}_N'}_{2}^2 + \mathcal{Q}_s^{-1} \norm{|\cdot|^{\frac{s}{2}}\tilde{u}_N}_{2}^2 \\
		+ \frac{6-\zeta}{6} \times \frac{\pi}{2}\iint_{\R^2} \ell_N U_N(\ell_N(x-y))|\tilde{u}_N(x)|^2|\tilde{u}_N(y)|^2\dix\diy
	\end{multline*}
	with $\zeta + o(1) \geq0$ even if $\zeta=0$, as in the proof of Theorem~\ref{thm:collapse_NLS}.
	
	\begin{itemize}[leftmargin=1em]
		\item If $\zeta>0$, we conclude the proof by similar arguments as in the previous case (but with the formula of $\ell_N$ depending on $A_N$ and not $B_N$).
		\item If $\zeta=0$, we assume further that $\beta<\alpha$. As it is now usual, we cannot recover the $H^1(\R)$-boundedness as before, but we have
		\[
			\frac{s+1}{s} + o(1) \geq \frac{ o(1) }{\mathfrak{b}} \norm{\tilde{u}_N'}_{2}^2 + \mathcal{Q}_s^{-1} \norm{|\cdot|^{\frac{s}{2}}\tilde{u}_N}_{2}^2 + \frac{\pi}{2} \iint_{\R^2} \ell_N U_N(\ell_N(x-y))|\tilde{u}_N(x)|^2|\tilde{u}_N(y)|^2\dix\diy\,,
		\]
		with the factor of $\norm{\tilde{u}_N'}_{2}^2$ being nonnegative. Hence,~$\{|\cdot|^{s/2}\tilde{u}_N\}_N$ is uniformly bounded in~$L^2(\R)$. We now prove the uniform boundedness of $\{\tilde{u}_N \}_N$ in $H^1(\R)$. Assume, on the contrary, that it is not the case. Then, there exists a subsequence (still denoted the same) such that $\eps_N := \norm{\tilde{u}_N'}_{2}^{-1} \to0$. 
		Let $v_N:=\eps_N^{1/2} \tilde{u}_N(\eps_N\cdot)$. It satisfies $\norm{v_N}_{2}=1=\norm{v_N'}_{2}$ and it follows from~\eqref{Hartree_general_lowerbound} and the upper bound on~$E_{a_N, b_N}^{\mathrm{H}, N}$ that
		\begin{multline*}
			\left( \frac{s+1}{s} + o(1) \right) \mathcal{Q}_s \ell_N^s \geq \eps_N^{-2}\ell_N^{-2}\mathcal{F}^{\NLS}_{b_N}(v_N) + \eps_N^s \ell_N^s \norm{ |\cdot|^\frac{s}{2}v_N }_{2}^2 \\
			+ \frac{a_N}{2}\iint_{\R^2} U_N(\eps_N \ell_N(x-y)) |v_N(x)|^2 |v_N(y)|^2\dix\diy \,.
		\end{multline*}
		All terms being nonnegative (using~\eqref{ineq:GNS} again) and since $\ell_N, \eps_N \to 0$ as $N\to\infty$, it implies
		\[
			0 \leq \iint_{\R^2} \eps_N \ell_N U_N( \eps_N \ell_N(x-y)) |v_N(x)|^2 |v_N(y)|^2\dix\diy \leq \frac{2}{\pi} \left( \frac{s+1}{s} + o(1) \right) \eps_N \underset{N\to\infty}{\longrightarrow} 0
		\]
		and
		\[
			0 \leq \mathcal{F}^{\NLS}_{b_N}(v_N) \leq \left(1 + o(1)\right) \mathcal{Q}_s \ell_N^{s+2}\eps_N^2 \underset{N\to\infty}{\longrightarrow} 0\,.
		\]
		The latter implies that $v_N\to v$ strongly in~$L^r(\R)$ for any $2\leq r < \infty$, where $v$ is an optimizer to~\eqref{ineq:GNS}, by the same arguments as in the proof of Theorem~\ref{thm:existence_NLS}. On the other hand, we observe that $\eps_N \ell_N = \ell_N \norm{\tilde{u}_N'}_{2}^{-1} \geq CN^{-\beta}$. Indeed, since $a_N>0$, the upper bound on~$E_{a_N, b_N}^{\mathrm{H}, N}$ yields
		\[
			\ell_N^{-2} \norm{\tilde{u}_N'}_{2}^2 \leq \frac{\mathfrak{b}}{6} \iiint_{\R^3} W_N( \ell_N(x-y), \ell_N(x-z) )|\tilde{u}_N(x)|^2 |\tilde{u}_N(y)|^2 |\tilde{u}_N(z)|^2\dix\diy\diz + O(\ell_N^s) = O(N^{2\beta})\,,
		\]
		where we also used~$\norm{W_N}_{\infty} = N^{2\beta}\norm{W}_{\infty}$ and~$\ell_N \sim N^{-\eta}$. Hence, $\eps_N \ell_N N^\alpha \geq C N^{\alpha-\beta} \to +\infty$ and~\eqref{ineq_cv:Hartree_to_NLS_2body} ---which applies the same to $\eps_N \ell_N U_N(\eps_N \ell_N\cdot)$ as long as $\eps_N \ell_N N^\alpha\to+\infty$--- implies
		\[
			0 \leq \norm{v_N}_{4}^4 - \iint_{\R^2} \eps_N U_N(\eps_N(x-y))|v_N(x)|^2|v_N(y)|^2\dix\diy \leq \norm{v_N}_{H^1}^6 o(1) = o(1)\,,
		\]
		and we now use crucially that this $o(1)$ does not depend on $v_N$, in order to obtain the contradiction
		\[
			\norm{v}_{4}^4 = \lim_{N\to\infty} \norm{v_N}_{4}^4 - \lim_{N\to\infty} \iint_{\R^2} \eps_N \ell_N U_N( \eps_N \ell_N(x-y)) |v_N(x)|^2 |v_N(y)|^2\dix\diy =0\,.
		\]
		This concludes the proof of the uniformly $H^1(\R)$-bounded of $\{\tilde{u}_N\}_N$. Together with the uniformly boundedness of $\norm{|\cdot|^{s/2}\tilde{u}_N}_{2}$, it gives (up to a subsequence) that $\tilde{u}_N \to u$ weakly in $H^1(\R)$ and strongly in $L^r(\R)$ for $2 \leq r < \infty$. From here, the proof is the same as in the case $\zeta\neq0$ above and we therefore omit it.
	\end{itemize}
\end{proof}

\section{The mean-field limit and collapse of the many-body ground states}\label{sec:many_body}
	The goal of this section is to prove Theorems~\ref{thm:many_body_critical} and~\ref{thm:many_body_noncritical}. The strategy and main difficulty is to compare the quantum energy $E_{a, b}^{\mathrm{Q}, N}$ in~\eqref{energy:quantum} and the Hartree energy $E_{a, b}^{\mathrm{H}, N}$ in~\eqref{energy:hartree}. The upper bound can be obtained using trial states $u^{\otimes N}$, and the main difficulty lies in the lower bound. Once those bounds are obtained, the result then comes from Theorem~\ref{thm:collapse_Hartree}.

\subsection{Proof of Theorem~\ref{thm:many_body_critical}}
We start by working under the setting $0<b\leq\mathfrak{b}$ and $a>0$. Note that the two\nobreakdash-$\R$\nobreakdash-body interacting term is controlled by the kinetic term, using the 1D Sobolev's inequality and removing center of mass as in~\cite[Lemma~3.7]{LewNamRou-16}. We have
\begin{equation}\label{ineq:2body}
	aN^{\alpha}U(N^{\alpha}(x-y)) \leq \eps(h_x+h_y) + C_\eps
\end{equation}
for every $\alpha > 0$ and $\eps > 0$. Furthermore, the three\nobreakdash-$\R$\nobreakdash-body interacting term $W_N$ defined in~\eqref{Def_W_N} satisfies $\norm{W_N}_{1}=\norm{W}_{1}=1$ and $\norm{W_N}_{\infty}=N^{2\beta}\norm{W}_{\infty}$. Hence,
\[
	E_{a, b}^{\mathrm{Q}, N} = \frac{\pscalSM{ \Psi_N, H_{a, b}^{N} \Psi_N }}{N} \geq \tr\left[h\gamma_{\Psi_N}^{(1)}\right] - \frac{\mathfrak{b}}{6}\norm{W_N}_{\infty} \geq \tr\left[h\gamma_{\Psi_N}^{(1)}\right] - CN^{2\beta}\,,
\]
where we also used $b\leq\mathfrak{b}$ and $a>0$. This gives an \emph{a priori} estimate on the kinetic energy, i.e.,
\begin{equation}\label{Apriori_estimate_kinetic}
	\tr\left[h\gamma_{\Psi_N}^{(1)}\right] \leq CN^{2\beta}\,.
\end{equation}
Taking~\eqref{Apriori_estimate_kinetic} as a starting point and using the localization method in Fock space~\cite{Lewin-11}, one can obtain~\eqref{cv:energy_quantum_to_NLS_critical}. A major ingredient in our proof is the quantitative version of the quantum de~Finetti theorem, originally proved in~\cite{ChrKonMitRen-07} (see also~\cite{FanVan-06,Chiribella-11,Harrow-13-arxiv,LewNamRou-15,LewNamRou-16,LewNamRou-17}). The following formulation is taken from~\cite[Lemma~18]{NamRicTri-23} for its first part and from~\cite[Proof of Lemma~19]{NamRicTri-23} for its second part.
\begin{theorem}[\textbf{Quantitative quantum de Finetti}]\label{thm:quantitative_deF}\leavevmode\\
	Let $\Psi_N\in\mathfrak{H}^N=\bigotimes_{\mathrm{sym}}^{N}L^2(\R)$ and $P$ be a finite-rank orthogonal projector. Then, there exists a positive Borel measure $\mu_{\Psi_N}$ on the unit sphere $SP\mathfrak{H}$ such that
	\begin{equation}\label{ineq:quantitative_deF}
		\tr\Big|P^{\otimes 3}\gamma_{\Psi_N}^{(3)} P^{\otimes 3} - \int_{SP\mathfrak{H}}|u^{\otimes 3}\rangle \langle u^{\otimes 3}| \di\mu_{\Psi_N}(u)\Big| \leq \frac{12\dim(P\mathfrak{H})}{N}\,.
	\end{equation}
	Furthermore, with $Q=\1-P$, we have
	\begin{equation}\label{ineq:quantitative_deF_measure}
		1 \geq \int_{SP\mathfrak{H}} \di \mu_{\Psi_N}(u) = \tr\left[P^{\otimes 3}\gamma_{\Psi_N}^{(3)} P^{\otimes 3}\right] \geq 1- 3\tr\left[Q\gamma_{\Psi_N}^{(1)}\right].
	\end{equation}
\end{theorem}

We will apply the above with $P$ a spectral projector below an energy cut-off $L > 0$ for the one-body operator $h$ given by~\eqref{eq:one_particle_operator}
\begin{equation}\label{Def_projection_P}
	P:=\1(h\leq L)\,.
\end{equation}
Note that, from~\cite[Lemma~3.3]{LewNamRou-16}, we have the following semi-classical inequality ``{\`a} la Cwikel--Lieb--Rozenblum''
\begin{equation}\label{CwikelLiebRosenblum_type_onedimensional}
	\dim(P\mathfrak{H})\lesssim L^{\frac12+\frac1s} \,.
\end{equation}
This, together with~\eqref{ineq:quantitative_deF}, gives a good control the error term made of the energy estimate in low dimension, which is the case in our present work. Now, we use Theorem~\ref{thm:quantitative_deF} to derive an energy lower via the de~Finetti measure.

\begin{lemma}[\textbf{Energy lower bound with the de~Finetti measure}]\leavevmode\\
	Let $\Psi_N$ be an arbitrary wave function in $\mathfrak{H}^N$. Let $\mu_{\Psi_N}$ be the Finetti measure defined in Theorem~\ref{thm:quantitative_deF} with the projector $P$ given by~\eqref{Def_projection_P}. Then, for all $L \gg 1$ we have
	\begin{equation}\label{lowerbound:deF_quantitative}
		\frac{\pscalSM{ \Psi_N, H_{a, b}^{N} \Psi_N }}{N} \geq \int_{SP\mathfrak{H}} \mathcal{E}_{a, b}^{\mathrm{H}, N}(u)\di\mu_{\Psi_N} (u) - C\frac{L+N^{2\beta}}{N}\dim(P\mathfrak{H}) - C\frac{N^{3\beta}}{L^{\frac{1}{2}}}\,.
	\end{equation}
\end{lemma}
This kind of energy lower bound with the de~Finetti measure can be found, e.g., in~\cite{LewNamRou-17,CheHol-17b,Shen-21}.
\begin{proof}
	Following the proof of~\cite[Lemma~4]{LewNamRou-17}, we write
	\begin{equation}\label{lowerbound:deF_quantitative_decomposition}
		\frac{\pscalSM{ \Psi_N, H_{a, b}^{N} \Psi_N }}N = \frac{1}{3}\tr\left[ H_3\gamma_{\Psi_N}^{(3)} \right] = \frac{1}{3}\tr\left[ H_3P^{\otimes3}\gamma_{\Psi_N}^{(3)}P^{\otimes3} \right] + \frac{1}{3}\tr\left[ H_3\left( \gamma_{\Psi_N}^{(3)}-P^{\otimes3}\gamma_{\Psi_N}^{(3)}P^{\otimes3} \right) \right]
	\end{equation}
	and will bound each term from below.
	
	For the main term in~\eqref{lowerbound:deF_quantitative_decomposition}, we use~\eqref{ineq:quantitative_deF} to obtain
	\begin{align*}
		\frac{1}{3}\tr\left[H_3P^{\otimes3}\gamma_{\Psi_N}^{(3)}P^{\otimes3}\right] \geq{}& \frac{1}{3}\tr\left[H_3\int_{SP\mathfrak{H}}|u^{\otimes 3}\rangle \langle u^{\otimes 3}| \di\mu_{\Psi_N}(u)\right] - C \frac{L}{N} \dim(P\mathfrak{H}) - C N^{2\beta-1} \dim(P\mathfrak{H}) \nonumber \\
		={}& \int_{SP\mathfrak{H}} \mathcal{E}_{a, b}^{\mathrm{H}, N}(u) \di\mu_{\Psi_N}(u) - C \frac{L}{N} \dim(P\mathfrak{H}) - C N^{2\beta-1} \dim(P\mathfrak{H})\,.
	\end{align*}
	Here, the first error term is due to~\eqref{ineq:2body} and $Ph \leq LP$, while the second error term follows from the simple estimate
	\[
		P^{\otimes 3} W_N P^{\otimes 3} \leq CN^{2\beta}P^{\otimes 3}\,.
	\]
	
	Let's estimate the error term in~\eqref{lowerbound:deF_quantitative_decomposition}. We use that
	\[
		h = PhP + QhQ
	\]
	and obtain the following decomposition for the three\nobreakdash-body noninteracting Hamiltonian
	\begin{align*}
		h \otimes \1 \otimes \1 + \1 \otimes h \otimes \1 + \1 \otimes \1 \otimes h ={} &
		P^{\otimes 3}\left( h \otimes \1 \otimes \1 + \1 \otimes h \otimes \1 + \1 \otimes \1 \otimes h \right)P^{\otimes 3} \\
		& + PhP \otimes P \otimes Q + PhP \otimes Q \otimes P + PhP \otimes Q \otimes Q \\
		& + P \otimes PhP \otimes Q + Q \otimes PhP \otimes P + Q \otimes PhP \otimes Q \\
		& + P \otimes Q \otimes PhP + Q \otimes P \otimes PhP + Q \otimes Q \otimes PhP \\
		& + QhQ \otimes \1 \otimes \1 + \1 \otimes QhQ \otimes \1 + \1 \otimes \1 \otimes QhQ\,.
	\end{align*}
	For the two\nobreakdash-body interaction term, we write
	\begin{multline*}
		\frac{1}{2}\left( U_{N} \otimes \1 + \1 \otimes U_{N} \right) = \frac{1}{2}P^{\otimes 3}\left( U_{N} \otimes \1 + \1 \otimes U_{N} \right)P^{\otimes 3} + \frac{1}{2}\Pi\left( U_{N} \otimes \1 + \1 \otimes U_{N} \right)\Pi \\
		+ \frac{1}{2}P^{\otimes 3}\left( U_{N} \otimes \1 + \1 \otimes U_{N} \right)\Pi + \frac{1}{2}\Pi\left( U_{N} \otimes \1 + \1 \otimes U_{N} \right)P^{\otimes 3}
	\end{multline*}
	with the orthogonal projection
	\[
		\Pi := P \otimes P \otimes Q + P \otimes Q \otimes P + P \otimes Q \otimes Q + Q \otimes \1 \otimes \1\,.
	\]
	To bound the error terms we use the inequality
	\[
		XAY + YAX \geq -\eps X|A|X - \eps^{-1} Y|A|Y\,,
	\]
	valid for any self-adjoint operator $A$, and any orthogonal projectors $X$, $Y$. Using now $Qh \geq LQ$, \eqref{ineq:2body}, and collecting our estimates, we find
	\begin{align*}
		h \otimes \1 \otimes \1 &+ \1 \otimes h \otimes \1 + \1 \otimes \1 \otimes h + \frac{1}{2}\left( U_{N} \otimes \1 + \1 \otimes U_{N} \right) \\
		\geq{}&
		\begin{multlined}[t]	
			P^{\otimes 3}\left( h \otimes \1 \otimes \1 + \1 \otimes h \otimes \1 + \1 \otimes \1 \otimes h + \frac{1}{2}\left( U_{N} \otimes \1 + \1 \otimes U_{N} \right) \right)P^{\otimes 3} \\
			+ \Pi\left(\frac{L}{4} + \frac{1}{2}(h \otimes \1 \otimes \1 + \1 \otimes h \otimes \1 + \1 \otimes \1 \otimes h) + \frac{1 - \eps^{-1}}{2}\left(U_{N} \otimes \1 + \1 \otimes U_{N} \right) \right)\Pi \\
			- \eps C P^{\otimes 3}\left(U_{N} \otimes \1 + \1 \otimes U_{N} \right)P^{\otimes 3}
		\end{multlined} \\
		\geq{}&
		\begin{multlined}[t]
			P^{\otimes 3}\left( h \otimes \1 \otimes \1 + \1 \otimes h \otimes \1 + \1 \otimes \1 \otimes h + \frac{1}{2}\left( U_{N} \otimes \1 + \1 \otimes U_{N} \right) \right)P^{\otimes 3} \\
			+ \left( \frac{L}{4}-C-C\eps^{-2} \right)\Pi - \eps C\left( h \otimes \1 \otimes \1 + \1 \otimes h \otimes \1 + \1 \otimes \1 \otimes h + C\right)\,.
		\end{multlined}
	\end{align*}
	Choosing now $\eps$ proportional to $L^{-1/2}$ such that $L/4-C-C/\eps^2>0$ for $L\gg1$, taking the trace against $\gamma_{\Psi_N}^{(3)}$, and finally using~\eqref{Apriori_estimate_kinetic}, we obtain for $L\gg1$ that
	\begin{multline}\label{lowerbound:energy_error_quantitative_0}
		\tr\left[\left( h \otimes \1 \otimes \1 + \1 \otimes h \otimes \1 + \1 \otimes \1 \otimes h \vphantom{\frac{1}{2}}+ \frac{1}{2}\left( U_{N} \otimes \1 + \1 \otimes U_{N} \right) \right)\left( \gamma_{\Psi_{N}}^{(3)}-P^{\otimes3}\gamma_{\Psi_{N}}^{(3)}P^{\otimes3} \right)\right] \\
		\geq -CL^{-\frac{1}{2}}\left(\tr\left[h\gamma_{\Psi_N}^{(1)}\right]+C\right) \geq -CL^{-\frac{1}{2}}N^{2\beta}\,.
	\end{multline}
	To deal with the three\nobreakdash-body interaction term, we use the Cauchy--Schwarz inequality for operators
	\[
		\pm(XY+Y^*X^*)\leq r XX^*+ r^{-1}Y^*Y, \quad \forall\, r>0\,,
	\]
	and we write, for all $r>0$, that
	\begin{align*}
		&\pm 2 \left( \gamma_{\Psi_N}^{(3)}-P^{\otimes3}\gamma_{\Psi_N}^{(3)}P^{\otimes3} \right) \\
			={}& \pm\left( \left( 1-P^{\otimes3} \right)\gamma_{\Psi_N}^{(3)} + \gamma_{\Psi_N}^{(3)}\left( 1-P^{\otimes3} \right) + P^{\otimes3}\gamma_{\Psi_N}^{(3)}\left( 1-P^{\otimes3} \right) + \left( 1-P^{\otimes3} \right)\gamma_{\Psi_N}^{(3)}P^{\otimes3} \right) \\
			\leq{}&2r\left( 1-P^{\otimes3} \right)\gamma_{\Psi_N}^{(3)}\left( 1-P^{\otimes3} \right) + r^{-1}\left( \gamma_{\Psi_N}^{(3)}+P^{\otimes3}\gamma_{\Psi_N}^{(3)}P^{\otimes3} \right).
	\end{align*}
	Taking now the trace against $W_N$ and optimizing over $r>0$, we find
	\begin{multline}\label{lowerbound:energy_error_quantitative_1}
		\tr\left[ W_N\left( \gamma_{\Psi_N}^{(3)}-P^{\otimes3}\gamma_{\Psi_N}^{(3)}P^{\otimes3} \right) \right] \geq - \sqrt{2}\left(\tr\left[W_N\left( \gamma_{\Psi_N}^{(3)}+P^{\otimes3}\gamma_{\Psi_N}^{(3)}P^{\otimes3} \right) \right]\right)^{\frac{1}{2}} \\
		\times \left(\tr\left[ W_N\left( 1-P^{\otimes3} \right)\gamma_{\Psi_N}^{(3)}\left( 1-P^{\otimes3} \right) \right] \right)^{\frac{1}{2}}\,.
	\end{multline}
	For the first trace under the square root, we use $W_N \leq CN^{2\beta}$ to obtain
	\begin{equation}\label{lowerbound:energy_error_quantitative_2}
		\tr\left[ W_N\left( \gamma_{\Psi_N}^{(3)}+P^{\otimes3}\gamma_{\Psi_N}^{(3)}P^{\otimes3} \right) \right] \leq CN^{2\beta}\,.
	\end{equation}
	For the second trace under the square root, using again $W_N \leq CN^{2\beta}$, $Q\leq L^{-1}h$ and~\eqref{Apriori_estimate_kinetic}, we find that
	\begin{equation}\label{lowerbound:energy_error_quantitative_3}
		\tr\left[ W_N\left( 1-P^{\otimes3} \right)\gamma_{\Psi_N}^{(3)}\left( 1-P^{\otimes3} \right) \right] \leq CN^{2\beta}\tr\left[ QhQ\gamma_{\Psi_N}^{(1)} \right] \leq \frac{CN^{2\beta}}{L}\tr\left[ h\gamma_{\Psi_N}^{(1)} \right] \leq \frac{CN^{4\beta}}{L}\,.
	\end{equation}
	The second error term in the lower bound~\eqref{lowerbound:deF_quantitative} to the quantum energy follows from~\eqref{lowerbound:energy_error_quantitative_0}--\eqref{lowerbound:energy_error_quantitative_3}.
\end{proof}

Coming back to~\eqref{lowerbound:deF_quantitative} and taking into account~\eqref{CwikelLiebRosenblum_type_onedimensional}, we obtain
\[
	E_{a, b}^{\mathrm{H}, N} \geq E_{a, b}^{\mathrm{Q}, N} \geq \int_{SP\mathfrak{H}} \mathcal{E}_{a, b}^{\mathrm{H}, N}(u) \di\mu_{\Psi_N}(u) - C\frac{L^{\frac{3}{2}+\frac{1}{s}}}{N} - C\frac{N^{3\beta}}{L^{\frac{1}{2}}}\,.
\]
By choosing optimally $L = N^{(3\beta+1) s/(2s+1)}$, we finally obtain
\begin{align}\label{main_ineq_proof_thm_many_body_critical}
	C \geq E_{a, b}^{\mathrm{H}, N} \geq E_{a, b}^{\mathrm{Q}, N} &\geq \int_{SP\mathfrak{H}} \mathcal{E}_{a, b}^{\mathrm{H}, N}(u) \di\mu_{\Psi_N}(u) - CN^{\frac{3\beta(3s+2)-s}{4s+2}} \nonumber \\
		&\geq E_{a, b}^{\mathrm{H}, N} \int_{SP\mathfrak{H}} \di\mu_{\Psi_N}(u) - CN^{\frac{3\beta(3s+2)-s}{4s+2}}\,.
\end{align}
Note that, with this choice of $L$, using again $Q\leq L^{-1}h$ and~\eqref{Apriori_estimate_kinetic}, we have
\[
	0 \leq \tr\left[Q\gamma_{\Psi_N}^{(1)}\right] \leq L^{-1} \tr\left[h\gamma_{\Psi_N}^{(1)}\right] \leq C L^{-1} N^{2\beta} \ll L^{-1} N^{6\beta} \underset{N\to+\infty}{\longrightarrow} 0
\]
and, consequently, by~\eqref{ineq:quantitative_deF_measure} the de~Finetti measure $\mu_{\Psi_N}$ converges to $1$ as $N\to\infty$
\begin{equation}\label{cv_deFinetti_measure}
	\int_{SP\mathfrak{H}} \di\mu_{\Psi_N}(u) \underset{N\to+\infty}{\longrightarrow} 1 \,.
\end{equation}

We can now complete the proof of Theorem~\ref{thm:many_body_critical}.
With the above work, the proof of \emph{(i)} reduces to the combination of~\eqref{main_ineq_proof_thm_many_body_critical} at $b=\mathfrak{b}$ with~\eqref{cv:Hartree_to_NLS_energy} ---hence the requirement $\alpha>\beta$---, since it gives~\eqref{cv:energy_quantum_to_NLS_critical}.

For the proof of \emph{(ii)}, the collapse of ground state energy in~\eqref{asymptotic:quantum_energy_critical} follows from~\eqref{main_ineq_proof_thm_many_body_critical}, \eqref{cv_deFinetti_measure}, and Theorem~\ref{thm:collapse_Hartree}. Hence, we are left with proving~\eqref{asymptotic:many_body_ground_state_critical}. Note that, by~\eqref{main_ineq_proof_thm_many_body_critical}, \eqref{cv_deFinetti_measure}, and the assumption~\eqref{collapse:speed_critical}, we have
\begin{equation}\label{cv:energy_measure_critical}
	\lim_{N\to\infty} \int_{SP\mathfrak{H}} \frac{\mathcal{E}_{a_N, b_N}^{\mathrm{H}, N}(u)}{E_{a_N, b_N}^{\mathrm{H}, N}} \di\mu_{\Psi_N}(u) = \lim_{N\to\infty} \int_{SP\mathfrak{H}} \di\mu_{\Psi_N}(u) = 1\,.
\end{equation}
Note that we also used here the fact that $E_{a_N, b_N}^{\mathrm{H}, N}>0$ with a non-trivial main order, since $\zeta=0$ by assumption in Theorem~\ref{thm:many_body_critical}\emph{(ii)}.
From~\eqref{ineq:quantitative_deF} and Cauchy--Schwarz inequality, we also obtain
\[
	\lim_{N\to\infty} \tr\left|\gamma_{\Psi_N}^{(3)} - \int_{SP\mathfrak{H}} |u^{\otimes 3}\rangle \langle u^{\otimes 3}| \di\mu_{\Psi_N}(u)\right| = 0\,,
\]
which in turn implies that
\[
	\lim_{N\to\infty} \tr\left|\gamma_{\Psi_N}^{(1)} - \int_{SP\mathfrak{H}} |u\rangle \langle u| \di\mu_{\Psi_N}(u)\right| = 0\,.
\]
To complete the proof, it suffices to prove the convergence of one-body density matrix, which is equivalent to
\begin{equation}\label{proof_critical_last_step}
	\lim_{N\to\infty} \int_{SP\mathfrak{H}} \left|\pscal{ u, Q_N }\right| \di\mu_{\Psi_N}(u) = 1 \,,
\end{equation}
where $Q_N = \ell_N^{-1/2}Q_{0}(\ell_N^{-1}\cdot)$ as in the proof of Theorem~\ref{thm:collapse_Hartree}. Defining
\[
	\delta_N := \int_{SP\mathfrak{H}} \left( \frac{\mathcal{E}_{a_N, b_N}^{\mathrm{H}, N}(u)}{E_{a_N, b_N}^{\mathrm{H}, N}} - 1 \right) \di\mu_{\Psi_N}(u)\,,
\]
we have $\delta_N \geq 0$ (recall that $E_{a_N, b_N}^{\mathrm{H}, N}>0$) and, by~\eqref{cv:energy_measure_critical}, $\delta_N \to 0$. Let $T_N$ be the set of all positive normalized functions $u$ in $SP\mathfrak{H}$ satisfying
\begin{equation}\label{behavior_approximate_critical}
	0\leq \frac{ \mathcal{E}_{a_N, b_N}^{\mathrm{H}, N}(u) }{E_{a_N, b_N}^{\mathrm{H}, N}} -1 \leq \sqrt{\delta_N} \,.
\end{equation}
The $T_N$'s are non-empty since they contain Hartree ground states.

We prove that we must have
\begin{equation}\label{limit:PsiQ_critical}
	\lim_{N\to\infty}\inf_{u\in T_N} \left|\pscal{ u, Q_N }\right| = 1 \,.
\end{equation}
If this was not the case, and since $\left|\pscal{ u, Q_N }\right| \leq \norm{u}_2 \norm{Q_N}_2=1$, there would exist a (sub)sequence $\{u_N\} \subset T_N$ such that
\begin{equation}\label{limit:PsiQ_fraud_critical}
	\limsup_{N\to\infty} \left|\pscal{ u_N, Q_N }\right| < 1 \,.
\end{equation}
Since $u_N \in T_N$ and $\delta_N \to 0$, we would deduce from~\eqref{behavior_approximate_critical} that
\[
	\lim_{N\to\infty} \frac{ \mathcal{E}_{a_N, b_N}^{\mathrm{H}, N}(u_N) }{E_{a_N, b_N}^{\mathrm{H}, N}} = 1 \,.
\]
That is, $\{u_N\}_N$ would be a sequence of approximate ground states. Theorem~\ref{thm:collapse_Hartree} would then imply
\[
	\lim_{N\to\infty} \left|\pscal{ u_N, Q_N }\right| = 1 \,,
\]
contradicting~\eqref{limit:PsiQ_fraud_critical}. Hence, we must have~\eqref{limit:PsiQ_critical}.

On the other hand, by the choice of $T_N$, we have
\[
	\frac{ \mathcal{E}_{a_N, b_N}^{\mathrm{H}, N}(u) }{E_{a_N, b_N}^{\mathrm{H}, N}} -1 \geq \sqrt{\delta_N} \,.
\]
for any $u\in {}^\complement T_N := SP\mathfrak{H} \setminus T_N$. Therefore,
\[
	\delta_N = \int_{SP\mathfrak{H}} \left( \frac{ \mathcal{E}_{a_N, b_N}^{\mathrm{H}, N}(u) }{E_{a_N, b_N}^{\mathrm{H}, N}} - 1 \right) \di\mu_{\Psi_N}(u) \geq \int_{{}^\complement T_N} \left( \frac{ \mathcal{E}_{a_N, b_N}^{\mathrm{H}, N}(u) }{E_{a_N, b_N}^{\mathrm{H}, N}} - 1 \right) \di\mu_{\Psi_N}(u) \geq \sqrt{\delta_N} \mu_{\Psi_N}({}^\complement T_N) \,.
\]
Thus, $\mu_{\Psi_N}({}^\complement T_N) \leq \sqrt{\delta_N} \to 0$ and consequently $\mu_{\Psi_N}(T_N) \to 1$, by~\eqref{cv:energy_measure_critical}. The latter convergence and~\eqref{limit:PsiQ_critical} imply that
\[
	\int_{SP\mathfrak{H}} \left|\pscal{ u, Q_N }\right| \di\mu_{\Psi_N}(u) \geq \int_{T_N} \left|\pscal{ u, Q_N }\right| \di\mu_{\Psi_N}(u) \geq \mu_{\Psi_N}(T_N)\inf_{u \in T_N} \left|\pscal{ u, Q_N }\right| \to 1 \,.
\]
Thus~\eqref{proof_critical_last_step} holds true and the proof of Theorem~\ref{thm:many_body_critical} is finished. \qed

\subsection{Proof of Theorem~\ref{thm:many_body_noncritical}}
We start with an inequality on the one-body kinetic operator $h_x$ defined in~\eqref{eq:one_particle_operator}.
\begin{lemma}\label{1D_apriori_lowerbound_on_h}
	As operators, we have
	\[
		-\frac{\mathrm{d}^2}{\mathrm{d} x^2} + 1 \leq C h_x \,.
	\]
\end{lemma}
\begin{proof}
	From Sobolev's inequality in 1D (see, e.g.,~\cite[Theorem~8.5]{LieLos-01}), we have
	\begin{align*}
		\norm{f}_{2}^{2} = \int_{-1}^1|f|^2+\int_{\R\setminus(-1,1)}|f|^2 &\leq 2\norm{f}_{\infty}^2+\int_{\R\setminus(-1,1)}|x|^s|f(x)|^2\dix \\
			&\leq 2\norm{f'}_{2}\norm{f}_{2}+\int_{\R}|x|^s|f(x)|^2\dix \\
			&\leq \frac{1}{2}\norm{f}_{2}^{2} + 2\norm{f'}_{2}^{2} + \int_{\R}|x|^s|f(x)|^2\dix \,.
	\end{align*}
	for every function $f\in H^1(\R)$. Hence
	\[
		\norm{f}_{2}^{2} \leq 4\norm{f'}_{2}^{2} + 2\int_{\R}|x|^s|f(x)|^2\dix
	\]
	which is nothing else than the operator inequality $1\leq -4\frac{\mathrm{d}^2}{\mathrm{d} x^2}+2|x|^s$. Thus,
	\[
		-\frac{\mathrm{d}^2}{\mathrm{d} x^2} + 1 \leq 5\left( -\frac{\mathrm{d}^2}{\mathrm{d} x^2} + |x|^s \right) = 5 h_x. \qedhere
	\]
\end{proof}

Note that \eqref{Apriori_estimate_kinetic} is a bad priori estimate on the kinetic energy and cannot be improved in the case $b=\mathfrak{b}$. In the contrary case where $b < \mathfrak{b}$, one can derive moment estimates. Together with a bootstrap argument as in~\cite{LewNamRou-17,NamRou-20}, one can obtain the convergence of energy~\eqref{cv:energy_quantum_to_NLS_noncritical} for a wider range of $\beta>0$. A major ingredient in our proof is the information-theoretic quantum de~Finetti theorem from~\cite{LiSmi-15,BraHar-17}. The following formulation is taken from~\cite[Theorem~3.2]{Girardot-20} (see~\cite{Rougerie-20a,Rougerie-20b} for a general discussion and more references).
\begin{theorem}[\textbf{Information-theoretic quantum de Finetti}]\label{thm:information_deF}\leavevmode\\
	Let $\mathfrak{H}$ be a complex separable Hilbert space, and $\mathfrak{H}_N = \mathfrak{H}^{\otimes_{\mathrm{sym}} N}$ the corresponding bosonic space. Let $\gamma_{\Psi_N}^{(3)}$ be the $3$\nobreakdash-body reduced density matrix of a $N$\nobreakdash-body state vector $\Psi_N \in \mathfrak{H}_N$ (or a general mixed state) and $P$ be a finite dimensional orthogonal projector. There exists a Borel measure $\mu_{\Psi_N}$ with total mass $\leq 1$ on the set of one-body mixed states
	\[
		\mathcal{S}_P := \left\{ \gamma \mbox{ positive trace-class self-adjoint operator on } P \mathfrak{H}, \, \tr \gamma = 1 \right\}
	\]
	such that
	\begin{equation}\label{ineq:information_deF}
		\sup_{0\leq A,B,C\leq 1}\tr\left| A \otimes B \otimes C \left( P^{\otimes 3} \gamma_{\Psi_N}^{(3)} P^{\otimes 3} - \int_{\mathcal{S}_P} \gamma^{\otimes 3} \di\mu_{\Psi_N}(\gamma) \right)\right| \leq C \sqrt{\frac{\log (\dim P)}{N}}
	\end{equation}
	where the sup is over bounded operators on $P\mathfrak{H}$. Furthermore,
	\[
		1\geq \int_{\mathcal{S}_P} \di\mu_{\Psi_N}(\gamma) = \tr\left[P^{\otimes 3}\gamma_{\Psi_N}^{(3)} P^{\otimes 3}\right] \geq 1- 3\tr\left[Q\gamma_{\Psi_N}^{(1)}\right].
	\]
\end{theorem}

We will again apply Theorem~\ref{thm:information_deF} to the spectral projector $P$ given by~\eqref{Def_projection_P}. In order to made use of~\eqref{ineq:information_deF}, we shall decompose the interaction operators using the Fourier transform in the manner
\begin{align}
	N^{\alpha} U(N^\alpha(x-y)) &= \int_{\R} \widehat{U}(N^{-\alpha}k) e^{ikx} e^{-iky}\di k\label{eq:decomp_U_noncritical}
	\intertext{and}
	N^{2\beta} W(N^\beta(x-y),N^\beta(x-z)) &= \iint_{\R^2} \widehat{W}(N^{-\beta}k_1,N^{-\beta}k_2)] e^{i(k_{1}+k_{2})x} e^{-ik_{1}y} e^{-ik_{2}z}\di k_{1}\di k_{2} \,.\label{eq:decomp_W_noncritical}
\end{align}
The involved multiplication operators in the above integral are indeed of the form $A \otimes B \otimes C$. We will interject a simple control of each term separately.

Our two key ingredients in the proof of Theorem~\ref{thm:many_body_noncritical} are the new two lemmas.
\begin{lemma}[\textbf{Lower bound with the de Finetti measure}]\label{lem:lowerbound_deF}\leavevmode\\
	Let $\Psi_N$ be an arbitrary wave function in $\mathfrak{H}^N$. Let $\mu_{\Psi_N}$ be the Finetti measure defined in Theorem~\ref{thm:information_deF} with the projector $P$ given by~\eqref{Def_projection_P}. Then, for all $L \gg 1$ and $0<\delta \leq 1/2$ we have
	\begin{multline}\label{lowerbound:information_deF}
		\frac{\pscalSM{ \Psi_N, H_{a, b}^{N} \Psi_N }}{N} \geq \int_{\mathcal{S}_P} \mathcal{E}_{a, b}^{\mathrm{mH}, N}(\gamma)\di\mu_{\Psi_N}(\gamma) - C\sqrt{\frac{\log L}{N}} L \log N - CL^{-\frac{1}{2}}\tr\left[h\gamma_{\Psi_N}^{(1)}\right] \\
		- C_{\delta}L^{-\frac{1}{4}+\frac{1}{2}\delta}\left( \tr\left[ h \gamma^{(1)}_{\Psi_N} \right] \right)^{\frac14-\frac\delta2}\left( \tr\left[ h\otimes h \gamma^{(2)}_{\Psi_N} \right] \right)^{\frac12+\delta}\,.
	\end{multline}
	Here, $\mathcal{E}_{a, b}^{\mathrm{mH}, N}(\gamma)$ is the modified Hartree energy defined on positive trace-class self-adjoint operators,
	\begin{multline*}
		\mathcal{E}_{a, b}^{\mathrm{mH}, N}(\gamma) := \tr\left[ h\gamma \right] + \frac{a}{2} \iint_{\R^2} N^{\alpha} U(N^{\alpha}(x-y)) \rho_{\gamma}(x)\rho_{\gamma}(y) \dix\diy \\
		- \frac{b}{6} \iiint_{\R^3} N^{2\beta} W(N^{\beta}(x-y),N^{\beta}(x-z)) \rho_{\gamma}(x)\rho_{\gamma}(y)\rho_{\gamma}(z) \dix\diy\diz
	\end{multline*}
	and the density $\rho_{\gamma}(x)=\gamma(x,x)$ ---defined properly by the spectral decomposition--- satisfies $\int_{\R}\rho_{\gamma} = \tr \gamma =1$. Moreover,
	\begin{equation}\label{measure_information_deF}
		1\geq \int_{\mathcal{S}_P} \di\mu_{\Psi_N}(\gamma) \geq 1 - 3\tr\left[Q\gamma_{\Psi_N}^{(1)} \right] \geq 1-3L^{-1} \tr\left[ h\gamma_{\Psi_N}^{(1)} \right]\,.
	\end{equation}
\end{lemma}

Lemma~\ref{lem:lowerbound_deF} provides a sharp lower bound to the ground state energy if we have a strong enough a priori control of the error terms in~\eqref{lowerbound:information_deF}. To go further, we need the following moment estimates.

\begin{lemma}[\textbf{Moments Estimates}]\label{Lemma_Moments_Estimates}\leavevmode\\
	Let $a\in\R$, $b>0$, and $0<\alpha, \beta<1$. Let $\Psi_N\in \mathfrak{H}^N$ be a ground state of~$H^{N}_{a, b}$, and $E_{a, b}^{\mathrm{Q}, N, \eps}$ be the ground state energy of the modified Hamiltonian $H^{N,\eps}_{a,b} := H^{N}_{a,b} - \eps\sum_{j=1}^N h_j$ , i.e.,
	\[
		E_{a, b}^{\mathrm{Q}, N, \eps} := N^{-1} \inf \left\{ \pscal{\Psi_N, H^{N,\eps}_{a,b} \Psi_N}:\Psi_N\in\mathfrak{H}^N, \norm{\Psi_N}_{2}=1\right\}.
	\]
	Then, for every $\eps\in(0,1)$, we have
	\begin{equation}\label{Moments_Estimates}
		\tr\left[ h \gamma^{(1)}_{\Psi_N} \right]\lesssim \frac{1+\left|E_{a, b}^{\mathrm{Q}, N, \eps}\right|}\eps \qquad \textrm{ and } \qquad \tr\left[ h\otimes h \gamma^{(2)}_{\Psi_N} \right] \lesssim \left( \frac{1+\left|E_{a, b}^{\mathrm{Q}, N, \eps}\right|}\eps \right)^2.
	\end{equation}
\end{lemma}

The first moment estimate in~\eqref{Moments_Estimates} can be obtaind directly using the Schr{\"o}\-dinger equation satisfying by the ground state $\Psi_N$. In order to obtain the second moment estimate in~\eqref{Moments_Estimates}, we need the following bounds for the three\nobreakdash-body interactions (which are similar to~\cite[Lemma~3.2]{NamRouSei-16}, obtained for two\nobreakdash-body interactions), some of which will also be used in the proof of Lemma~\ref{lem:lowerbound_deF}.

\begin{lemma}[\textbf{Operator bounds for three\nobreakdash-body interaction}]\label{lem:W_operator_bounds_three_R_body}\leavevmode\\
	Let $h$ be given by~\eqref{eq:one_particle_operator}. For every $W \in L^p(\R\times\R, \R)$ with $p > 1$ and satisfying the symmetries~\eqref{condition:potential_three_body_symmetry}, the multiplication operator $W(x-y,x-z)$ on $L^2(\R^3)$ satisfies
	\begin{equation}\label{W_operator_bounds_three_R_body_h_Lp}
		|W(x-y,x-z)| \lesssim \norm{W}_{p}(1-\Delta_y-\Delta_z) \lesssim \norm{W}_{p}(h_y+h_z)
	\end{equation}
	and
	\begin{equation}\label{W_operator_bounds_three_R_body_hx_hy_L1_delta}
		\begin{aligned}[b]
			|W(x-y,x-z)| &\leq \norm{G_\delta}_{\infty}^2\norm{W}_{1}(1-\Delta_y)^{\frac12+\delta}(1-\Delta_z)^{\frac12+\delta} \\
				&\leq C_{\delta}\norm{G_\delta}_{\infty}^2\norm{W}_{1}(h_y h_z)^{\frac12+\delta}
		\end{aligned}
	\end{equation}
	for any $\delta > 0$. Here $G_\delta$ is the Green function of $(1-\Delta_x)^{-1/2-\delta}$.
	
	If in addition $\nabla_{1}W \in L^q(\R\times\R, \R)$ with $q>1$, then
	\begin{equation}\label{W_operator_bounds_three_R_body_hx_hy_Lp_grad_Ls}
		\pm \{ h_x, W(x-y,x-z) \} \lesssim \left(\norm{W}_{p} + \norm{\nabla_{1}W}_{q}\right) h_x(h_y+h_z) \,.
	\end{equation}
\end{lemma}
Note that, from now on and for shortness, we denote by $\{A,B\}$ the anticommutator $AB+BA$.

\begin{remark*}
	The bound in~\eqref{W_operator_bounds_three_R_body_hx_hy_L1_delta} could be distributed to all variables and we can actually prove that
	\[
		|W(x-y,x-z)| \leq C(1-\Delta_x)^{\frac13+\delta}(1-\Delta_y)^{\frac13+\delta}(1-\Delta_z)^{\frac13+\delta}, \quad \forall\, \delta > 0.
	\]
	With this better operator bound, however, we need to control the third moment estimate which is not available.
\end{remark*}
\begin{proof}[Proof of Lemma~\ref{lem:W_operator_bounds_three_R_body}]
	The proof of the first inequality of~\eqref{W_operator_bounds_three_R_body_h_Lp} relies on the Sobolev embedding $H^1(\R^2)\hookrightarrow L^r(\R^2)$ for $2\leq r<\infty$ and goes as follow. For $2\leq 2q<\infty$ and $1<p=q/(q-1)\leq\infty$, we have
	\begin{align*}
		\pscal{f(x,y,z),|W(x-y,x-z)| f(x,y,z)}_{L^2(\R^3)} \leq{}& \norm{W}_{p}\int_{\R}\norm{f(x,\cdot,\cdot)}_{L^{2q}(\R^2)}^2\dix \nonumber \\
		\lesssim{}& \norm{W}_{p} \int_{\R}\norm{f(x,\cdot,\cdot)}_{H^1(\R^2)}^2\dix \nonumber \\
		={}& \norm{W}_{p} \pscal{f(x,y,z),(1-\Delta_y-\Delta_z)f(x,y,z)}_{L^2(\R^3)} \,.
	\end{align*}
	The second inequality of~\eqref{W_operator_bounds_three_R_body_h_Lp} is obtain using Lemma~\ref{1D_apriori_lowerbound_on_h}.
	
	For the proof of~\eqref{W_operator_bounds_three_R_body_hx_hy_L1_delta}, we use as in~\cite[Proof of Lemma 3.2]{NamRouSei-16} that~\eqref{W_operator_bounds_three_R_body_hx_hy_L1_delta} is equivalent to
	\[
		\sqrt{|W(x-y,x-z)|}(1-\Delta_y)^{-\frac12-\delta}(1-\Delta_z)^{-\frac12-\delta}\sqrt{|W(x-y,x-z)|} \leq \norm{G_\delta}_{\infty}^2\norm{W}_{1} \,.
	\]
	The Fourier transform of the Green function $G_\delta$ of $(1-\Delta_x)^{-1/2-\delta}$ is
	\[
		\widehat{G_\delta}(\xi):=\int_{\R}e^{-2\pi\xi x}G_\delta(x)\dix=\frac1{(1+4\pi^2\xi^2)^{\frac12+\delta}}\,,
	\]
	which verifies for any $\delta>0$, that
	\[
		\norm{G_\delta}_{\infty} \leq \normSM{\widehat{G_\delta}}_{1} = \int_\R \frac{\di\xi}{(1+4\pi^2\xi^2)^{\frac12+\delta}}<+\infty \,.
	\]
	Thus, for every $f\in L^2(\R^3)$, recalling the symmetries~\eqref{condition:potential_three_body_symmetry} of $W$, we have
	\begin{align*}
		&\pscal{f,\sqrt{|W(x-y,x-z)|}(1-\Delta_y)^{-\frac12-\delta}(1-\Delta_z)^{-\frac12-\delta}\sqrt{|W(x-y,x-z)|}f} \\
		={}& \begin{multlined}[t][0.9\textwidth]
			\int_{\R^5}\overline{f(x,y,z)}\sqrt{|W(x-y,x-z)|}G_\delta(y-y') \\
			\times G_\delta(z-z')\sqrt{|W(x-y',x-z')|}f(x,y',z')\dix\diy\diz\diy'\diz'
			\end{multlined} \\
		\leq{}& \norm{G_\delta}_{\infty}^2\int_{\R^5}|f(x,y,z)|\sqrt{|W(x-y,x-z)} \sqrt{|W(x-y',x-z')|}|f(x,y',z')|\dix\diy\diz\diy'\diz' \\
		={}& \norm{G_\delta}_{\infty}^2 \norm{W}_{1}\norm{f}_{L^2(\R^3)}^{2}\,.
	\end{align*}
	The inequality is obtained again using Lemma~\ref{1D_apriori_lowerbound_on_h}.
	
	For the proof of~\eqref{W_operator_bounds_three_R_body_hx_hy_Lp_grad_Ls}, recalling the symmetries~\eqref{condition:potential_three_body_symmetry} of $W$, we have
	\[
		W(X,Y)=W(Y,X) \Rightarrow \nabla_{1} W(X,Y) = \nabla_{2} W(Y,X) \,.
	\]
	Hence, we obtain
	\begin{align*}
		\nabla_x[W(x-y,x-z)] &= \nabla_{1}[W(x-y,x-z)] + \nabla_{2}[W(x-y,x-z)] \\
			&= \nabla_{1}[W(x-y,x-z)] + \nabla_{1}[W(x-z,x-y)] \,.
	\end{align*}
	Then, a straightforward computation using integration by part gives us
	\begin{align*}
		\pscal{f, \pm\left\{ -\Delta_x, W(x-y,x-z) \right\} f} ={}& \pm2\Re\pscal{\nabla_xf,\nabla_x\left[W(x-y,x-z)f\right]} \\
		={}& \pm 2\pscal{\nabla_xf,W(x-y,x-z)\nabla_xf} \\
		& \pm 2\Re\pscal{\nabla_x f,(\nabla_{1}W)(x-y,x-z) f} \\
		& \pm 2\Re\pscal{\nabla_x f,(\nabla_{1} W)(x-z,x-y) f}
	\end{align*}
	for any $f\in H^2(\R^3)$. Moreover
	\[
		\pm\pscal{\nabla_xf,W(x-y,x-z)\nabla_xf}\leq\pscal{\nabla_xf,|W(x-y,x-z)|\nabla_xf}
	\]
	and
	\begin{align*}
		& \mp 2\Re\pscal{\nabla_xf,(\nabla_{1}W)(x-y,x-z)f} \\
		\leq{}& 2\pscal{|\nabla_xf|\sqrt{|(\nabla_{1}W)(x-y,x-z)|},|f|\sqrt{|(\nabla_{1}W)(x-y,x-z)|}} \\
		\leq{}& \pscal{\nabla_xf,|(\nabla_{1}W)(x-y,x-z)|\nabla_xf}+\pscal{f,|(\nabla_{1}W)(x-y,x-z)|f}\,.
	\end{align*}
	Now, we apply~\eqref{W_operator_bounds_three_R_body_h_Lp} to $|W(x-y,x-z)|$, $|(\nabla_{1}W)(x-y,x-z)|$ and $|(\nabla_{1}W)(x-z,x-y)|$, together with Lemma~\ref{1D_apriori_lowerbound_on_h}, and obtain
	\begin{align}\label{W_operator_bounds_three_R_body_hx_hy_Lp_grad_Ls_case1_part1}
		&\pscal{f, \pm\left\{ -\Delta_x, W(x-y,x-z) \right\} f} \nonumber \\
		\lesssim{}& \norm{W}_{p}\pscal{f,(1-\Delta_y-\Delta_z)(-\Delta_x)f} + \norm{\nabla_{1}W}_{q}\pscal{f,(1-\Delta_y-\Delta_z)(1-\Delta_x)f} \nonumber \\
		\lesssim{}& \norm{W}_{p}\pscal{f,(h_y+h_z)(-\Delta_x)f} + \norm{\nabla_{1}W}_{q}\pscal{f,(h_y+h_z)(1-\Delta_x)f}\,.
	\end{align}
	We conclude by proving, for any $s>0$, that
	\[
		\pm\left\{ |x|^s, W(x-y,x-z) \right\} =\pm2 W(x-y,x-z)|x|^s \lesssim \norm{W}_{p}(h_y+h_z)|x|^s\,.
	\]
	For any $f\in L^2(\R^3)$, we can suppose $|x|^{s/2}f(x,y,z)\in L^2(\R^3)$ as the inequality is otherwise trivially true since the r.h.s.\ is then $+\infty$. Therefore, we can use~\eqref{W_operator_bounds_three_R_body_h_Lp} to obtain
	\[
		\pm \pscal{f,|x|^s W(x-y,x-z)f} \lesssim \norm{W}_{p}\pscal{|x|^{\frac{s}{2}}f,(h_y+h_z)|x|^{\frac{s}{2}}f}= \norm{W}_{p}\pscal{f,(h_y+h_z)|x|^s f}\,.
	\]
	Together with~\eqref{W_operator_bounds_three_R_body_hx_hy_Lp_grad_Ls_case1_part1}, this gives the announced result
	\[
		\pm\{ h_x, W(x-y,x-z) \} \lesssim \left( \norm{W}_{p}+ \norm{\nabla_{1}W}_{q}\right) (h_y+h_z)h_x \,.
	\]
	This concludes the proof of Lemma~\ref{lem:W_operator_bounds_three_R_body}.
\end{proof}

We also need the following bounds on the two\nobreakdash-body interactions. These bounds are the 1D analogue of~\cite[Lemma~3.2]{NamRouSei-16} in 3D (see also~\cite{LewNamRou-17} in 2D) with a very similar proof that we therefore omit for brevity.
\begin{lemma}[\textbf{Operator bounds for two\nobreakdash-body interaction}]\label{lem:U_operator_bounds_two_R_body}\leavevmode\\
	For every $U \in L^1(\R)$, the multiplication operator $U(x-y)$ on $L^2(\R^2)$ satisfies
	\begin{align}
		&|U(x-y)| \leq \norm{U}_1 (1 - \Delta_x ) \lesssim \norm{U}_1 h_x\,, \label{U_operator_bounds_two_R_body_h_L1} \\
		&|U(x-y)| \leq \norm{G_\delta}_2^2 \norm{U}_1 (1 - \Delta_x )^{\frac14+\delta} (1 - \Delta_y )^{\frac14+\delta} \lesssim \norm{G_\delta}_2^2 \norm{U}_{1} (h_y h_z)^{\frac14+\delta}\,, \quad \forall\, \delta>0\,, \label{U_operator_bounds_two_R_body_hx_hy_L1_delta}
		\intertext{where $G_\delta$ is the Green function of $(1-\Delta_x)^{-1/4-\delta}$, and if additionally $U \in L^2(\R)$, then}
		\pm &\{ h_x, U(x-y) \} \leq 2 \left( \norm{U}_1 + \norm{U}_2 \right) h_x h_y\,. \label{U_hx_anticommutator_operator_bounds}
	\end{align}
\end{lemma}

With Lemmas~\ref{lem:W_operator_bounds_three_R_body} and~\ref{lem:U_operator_bounds_two_R_body}, we can now prove Lemma~\ref{Lemma_Moments_Estimates}. The proof follows very closely the one of the same result~\cite[Lemma~5]{LewNamRou-17} for the 2-body interaction in 2D. The proof of the first result is even strictly the same. For the convenience of the reader we nevertheless do it here in details.

\begin{proof}[Proof of Lemma~\ref{Lemma_Moments_Estimates}]
	For the first result in~\eqref{Moments_Estimates}, we have on one hand $E_{a, b}^{\mathrm{Q}, N, \eps} \leq E_{a, b}^{\mathrm{Q}, N} \leq E_{a, b}^{\mathrm{H}, N}\leq C$ and, by definitions of $H^{N,\eps}_{a,b}$ and $E_{a, b}^{\mathrm{Q}, N, \eps}$, it implies that $\left|E_{a, b}^{\mathrm{Q}, N}\right| \leq C (1+ \left|E_{a, b}^{\mathrm{Q}, N, \eps}\right| )$. On another hand, again by definition of $E_{a, b}^{\mathrm{Q}, N, \eps}$ we have $H_{a,b}^{N,\eps} \geq N E_{a, b}^{\mathrm{Q}, N, \eps}$ and consequently
	\[
		E_{a, b}^{\mathrm{Q}, N, \eps} \leq \frac1N \pscal{\Psi_N, H_{a,b}^{N,\eps} \Psi_N} = E_{a, b}^{\mathrm{Q}, N} \norm{\Psi_N}_{2}^{2} - \frac{\eps}{N} \sum_{j=1}^N \pscal{\Psi_N, h_j\Psi_N} = E_{a, b}^{\mathrm{Q}, N} - \eps \pscal{\Psi_N, h_{1} \Psi_N},
	\]
	where we used that $\Psi_N$ is a ground state (in particular, $\norm{\Psi_N}_{2}=1$). Therefore, combining the above and using the definition of $\gamma^{(1)}_{\Psi_N}$, we obtain
	\begin{equation}\label{Proof_Lemma_Moments_Estimates_first}
		\eps \tr\left[ h \gamma^{(1)}_{\Psi_N} \right] = \eps \pscal{\Psi_N, h_{1} \Psi_N} \leq E_{a, b}^{\mathrm{Q}, N} - E_{a, b}^{\mathrm{Q}, N, \eps} \leq C\left (1 + \left|E_{a, b}^{\mathrm{Q}, N, \eps}\right| \right)\,.
	\end{equation}
	
	For the second result in~\eqref{Moments_Estimates}, the strategy is to bound $\tr\left[ h\otimes h \gamma^{(2)}_{\Psi_N} \right]$, using the definition of the two\nobreakdash-body density matrix $\gamma^{(2)}_{\Psi_N}$, as follows
	\begin{align}\label{Proof_Lemma_Moments_Estimates_second_initial_ineq}
		\pscal{\Psi_N, \biggl( \sum_{j=1}^N h_j \biggr)^2\Psi_N} ={} & \sum_{j=1}^N \pscal{\Psi_N,h_j^2\Psi_N} + \sum_{1\leq j\neq k\leq N} \pscal{\Psi_N,h_j h_k \Psi_N} \nonumber \\
		={} & N\tr\left[ h \gamma^{(1)}_{\Psi_N} \right] + N(N-1)\tr\left[ h\otimes h \gamma^{(2)}_{\Psi_N} \right] \geq C N^2 \tr\left[ h\otimes h \gamma^{(2)}_{\Psi_N} \right].
	\end{align}
	Then, we bound the LHS of~\eqref{Proof_Lemma_Moments_Estimates_second_initial_ineq} by
	\[
		\pscal{\Psi_N, \left\{ \sum_{j=1}^N h_j , H_{a, b}^{N} \right\} \Psi_N} \quad \text{and} \quad \frac{\left( 1+\left|E_{a, b}^{\mathrm{Q}, N, \eps}\right| \right)^2}\eps\,,
	\]
	and finally the former by the latter. We start from the identity
	\begin{equation}\label{Proof_Lemma_Moments_Estimates_second_initial_identity}
		\left\{ \sum_{j=1}^N h_j , H_{a, b}^{N} \right\} = \biggl( \sum_{j=1}^N h_j \biggr) H_{a, b}^{N} + H_{a, b}^{N} \biggl( \sum_{j=1}^N h_j \biggr) = 2\biggl( \sum_{j=1}^N h_j \biggr)^2 + \mathcal{W}_N\,,
	\end{equation}
	where
	\[
		\mathcal{W}_N := \!\!\! \sum_{1\leq i \leq N} \! \left\{ h_i, \frac{a}{N-1} \sum_{1\leq j<k\leq N} U_N(x_j-x_k) - \frac{b}{(N-1)(N-2)} \sum_{1\leq j<k<\ell\leq N} W_N(x_j-x_k, x_j-x_\ell) \right\}
	\]
	and distinguish two cases for the summations. On one hand, for any fixed $i$, using the splitting
	\[
		\sum_{\substack{1\leq j<k<\ell\leq N \\ j,k,\ell \neq i}} = \sum_{1\leq j<k<\ell\leq N } - \biggl( \sum_{i= j<k<\ell\leq N}+\sum_{1\leq j<k=i<\ell\leq N}+\sum_{1\leq j<k<\ell=i\leq N} \biggr),
	\]
	of the four indices summation, the similar one for the three indices summation, the definition of $H_{a,b}^{N,\eps}$ through the identity
	\[
		\frac{a}{N-1} \sum_{1\leq j<k\leq N} U_N(x_j-x_k) - \frac{b}{(N-1)(N-2)} \sum_{1\leq j<k<\ell\leq N} W_N(x_j-x_k, x_j-x_\ell) = H_{a,b}^{N,\eps} - (1-\eps) \sum_{j=1}^N h_j\,,
	\]
	and the symmetries of $U$ and $W$, we obtain for any $r>0$ (and still any fixed $i$) that
	\begin{multline*}
		\frac{a}{N-1} \sum_{\substack{1\leq j<k\leq N \\ j \neq i \neq k}} U_N(x_j-x_k) - \frac{b}{(N-1)(N-2)} \sum_{\substack{1\leq j<k<\ell\leq N \\ j,k,\ell \neq i}} W_N(x_j-x_k, x_j-x_\ell) \\
		\begin{aligned}[t]
			&=
			\begin{multlined}[t][0.75\textwidth]
				H_{a,b}^{N,\eps} - (1-\eps) \biggl( \sum_{j=1}^N h_j \biggr) - \frac{a}{N-1} \sum_{\substack{1\leq j \leq N \\ j \neq i }} U_N(x_i-x_j) \\
				+ \frac{b}{(N-1)(N-2)} \sum_{\substack{1\leq j<k\leq N \\ j \neq i \neq k}} W_N(x_i-x_j, x_i-x_k)
			\end{multlined} \\
			&\geq N E_{a, b}^{\mathrm{Q}, N, \eps} - (1-\eps) \biggl( \sum_{j=1}^N h_j \biggr) - a_+ \frac{\norm{U_N}_1}{N-1} \sum_{\substack{1\leq j\leq N \\ j \neq i}} h_j \\
			&\geq N E_{a, b}^{\mathrm{Q}, N, \eps} - \left( 1 - \eps + 2 |a| N^{-1} \right) \sum_{j=1}^N h_j\,.
		\end{aligned}
	\end{multline*}
	where $a_+:=\max\{0,a\}\leq |a|$. Here, we used the definition of $E_{a, b}^{\mathrm{Q}, N, \eps}$, $b>0$, $W\geq0$, and~\eqref{U_operator_bounds_two_R_body_h_L1} for the first inequality, and $h\geq0$ as well as $\norm{U_N}_1=\norm{U}_1=1$ and $1/(N-1)\leq2/N$ for the last one. Therefore,
	\begin{multline}\label{Proof_Lemma_Moments_Estimates_second_split1}
		\sum_{1\leq i \leq N} \left\{ h_i, \frac{a}{N-1} \sum_{\substack{1\leq j<k\leq N \\ j,k \neq i}} U_N(x_j-x_k) - \frac{b}{(N-1)(N-2)} \sum_{\substack{1\leq j<k<\ell\leq N \\ j,k,\ell \neq i}} W_N(x_j-x_k, x_j-x_\ell) \right\} \\
		\geq 2 N E_{a, b}^{\mathrm{Q}, N, \eps} \biggl( \sum_{j=1}^N h_j \biggr) - 2 \left( 1 - \eps + 2 |a| N^{-1} \right) \biggl( \sum_{j=1}^N h_j \biggr)^2.
	\end{multline}
	On another hand, by~\eqref{W_operator_bounds_three_R_body_hx_hy_Lp_grad_Ls} and~\eqref{U_hx_anticommutator_operator_bounds}, using the splitting
	\[
		\sum_{\substack{1\leq j<k<\ell\leq N \\ i\in\{j,k,\ell\} }} = \sum_{i= j<k<\ell\leq N}+\sum_{1\leq j<k=i<\ell\leq N}+\sum_{1\leq j<k<\ell=i\leq N}
	\]
	as well as $N/(N-2)\lesssim 1$ (for $N\geq3$), we obtain, for any $r,t>0$,
	\begin{multline*}
		\sum_{1\leq i \leq N} \left\{ h_i, \frac{a}{N-1} \sum_{\substack{1\leq j<k\leq N \\ i\in\{j,k\} }} U_N(x_j-x_k) - \frac{b}{(N-1)(N-2)} \sum_{\substack{1\leq j<k<\ell\leq N \\ i\in\{j,k,\ell\} }} W_N(x_j-x_k, x_j-x_\ell) \right\} \\
		\begin{aligned}[t]
			&=\sum_{1\leq i \leq N} \left\{ h_i, \frac{a}{N-1} \sum_{\substack{1\leq j\leq N \\ j \neq i }} U_N(x_i-x_j) - \frac{b}{(N-1)(N-2)} \sum_{\substack{1\leq j<k\leq N \\ j \neq i \neq k }} W_N(x_i-x_j, x_i-x_k) \right\} \\
			&\geq
			\begin{multlined}[t][0.85\textwidth]
				- \sum_{1\leq i \leq N} \left( \frac{2 |a|}{N-1} \left( \norm{U_N}_1 + \norm{U_N}_2 \right) \sum_{\substack{1\leq j\leq N \\ j \neq i }} h_i h_j \right. \\
				\left. + \frac{b}{(N-1)(N-2)} C \left(\norm{W_N}_{1+r} + \norm{\nabla_{1} W_N}_{1+t}\right) \sum_{\substack{1\leq j<k\leq N \\ j,k \neq i}} h_i (h_j + h_k) \right)
			\end{multlined} \\
			&= - \left( \frac{2 |a|}{N-1} \left( 1 + N^\alpha \norm{U}_2 \right) + C\frac{b}{N-1} \left(N^{\frac{2r}{1+r} \beta}\norm{W}_{1+r} + N^{\beta + \frac{2t}{1+t} \beta}\norm{\nabla_{1} W}_{1+t}\right) \right) \sum_{1\leq i \neq j\leq N} h_i h_j \\
			&\geq - \left( \frac{2 |a|}{N-1} \left( 1 + N^\alpha \norm{U}_2 \right) + C\frac{b}{N-1} \left(N^{\frac{2r}{1+r} \beta}\norm{W}_{1+r} + N^{\beta + \frac{2t}{1+t} \beta}\norm{\nabla_{1} W}_{1+t}\right) \right) \biggl( \sum_{j=1}^N h_j \biggr)^2.
		\end{aligned}
	\end{multline*}
	Restricting to $r>1$ and choosing $t = (r-1)/(r+3)\in(0,1)$ so that the two powers of $N$ in the factor of $b$ are equal, we have
	\begin{multline}\label{Proof_Lemma_Moments_Estimates_second_split2}
		\sum_{1\leq i \leq N} \left\{ h_i, \frac{a}{N-1} \sum_{\substack{1\leq j<k\leq N \\ i\in\{j,k\} }} U_N(x_j-x_k) - \frac{b}{(N-1)(N-2)} \sum_{\substack{1\leq j<k<\ell\leq N \\ i\in\{j,k,\ell\} }} W_N(x_j-x_k, x_j-x_\ell) \right\} \\
		\geq - \left( 4 |a| \left( N^{-1} + N^{\alpha-1} \norm{U}_2 \right) + C_r bN^{\frac{2r}{1+r}\beta-1} \right) \biggl( \sum_{j=1}^N h_j \biggr)^2,
	\end{multline}
	where, recalling the assumption $\nabla_{1}W \in L^q(\R^2)$ for $q>1$ and choosing a $q\in(1,2)$,
	\[
		C_r := C \left(\norm{W}_{1+r} + \norm{\nabla_{1} W}_{2\frac{r+1}{r+3}}\right) = C \left(\norm{W}_{1+r} + \norm{\nabla_{1} W}_q\right) < +\infty
	\]
	for the choice $r=(3q-2)/(2-q)>1$.
	
	Inserting~\eqref{Proof_Lemma_Moments_Estimates_second_split1} and~\eqref{Proof_Lemma_Moments_Estimates_second_split2} into~\eqref{Proof_Lemma_Moments_Estimates_second_initial_identity} yields
	\begin{multline}
		\left\{ \sum_{j=1}^N h_j , H_{a, b}^{N} \right\} \geq 2 N E_{a, b}^{\mathrm{Q}, N, \eps} \biggl( \sum_{j=1}^N h_j \biggr) \\
		+ \left( 2 \eps - 4 |a| \left( 2 N^{-1} + N^{\alpha-1} \norm{U}_2 \right) - C_r b N^{\frac{2r}{1+r} \beta - 1} \right) \biggl( \sum_{j=1}^N h_j \biggr)^2\,.
	\end{multline}	
	Hence, using the trivial bound $E_{a, b}^{\mathrm{Q}, N, \eps}\geq -\left|E_{a, b}^{\mathrm{Q}, N, \eps}\right|$ then~\eqref{Proof_Lemma_Moments_Estimates_first}, we obtain
	\begin{multline}\label{Proof_Lemma_Moments_Estimates_second_lowerbd}
		\pscal{\Psi_N, \left\{ \sum_{j=1}^N h_j , H_{a, b}^{N} \right\} \Psi_N}\\
			\begin{aligned}[b]
				&\geq
				\begin{multlined}[t][0.8\textwidth]
					-2 N^2 \left|E_{a, b}^{\mathrm{Q}, N, \eps}\right| \pscal{\Psi_N, h_{1} \Psi_N} \\
					+ \left( 2 \eps - 4 |a| \left( 2 N^{-1} + N^{\alpha-1} \norm{U}_2 \right) - C_r b N^{\frac{2r}{1+r} \beta - 1} \right) \pscal{\Psi_N, \biggl( \sum_{j=1}^N h_j \biggr)^2 \Psi_N}
				\end{multlined} \\
				&\geq
				\begin{multlined}[t][0.8\textwidth]
					- C N^2 \frac{\left(1 + \left|E_{a, b}^{\mathrm{Q}, N, \eps}\right|\right)^2}{\eps} \\
					+ \left( 2 \eps - 4 |a| \left( 2 N^{-1} + N^{\alpha-1} \norm{U}_2 \right) - C_r b N^{\frac{2r}{1+r} \beta - 1} \right) \pscal{\Psi_N, \biggl( \sum_{j=1}^N h_j \biggr)^2 \Psi_N}\,.
				\end{multlined}
			\end{aligned}
	\end{multline}
	
	For the upper bound on $\left\{ \sum_{j=1}^N h_j , H_{a, b}^{N} \right\}$, using that $\Psi_N$ is a ground state, we write
	\begin{align*}
		\pscal{\Psi_N, \left\{ \sum_{j=1}^N h_j , H_{a, b}^{N} \right\} \Psi_N} = 2 N E_{a, b}^{\mathrm{Q}, N} \pscal{\Psi_N, \biggl( \sum_{j=1}^N h_j \biggr) \Psi_N} &= 2 N^2 E_{a, b}^{\mathrm{Q}, N} \pscal{\Psi_N, h_{1} \Psi_N} \\
		&\leq C N^2 \frac{\left(1 + \left|E_{a, b}^{\mathrm{Q}, N, \eps}\right|\right)^2}{\eps}\,.
	\end{align*}
	Here we used again~\eqref{Proof_Lemma_Moments_Estimates_first} and $\left|E_{a, b}^{\mathrm{Q}, N}\right| \leq C \left( 1+ \left|E_{a, b}^{\mathrm{Q}, N, \eps}\right| \right)$.
	
	Finally, since $(r+1)/(2r)=q/(3q-2)\in(1/2, 1)$ can be made as close to $1$ as needed (because $q$ can be taken as close as $1$ as needed by assumption on~$\nabla_{1}W$) and since $\beta<1$, we can assume $2r\beta<r+1$. Consequently $2 \eps - 4 |a| \left( 2 N^{-1} + N^{\alpha-1} \norm{U}_2 \right) - C_r b N^{\frac{2r}{1+r} \beta - 1} \geq \eps>0$ for $N$ large enough, since we also assume $\alpha<1$, and the above inequality, together with~\eqref{Proof_Lemma_Moments_Estimates_second_initial_ineq} and~\eqref{Proof_Lemma_Moments_Estimates_second_lowerbd}, gives the wanted result
	\[
		\tr\left[ h\otimes h \gamma^{(2)}_{\Psi_N} \right] \leq C \left( \frac{1+\left|E_{a, b}^{\mathrm{Q}, N, \eps}\right|}\eps \right)^2. \qedhere
	\]
\end{proof}

We now turn to the proof of the energy lower bound~\eqref{lowerbound:information_deF} in Lemma~\ref{lem:lowerbound_deF}.
\begin{proof}[Proof of Lemma~\ref{lem:lowerbound_deF}]
	Again, we use the decomposition
	\begin{equation}\label{lowerbound:deF_decomposition_noncritical}
		\frac{ \pscalSM{ \Psi_N, H_{a, b}^{N} \Psi_N } }{N} = \frac{1}{3} \tr\left[ H_3P^{\otimes3}\gamma_{\Psi_N}^{(3)}P^{\otimes3} \right] + \frac{1}{3} \tr\left[ H_3\left( \gamma_{\Psi_N}^{(3)}-P^{\otimes3}\gamma_{\Psi_N}^{(3)}P^{\otimes3} \right] \right).
	\end{equation}
	For the main term in~\eqref{lowerbound:deF_decomposition_noncritical}, we define
	\[
		\widetilde{\gamma_{\Psi_N}}^{(3)} := \int_{\mathcal{S}_P} \gamma^{\otimes 3} \di\mu_{\Psi_N}(\gamma)\,.
	\]
	By decomposing the interaction term as announced in~\eqref{eq:decomp_U_noncritical} and~\eqref{eq:decomp_W_noncritical}, using the triangle inequality with~\eqref{ineq:information_deF} and recalling~\eqref{CwikelLiebRosenblum_type_onedimensional}, we obtain
	\begin{equation}\label{main_ineq_error_proof_lemma_lowerbound_deF}
		\frac{1}{3} \tr\left[ H_3P^{\otimes3}\gamma_{\Psi_N}^{(3)}P^{\otimes3} \right] \geq \frac{1}{3} \tr\left[ H_3P^{\otimes3}\widetilde{\gamma_{\Psi_N}}^{(3)}P^{\otimes3} \right] - \mathcal{R} = \int_{\mathcal{S}_P} \mathcal{E}_{a, b}^{\mathrm{mH}, N}(\gamma) \di\mu_{\Psi_N}(\gamma) - \mathcal{R} \,.
	\end{equation}
	Here, in order to control the error term
	\begin{multline*}
		\mathcal{R} := C \sqrt{\frac{\log L}{N}} \Biggl(L + \sum_{\mathbf{e}_{k}\in \{\cos(k\cdot), \sin (k\cdot)\}} \int_{\R} \norm{P \mathbf{e}_{k} P}^2 |\widehat{U}(N^{-\alpha}k)| \di k \\
		+ \sum_{\mathbf{e}_{k}\in \{\cos(k\cdot), \sin (k\cdot)\}} \iint_{\R^2} \norm{P \mathbf{e}_{k_{1}} P}\norm{P \mathbf{e}_{k_{2}} P}\norm{P \mathbf{e}_{k_{1}+k_{2}} P} |\widehat{W}(N^{-\beta}k_{1},N^{-\beta}k_{2})| \di k_{1}\di k_{2} \Biggr),
	\end{multline*}
	we make use of the following lemma, which is the 1D analogue of~\cite[Lemma~3.4]{NamRou-20} and has a similar proof that we therefore omit it for brevity.
	\begin{lemma}[\textbf{Multiplication by plane waves}]\label{lem:cos_sin}\leavevmode\\
		Let $k\in \R$, $k\neq 0$ and $\mathbf{e}_{k}$ be the multiplication operator on $L^2 (\R)$ by either $\cos(k\cdot)$ or $\sin (k\cdot)$. Let $P$ be the spectral projector given by~\eqref{Def_projection_P}. As operators, we have
		\[
			\pm P \mathbf{e}_{k} P \leq \min\left\{ 1, \frac{C L^{\frac{1}{2}}}{|k|} \right\}.
		\]
	\end{lemma}	
	The first error term in $\mathcal{R}$, $L$, is the operator norm of $P^{\otimes 3} (h_{x_1}+h_{x_2}+h_{x_3}) P^{\otimes 3} /3$. The second error term, made of the two\nobreakdash-body, interaction can be estimated using Lemma~\ref{lem:cos_sin} and the condition~\eqref{condition:potential_two_body}, which ensures $\widehat{U} \in L^\infty(\R)$, as follows
	\begin{align*}
		\int_{\R} \norm{P \mathbf{e}_{k} P}^2 |\widehat{U}(N^{-\alpha}k)| \di k &\leq \int_{\R} \min\left\{1,\frac{CL}{|k|^{2}}\right\} |\widehat{U}(N^{-\alpha}k)| \di k \\
		&\leq \normSM{\widehat{U}}_{\infty} \left(\int_{|k|\leq 1}\di k + \int_{|k|>1}\frac{CL}{|k|^{2}}\di k\right) \leq C + CL.
	\end{align*}
	Similarly, we use again Lemma~\ref{lem:cos_sin} and the condition~\eqref{condition:potential_three_body}, which ensures $\widehat{W} \in L^2\cap L^\infty(\R^2)$, to estimate the third error term in $\mathcal{R}$ made of the three\nobreakdash-body interaction. We have
	\begin{align*}
		&\iint_{\R^2} \norm{P \mathbf{e}_{k_{1}} P}\norm{P \mathbf{e}_{k_{2}} P}\norm{P \mathbf{e}_{k_{1}+k_{2}} P} |\widehat{W}(N^{-\beta}k_{1},N^{-\beta}k_{2})| \di k_{1}\di k_{2} \\
		\leq{}& \iint_{\R^2} \min\left\{1, \frac{CL^{\frac{1}{2}}}{ |k_{1}|} \right\} \min\left\{1, \frac{CL^{\frac{1}{2}}}{|k_{2}|}\right\} \min\left\{1, \frac{CL^{\frac{1}{2}}}{|k_{1} + k_{2}|}\right\} |\widehat{W}(N^{-\beta}k_{1},N^{-\beta}k_{2})| \di k_{1}\di k_{2} \\
		\leq{}& \begin{multlined}[t]
			\int_{\substack{|k_{1}|\leq 1 \\ |k_{2}|\leq 1}} \normSM{\widehat{W}}_{\infty} \di k_{1}\di k_{2} + C\int_{\substack{1 < |k_{1}|\leq 4N^{\beta} \\ |k_{2}|\leq 1}} L^{\frac{1}{2}} \frac{ \normSM{\widehat{W}}_{\infty} }{ |k_{1}| } \di k_{1}\di k_{2} + C\int_{\substack{|k_{1}| \leq 1 \\ 1 < |k_{2}|\leq 2N^{\beta}}} L^{\frac{1}{2}} \frac{ \normSM{\widehat{W}}_{\infty} }{ |k_{2}| } \di k_{1}\di k_{2} \\
			+ C\int_{\substack{1 < |k_{1}| \leq 4N^{\beta} \\ 1 < |k_{2}|\leq 2N^{\beta}}} L \frac{ \normSM{\widehat{W}}_{\infty} }{ |k_{1}| |k_{2}| } \di k_{1}\di k_{2} + C\int_{\substack{|k_{1}| > N^{\beta} \\ |k_{2}| > 2N^{\beta}}} L \frac{ |\widehat{W}(N^{-\beta}k_{1},N^{-\beta}k_{2})| }{ |k_{1}| |k_{2}| } \di k_{1}\di k_{2} \\
			+ C\int_{\substack{|k_{1}| > 4N^{\beta} \\ |k_{2}|\leq 2N^{\beta}}} L \frac{ |\widehat{W}(N^{-\beta}k_{1},N^{-\beta}k_{2})| }{ |k_{1}| |k_{1}+k_{2}| } \di k_{1}\di k_{2} + C\int_{\substack{|k_{1}| \leq N^{\beta} \\ |k_{2}| > 2N^{\beta}}} L \frac{ |\widehat{W}(N^{-\beta}k_{1},N^{-\beta}k_{2})| }{ |k_{1}+k_{2}| |k_{2}| } \di k_{1}\di k_{2}
		\end{multlined} \\
		\leq{}& C + C L^{\frac{1}{2}} \log N + C L \log N + C L \,.
	\end{align*}
	The error terms $\mathcal{R}$ in~\eqref{main_ineq_error_proof_lemma_lowerbound_deF} give the first error term on the lower bound of the quantum energy in~\eqref{lowerbound:information_deF}.
	
	Now we deal with the error in~\eqref{lowerbound:deF_decomposition_noncritical}. By the same arguments as in the proof of Theorem~\ref{thm:many_body_critical}, we have
	\begin{multline*}
		\tr\left[\left( h \otimes \1 \otimes \1 + \1 \otimes h \otimes \1 + \1 \otimes \1 \otimes h \vphantom{\frac{1}{2}}+ \frac{1}{2}\left( U_{N} \otimes \1 + \1 \otimes U_{N} \right) \right)\left( \gamma_{\Psi_{N}}^{(3)}-P^{\otimes3}\gamma_{\Psi_{N}}^{(3)}P^{\otimes3} \right)\right] \\
		\geq -CL^{-\frac{1}{2}}\tr\left[h\gamma_{\Psi_N}^{(1)}\right].
	\end{multline*}
	Furthermore, we recall from~\eqref{lowerbound:energy_error_quantitative_1} that
	\begin{multline}\label{Proof_Lemma_Lower_bound_ground_state_energy_2nd_error_term_step1}
		\tr\left[ W_N\left( \gamma_{\Psi_N}^{(3)}-P^{\otimes3}\gamma_{\Psi_N}^{(3)}P^{\otimes3} \right) \right] \geq - \sqrt{2} \left(\tr\left[ W_N\left( \gamma_{\Psi_N}^{(3)}+P^{\otimes3}\gamma_{\Psi_N}^{(3)}P^{\otimes3} \right) \right]\right)^{\frac{1}{2}} \\
		\times \left(\tr\left[ W_N\left( 1-P^{\otimes3} \right)\gamma_{\Psi_N}^{(3)}\left( 1-P^{\otimes3} \right) \right] \right)^{\frac{1}{2}}\,.
	\end{multline}
	For the first trace under the square root in~\eqref{Proof_Lemma_Lower_bound_ground_state_energy_2nd_error_term_step1}, we use~\eqref{W_operator_bounds_three_R_body_hx_hy_L1_delta} to obtain
	\[
		W_N(x-y,x-z) \leq C_\delta (h_y h_z)^{\frac12+\delta} \leq C_\delta\left( \left( \frac12+\delta \right) r^{-1} h_y h_z+\left( \frac12-\delta \right)r^{\frac{1+2\delta}{1-2\delta}} \right)
	\]
	for any $\delta\in(0, 1/2)$ and any $r>0$, where we have applied to $t=h_y h_z$ and $\kappa=1/2+\delta$ the identity
	\begin{equation}\label{upperbound_t_power_r}
		t^\kappa=\inf_{r>0}\left( \kappa t r^{-1} + (1-\kappa) r^{\frac{\kappa}{1-\kappa}} \right), \quad \forall\, t\geq0, \, \forall\, \kappa\in(0,1) \,.
	\end{equation}
	Taking the trace against $\gamma_{\Psi_N}^{(3)}+P^{\otimes3}\gamma_{\Psi_N}^{(3)}P^{\otimes3}\leq2\gamma_{\Psi_N}^{(3)}$ then optimizing over $r>0$, we find that, for all $\delta\in(0,1/2)$,
	\begin{align}\label{Proof_Lemma_Lower_bound_ground_state_energy_2nd_error_term_1st_half}
		& \tr\left[ W_N(x-y,x-z)\left( \gamma_{\Psi_N}^{(3)}+P^{\otimes3}\gamma_{\Psi_N}^{(3)}P^{\otimes3} \right) \right] \nonumber \\
		\leq{}& 2 C_\delta \left\{ \left( \frac12+\delta \right) r^{-1} \tr\left[ \1\otimes h\otimes h\gamma_{\Psi_N}^{(3)} \right] + \left( \frac12-\delta \right) r^{\frac{1+2\delta}{1-2\delta}}\tr\left[ \gamma_{\Psi_N}^{(3)} \right] \right\} \nonumber \\
		={}& 2 C_\delta \left[\tr\left[ h\otimes h\gamma_{\Psi_N}^{(2)} \right]\right]^{\frac12+\delta}\,.
	\end{align}
	For the second trace under the square root in~\eqref{Proof_Lemma_Lower_bound_ground_state_energy_2nd_error_term_step1}, using the relations $Qh^\kappa P=Ph^\kappa Q=0$, $Qh^\kappa\leq L^{\kappa-1}h$, and $Ph^\kappa\leq L^\kappa P$ at $\kappa=1/2+\delta\in(0,1)$, as well as $Q\leq L^{-1}h$ and $P\leq1$, and again~\eqref{W_operator_bounds_three_R_body_hx_hy_L1_delta} and~\eqref{upperbound_t_power_r}, we obtain
	\begin{align*}
		& \left( 1-P^{\otimes3} \right)W_N(x-y,x-z)\left( 1-P^{\otimes3} \right) \\
		\leq{}& C_\delta \left( 1-P^{\otimes3} \right)(\1\otimes h\otimes h)^{\frac12+\delta}\left( 1-P^{\otimes3} \right) \\
		={}& C_\delta\left( \1\otimes h^{\frac12+\delta}\otimes Q h^{\frac12+\delta}Q + \1\otimes Q h^{\frac12+\delta}Q\otimes P h^{\frac12+\delta}P + QQ\otimes P h^{\frac12+\delta}P\otimes P h^{\frac12+\delta}P \right) \nonumber \\
		\leq{}& \frac{C_\delta}{L^{\frac12-\delta}}\left( \1\otimes h^{\frac12+\delta}\otimes h+\1\otimes h\otimes h^{\frac12+\delta}+h\otimes \1\otimes h^{\frac12+\delta} \right).
	\end{align*}
	Here, we can again use~\eqref{upperbound_t_power_r} to write, for any $\kappa\in(0,1)$ and any $r>0$,
	\begin{multline*}
		\1\otimes h^\kappa\otimes h+\1\otimes h\otimes h^\kappa+h\otimes \1\otimes h^\kappa \leq \kappa r^{-1}(\1\otimes h\otimes h+\1\otimes h\otimes h+h\otimes \1\otimes h) \\
		+(1-\kappa) r^{\frac{\kappa}{1-\kappa}}(\1\otimes \1\otimes h+\1\otimes h\otimes \1+h\otimes \1\otimes \1)\,.
	\end{multline*}
	Using it at $\kappa=1/2+\delta\in(0,1)$, taking the trace against $\gamma_{\Psi_N}^{(3)}$, and optimizing over $r>0$ leads to
	\begin{align}\label{lowerbound:energy_error_quantitative}
		\tr\left[ W_N\left( 1-P^{\otimes3} \right)\gamma_{\Psi_N}^{(3)}\left( 1-P^{\otimes3} \right) \right] \leq{}& \frac{C_\delta}{L^{\frac12-\delta}} \left( \kappa r^{-1} \tr\left[ h\otimes h\gamma_{\Psi_N}^{(2)} \right] + (1-\kappa) r^{\frac{\kappa}{1-\kappa}} \tr\left[ h\gamma_{\Psi_N}^{(1)} \right]\right) \nonumber \\
		\leq{}& \frac{C_\delta}{L^{\frac12-\delta}} \left( \tr\left[ h\otimes h\gamma_{\Psi_N}^{(2)} \right] \right)^{\frac12+\delta} \left( \tr\left[ h\gamma_{\Psi_N}^{(1)} \right] \right)^{\frac12-\delta}
	\end{align}
	for any $\delta\in(0,1/2)$. Finally inserting~\eqref{Proof_Lemma_Lower_bound_ground_state_energy_2nd_error_term_1st_half} and~\eqref{lowerbound:energy_error_quantitative} into~\eqref{Proof_Lemma_Lower_bound_ground_state_energy_2nd_error_term_step1}, we obtain
	\[
		\tr\left[ W_N\left( \gamma_{\Psi_N}^{(3)}-P^{\otimes3}\gamma_{\Psi_N}^{(3)}P^{\otimes3} \right) \right] \geq -C_{\delta}L^{-\frac{1}{4}+\frac{1}{2}\delta}\left(\tr\left[ h\gamma_{\Psi_N}^{(1)} \right]\right)^{\frac14-\frac\delta2} \left(\tr\left[ h\otimes h\gamma_{\Psi_N}^{(2)} \right]\right)^{\frac12+\delta}\,.
	\]
	This yields the second error term on the lower bound of the quantum energy in~\eqref{lowerbound:information_deF}.
\end{proof}

According to Lemmas~\ref{lem:lowerbound_deF} and~\ref{Lemma_Moments_Estimates}, we have
\begin{equation}\label{ineq:first_step_in_proof_thm_many_body_noncritical}
	E_{a, b}^{\mathrm{Q}, N} \geq \int_{\mathcal{S}_P} \mathcal{E}_{a, b}^{\mathrm{mH}, N}(\gamma) \di\mu_{\Psi_N}(\gamma) - C\sqrt{\frac{\log L}{N}} L \log N - C_{\delta}L^{-\frac{1}{4}+\frac{1}{2}\delta}\left( \frac{1+\left|E_{a, b}^{\mathrm{Q}, N, \eps}\right|}\eps \right)^{\frac{5}{4}+\frac{3}{2}\delta}
\end{equation}
for all $0<\delta \leq 1/2$ and $0<\eps<1$. On one hand, using the convexity of the kinetic energy~\cite[Theorem~7.8]{LieLos-01}, we have the Hoffmann-Ostenhof-type inequality
\[
	\tr\left[ h\gamma \right] \geq \int_\R |\nabla \sqrt{\rho_{\gamma}}|^2\,.
\]
Then, it follows from~\eqref{ineq_cv:Hartree_to_NLS_2body}, \eqref{ineq_cv:Hartree_to_NLS_3body}, and~\eqref{ineq:GNS} that
\begin{equation}\label{ineq:energy_Hartree_larger_constant}
	\mathcal{E}_{a, b}^{\mathrm{mH}, N}(\gamma) \geq \mathcal{E}_{a, b}^{\mathrm{H}, N}(\sqrt{\rho_{\gamma}}) \geq -C \,.
\end{equation}
On another hand, we note that the error terms in~\eqref{ineq:first_step_in_proof_thm_many_body_noncritical} are exactly the same as in~\cite{NamRou-20}. By a bootstrap argument as in~\cite{LewNamRou-17,NamRou-20}, we deduce that $\left|E_{a, b}^{\mathrm{Q}, N, \eps}\right|$ is bounded independently of $N$. We then deduce from the upper bound on $E_{a, b}^{\mathrm{Q}, N}$ in~\eqref{main_ineq_proof_thm_many_body_critical}, from~\eqref{ineq:first_step_in_proof_thm_many_body_noncritical}, and  from~\eqref{ineq:energy_Hartree_larger_constant} that
\begin{align*}
	E_{a, b}^{\mathrm{H}, N} \geq E_{a, b}^{\mathrm{Q}, N} &\geq \int_{\mathcal{S}_P} \mathcal{E}_{a, b}^{\mathrm{mH}, N}(\gamma) \di\mu_{\Psi_N}(\gamma) - C_\delta\left( \sqrt{\frac{\log L}{N}} L \log N + L^{-\frac{1}{4}+\frac{1}{2}\delta} \right) \\
		&\geq E_{a, b}^{\mathrm{H}, N} \int_{\mathcal{S}_P} \di\mu_{\Psi_N}(\gamma) - C_\delta\left( \sqrt{\frac{\log L}{N}} L \log N + L^{-\frac{1}{4}+\frac{1}{2}\delta} \right)\,,
\end{align*}
for all $0<\delta \leq 1/2$, $L \gg 1$ and $N\geq2$. Choosing optimally $L=N^{2/5}$, we obtain
\begin{equation}\label{main_ineq_proof_thm_many_body_noncritical}
	E_{a, b}^{\mathrm{H}, N} \geq E_{a, b}^{\mathrm{Q}, N} \geq \int_{\mathcal{S}_P}\mathcal{E}_{a, b}^{\mathrm{mH}, N}(\gamma) \di\mu_{\Psi_N} (\gamma) - C_\eta N^{-\sigma} \geq E_{a, b}^{\mathrm{H}, N} \int_{\mathcal{S}_P} \di\mu_{\Psi_N}(\gamma) - C_\eta N^{-\sigma}\,,
\end{equation}
for all $0<\sigma<1/10$.

Next, we deduce from~\eqref{Moments_Estimates} and the boundedness of $\left|E_{a, b}^{\mathrm{Q}, N, \eps}\right|$, for small enough $0 < \eps < 1$, that
\[
	\tr\left[ h \gamma_{\Psi_N}^{(1)} \right] \leq C_{\eps}\,.
\]
Since $h$ has compact resolvent we deduce (modulo subsequence) that $\gamma_{\Psi_N}^{(1)} \to \gamma^{(1)}$ strongly in trace-class, for some limit one-body bosonic density matrix $\gamma^{(1)}$. But we also have (again, modulo subsequences)
\[
	\gamma_{\Psi_N}^{(k)} \underset{\star}{\wto} \gamma^{(k)}, \quad \forall\, k=1,2,\dots
\]
weakly-$\star$ in the trace-class. Applying the weak quantum de~Finetti theorem~\cite[Theorem~2.2]{LewNamRou-14} we deduce that there exists a measure $\nu$ on the unit ball of $L^2 (\R)$ such that
\[
	\gamma^{(k)} = \int |u^{\otimes k } \rangle \langle u^{\otimes k}| \di\nu (u), \quad \forall\, k=1,2,\dots
\]
But since $\gamma^{(1)}$ must have trace $1$, the measure $\nu$ must actually live on the unit sphere of $L^2 (\R)$, i.e.,
\[
	S L^2 (\R) = \left\{ u\in L^2 (\R),\: \int_\R |u|^2 = 1 \right\}.
\]
By the choice of $L$, \eqref{measure_information_deF} implies that
\[
	\lim_{N\to\infty}\int_{\mathcal{S}_P} \di\mu_{\Psi_N}(\gamma) = 1\,.
\]
Thus, the sequence of measures $\{\mu_{\Psi_N}\}_N$ given by Theorem~\ref{thm:information_deF} is tight on the set of one-body mixed states
\[
	\mathcal{S} := \left\{ \gamma \mbox{ positive trace-class self-adjoint operator on } L^2 (\R), \, \tr \gamma = 1 \right\}.
\]
Modulo a subsequence, $\{\mu_{\Psi_N}\}_N$ converges to a measure $\mu$.

Next we claim that the two measures $\mu$ and $\nu$ just found are related by
\begin{equation}\label{eq:measure_same}
	 \int_{\mathcal{S}} \gamma^{\otimes 3} \di\mu (\gamma) = \int |u^{\otimes 3} \rangle \langle u^{\otimes 3}| \di\nu(u)\,.
\end{equation}
Indeed, let $\tilde{P} = \1_{h\leq \tilde{L}}$ where $\tilde{L}$ is a fixed cut-off (different from $L$ above). Testing~\eqref{ineq:information_deF} with $A,B,C$ finite rank operators whose ranges lie within that of $\tilde{P}$ we get
\[
	\tr\left[ A \otimes B \otimes C \gamma_{\Psi_N}^{(3)} \right] \underset{N\to\infty}{\longrightarrow} \tr\left[ A \otimes B \otimes C \int_{\mathcal{S}} \gamma^{\otimes 3} \di\mu (\gamma) \right]
\]
using the convergence of $\mu_{\Psi_N}$ to $\mu$. On the other hand, by the convergence of $\gamma_{\Psi_N}^{(3)}$ to $\gamma^{(3)}$ we also have
\[
	\tr\left[ A \otimes B \otimes C \gamma_{\Psi_N}^{(3)} \right] \underset{N\to\infty}{\longrightarrow} \tr\left[ A \otimes B \otimes C \gamma^{(3)} \right] = \tr\left[ A \otimes B \otimes C \int_{S L^2 (\R)} |u^{\otimes 3} \rangle \langle u^{\otimes 3}| \di\nu (u) \right].
\]
Thus,
\begin{equation}\label{eq:measure_trace_equality}
	\tr\left[ A \otimes B \otimes C \int_{\mathcal{S}} \gamma^{\otimes 3} \di\mu (\gamma) \right] = \tr\left[ A \otimes B \otimes C \int_{S L^2 (\R)} |u^{\otimes 3} \rangle \langle u^{\otimes 3}| \di\nu (u) \right]
\end{equation}
for any $A,B,C$ with range within that of $\tilde{P}$. Letting finally $\tilde{L} \to \infty$ yields $\tilde{P} \to \1$ and thus~\eqref{eq:measure_trace_equality} holds for any compact operators $A,B,C$. This implies~\eqref{eq:measure_same}. In particular, since the right-hand side of~\eqref{eq:measure_same} is $\gamma^{(3)}$, a bosonic operator, $\mu$ must be supported on pure states $\gamma = |u\rangle \langle u|$, see~\cite{HudMoo-76}.

Let us return to~\eqref{main_ineq_proof_thm_many_body_noncritical}. We split the integral over
one-body states $\gamma$ between low and high kinetic energy states
\[
	\mathrm{K}_- = \{\gamma\in \mathcal{S}: \tr[h\gamma] \leq C_{\mathrm{Kin}}\} \quad \text{and} \quad \mathrm{K}_+ = \mathcal{S} \setminus \mathrm{K}_-.
\]
with $C_{\mathrm{Kin}}$ a constant independent of $N$. Using~\eqref{ineq_cv:Hartree_to_NLS_2body} and~\eqref{ineq_cv:Hartree_to_NLS_3body}, we obtain, for every $\sigma< 1/10$, 
\begin{align}\label{final_main_ineq_proof_thm_many_body_noncritical}
	E_{a, b}^{\mathrm{H}, N} \geq E_{a, b}^{\mathrm{Q}, N} &\geq \int_{\mathcal{S}_P}\mathcal{E}_{a, b}^{\mathrm{mH}, N}(\gamma) \di\mu_{\Psi_N} (\gamma) - C_\eta N^{-\sigma} \nonumber \\
	&\geq C_{\mathrm{Kin}}\int_{\mathrm{K}_+}\di\mu_{\Psi_N} (\gamma) + \int_{\mathrm{K}_-}\mathcal{E}_{a, b}^{\mathrm{m}\NLS}(\gamma)\di\mu_{\Psi_N} (\gamma) - o(1) - C_\eta N^{-\sigma} \nonumber\\
	&\geq \int_{\mathcal{S}_P}\min\left\{C_{\mathrm{Kin}},\mathcal{E}_{a, b}^{\mathrm{m}\NLS}(\gamma)\right\}\di\mu_{\Psi_N} (\gamma) - o(1) - C_\eta N^{-\sigma}.
\end{align}
Here,
\[
	\mathcal{E}_{a, b}^{\mathrm{m}\NLS}(\gamma) := \tr\left[ h\gamma \right] +\frac{a}{2}\int_{\R}\rho_{\gamma}(x)^2\dix - \frac{b}{6} \int_{\R} \rho_{\gamma}(x)^{3}\dix \,.
\]
Passing to the limit $N\to \infty$ in~\eqref{final_main_ineq_proof_thm_many_body_noncritical}, using~\eqref{cv:Hartree_to_NLS_energy}, then taking $C_{\mathrm{Kin}} \to \infty$, we obtain
\begin{equation}\label{eq:final_lowerbound_noncritical}
	E_{a,b}^{\NLS} \geq \lim_{N\to \infty} E_{a, b}^{\mathrm{Q}, N} \geq \int_{\mathcal{S}} \mathcal{E}_{a, b}^{\mathrm{m}\NLS}(\gamma) \di\mu (\gamma) \,.
\end{equation}
But, as we saw above, $\mu$ must be supported on pure states $\gamma = |u\rangle \langle u |$, which yields both the energy lower bound and the fact that $\mu$ must be supported on $\mathcal{M}^{\NLS}$. Because $\mathcal{E}_{a, b}^{\mathrm{m}\NLS}(\gamma)$ is a linear function of $\gamma^{\otimes 3}$, we can also combine~\eqref{eq:final_lowerbound_noncritical} with~\eqref{eq:measure_same} to deduce that $\nu$ must also be supported on $\mathcal{M}^{\NLS}$, which proves~\eqref{cv:ground_state_quantum_to_NLS_noncritical}.

Let's complete the proof of Theorem~\ref{thm:many_body_noncritical} by proving~\eqref{asymptotic:quantum_energy_noncritical} and~\eqref{asymptotic:many_body_ground_state_noncritical}. The collapse of ground state energy in~\eqref{asymptotic:quantum_energy_noncritical} follows from~\eqref{main_ineq_proof_thm_many_body_noncritical} and Theorem~\ref{thm:collapse_Hartree}. From the proof of Lemma~\ref{Lemma_Moments_Estimates}, we have
\[
	\tr\left[ h \gamma^{(1)}_{\Psi_N} \right]\lesssim \frac{ E_{a_N, b_N}^{\mathrm{Q}, N} - E_{a_N, b_N}^{\mathrm{Q}, N, \eps} }{\eps} \qquad \textrm{ and } \qquad \tr\left[ h\otimes h \gamma^{(2)}_{\Psi_N} \right] \lesssim \left( \frac{ E_{a_N, b_N}^{\mathrm{Q}, N} - E_{a_N, b_N}^{\mathrm{Q}, N, \eps} }{\eps} \right)^2\,,
\]
Choosing $\eps = 1 - b_N/\mathfrak{b}$, using the nonnegativity of the trapping term and the asymptotic formula of the ground state energy in Theorem~\ref{thm:many_body_noncritical}, we obtain
\begin{equation}\label{moments_estimates_blow_up}
	\tr\left[h\gamma^{(1)}_{\Psi_N} \right]\lesssim (\mathfrak{b}-b_N)^{-\frac{2}{s+2}} \qquad \textrm{ and } \qquad \tr\left[ h\otimes h \gamma^{(2)}_{\Psi_N} \right] \lesssim (\mathfrak{b}-b_N)^{-\frac{4}{s+2}}.
\end{equation}
On the other hand, let $\tilde{E}_{a, b}^{\mathrm{Q}, N}$ be the associated ground state energy of the modified Hamiltonian
\[
	\tilde{H}^{N}_{a,b} := H^{N}_{a,b} - \frac{1}{2}\sum_{j=1}^N|x_j|^s\,,
\]
i.e.,
\[
	\tilde{E}_{a, b}^{\mathrm{Q}, N} := N^{-1} \inf \left\{ \pscal{\Psi_N, \tilde{H}^{N}_{a,b} \Psi_N}:\Psi_N\in\mathfrak{H}^N, \norm{\Psi_N}_{2}=1\right\}.
\]
By variational principle, we have
\[
	\tr\left[|\cdot|^s\gamma^{(1)}_{\Psi_N} \right]\lesssim E_{a_N, b_N}^{\mathrm{Q}, N} - \tilde{E}_{a_N, b_N}^{\mathrm{Q}, N}\,.
\]
We use again the asymptotic formula of the ground state energy in Theorem~\ref{thm:many_body_noncritical}, we obtain
\begin{equation}\label{trapping_estimates_blow_up}
	\tr\left[|\cdot|^s\gamma^{(1)}_{\Psi_N} \right]\lesssim (\mathfrak{b}-b_N)^{\frac{s}{s+2}}\,.
\end{equation}
Inserting~\eqref{moments_estimates_blow_up} into Lemma~\ref{lem:lowerbound_deF} and using~\eqref{ineq:energy_Hartree_larger_constant}, we obtain
\begin{align*}
	E_{a_N, b_N}^{\mathrm{H}, N} \geq{} & \frac{\pscalSM{\Psi_N, H_{a, b}^{N} \Psi_N}}{N} \\
	\geq{} & \int_{\mathcal{S}_P} \mathcal{E}_{a_N, b_N}^{\mathrm{H}, N}(\sqrt{\rho_{\gamma}}) \di\mu_{\Psi_N} (\gamma) - C \sqrt{\frac{\log L}{N}} L \log N - C_{\delta}L^{-\frac{1}{4}+\frac{1}{2}\delta} (\mathfrak{b}-b_N)^{-\frac{2}{s+2} \left( \frac{5}{4}+\frac{3\delta}{2} \right)}
\end{align*}
for all $0<\delta \leq 1/2$, $L \gg 1$ and $N\geq3$. Furthermore,
\[
	1\geq \int_{\mathcal{S}_P} \di\mu_{\Psi_N} (\gamma) \geq 1-3L^{-1}(\mathfrak{b}-b_N)^{-\frac{2}{s+2}}\,.
\]
Recalling the assumption $\zeta\neq 12(s+1)/s$, it is straightforward that, if additionally to~\eqref{collapse:speed_noncritical} we assume that $\eta<1/(10(s+2))$, then we can choose $L > 0$ appropriately such that
\begin{equation}\label{cv:energy_measure_noncritical}
	\lim_{N\to\infty}\int_{\mathcal{S}_P} \frac{\mathcal{E}_{a_N, b_N}^{{\rm H}, N}(\sqrt{\rho_\gamma})}{E_{a_N,b_{N}}^{{\rm H}, N}} \di\mu_{\Psi_N}(\gamma) = \lim_{N\to\infty} \int_{\mathcal{S}_P} \di\mu_{\Psi_{N}}(\gamma) = 1\,.
\end{equation}
Let $\Phi_N = \ell_N^{N/2}\Psi_{N}(\ell_{N}\cdot)$. Then, it follows from~\eqref{moments_estimates_blow_up} and~\eqref{trapping_estimates_blow_up} that
\[
	\tr\left[ h \gamma_{\Phi_N} ^{(1)} \right] \leq C\,.
\]
Since $h$ has compact resolvent, we deduce that, up to extraction of a subsequence, $\gamma_{\Phi_N}^{(1)}$ converges to $\gamma^{(1)}$ strongly in the trace class. Modulo a diagonal extraction argument, one can assume that the convergence is along the same subsequence. By~\cite[Corollary 2.4]{LewNamRou-14}, $\gamma_{\Phi_N}^{(k)}$ converges to $\gamma^{(k)}$ strongly as well for all $k \geq 1$. By the quantum de~Finetti Theorem, there exists a Borel probability measure $\nu$ on the unit sphere $\mathcal{S}\mathfrak{H}$ such that
\[
	\gamma ^{(k)} = \int |u ^{\otimes k } \rangle \langle u ^{\otimes k}| \di\nu (u), \quad \forall\, k=1,2,\dots\,.
\]
Now, let us define
\[
	\overline{P} = \1(\overline{h} \leq L) \quad \text{with} \quad \overline{h} = \ell_{N}^{-2}\frac{{\rm d}^2}{{\rm d}x^2} + \ell_{N}^{s}|x|^s.
\]
On one hand, it follows from~\eqref{cv:energy_measure_noncritical} that
\begin{equation}\label{cv:energy_measure_scaled_noncritical}
	\lim_{N\to\infty}\int_{\mathcal{S}_{\overline{P}}} \frac{\mathcal{E}_{a_N, b_N}^{{\rm H}, N}\!\left( \ell_{N}^{-1/2}\sqrt{\rho_\gamma}(\ell_{N}^{-1}\cdot) \right)}{E_{a_N,b_{N}}^{{\rm H}, N}} \di\mu_{\Phi_N}(\gamma) = \lim_{N\to\infty}\int_{\mathcal{S}_{\overline{P}}} \di\mu_{\Phi_N}(\gamma) = 1.
\end{equation}
On another hand, we deduce from~\eqref{ineq:information_deF} that
\begin{equation}\label{ineq:information_deF_noncritical}
	\lim_{N\to\infty}\tr\left| \overline{A} \otimes \overline{B} \otimes \overline{C} \left( \overline{P} ^{\otimes 3} \gamma_{\Phi_N}^{(3)} \overline{P} ^{\otimes 3} - \int_{\mathcal{S}_{\overline{P}}} \gamma ^{\otimes 3} \di\mu_{\Phi_N} (\gamma) \right)\right| = 0.
\end{equation}
for every bounded operators $0\leq \overline{A},\overline{B},\overline{C}\leq 1$ on $\overline{P}\mathfrak{H}$.

To complete the proof, it suffices to prove the convergence of one-body density matrix
\begin{equation}\label{cv:many_body_pre_final_noncritical}
	\lim_{N\to\infty}\tr\left||Q_0^{\otimes 3}\rangle \langle Q_0^{\otimes 3}| - \int_{\mathcal{S}_{\overline{P}}} \gamma ^{\otimes 3} \di\mu_{\Phi_N} (\gamma)\right| = 0,
\end{equation}
which is equivalent to
\begin{equation}\label{cv:many_body_final_noncritical}
	\lim_{N\to\infty}\tr\int_{\mathcal{S}_{\overline{P}}} \left|\pscal{ Q_0,\sqrt{\rho_\gamma} }\right|^{3} \di\mu_{\Phi_N} (\gamma) = 1
\end{equation}
Before proving the above, we explain how to conclude the proof of~\eqref{asymptotic:many_body_ground_state_noncritical}. Let $\tilde{P} = \1(\overline{h}\leq \tilde{L})$ where $\tilde{L}$ is a fixed cut-off (different from $L$ above). Testing~\eqref{ineq:information_deF_noncritical} with $\overline{A},\overline{B},\overline{C}$ finite rank operators whose ranges lie within that of $\tilde{P}$ and combining with~\eqref{cv:many_body_pre_final_noncritical} we get
\[
	\tr\left[ \overline{A} \otimes \overline{B} \otimes \overline{C} \gamma_{\Phi_N} ^{(3)} \right] \underset{N\to\infty}{\longrightarrow} \tr\left[ \overline{A} \otimes \overline{B} \otimes \overline{C} |Q_0^{\otimes 3}\rangle \langle Q_0^{\otimes 3}| \right].
\]
On the other hand, by the convergence of $\gamma_{\Psi_N} ^{(3)}$ to $\gamma^{(3)}$ we also have
\[
	\tr\left[ \overline{A} \otimes \overline{B} \otimes \overline{C} \gamma_{\Phi_N} ^{(3)} \right] \underset{N\to\infty}{\longrightarrow} \tr\left[ \overline{A} \otimes \overline{B} \otimes \overline{C} \gamma^{(3)} \right].
\]
Thus,
\[
	\tr\left[ \overline{A} \otimes \overline{B} \otimes \overline{C} |Q_0^{\otimes 3}\rangle \langle Q_0^{\otimes 3}| \right] = \tr\left[ \overline{A} \otimes \overline{B} \otimes \overline{C} \gamma^{(3)} \right]
\]
for any $\overline{A}, \overline{B}, \overline{C}$ with range within that of $\tilde{P}$. Letting finally $\tilde{L} \to \infty$ yields $\tilde{P} \to \1$ and thus the above holds for any compact operators $\overline{A}, \overline{B}, \overline{C}$. This implies $\gamma^{(3)} \equiv |Q_0^{\otimes 3}\rangle \langle Q_0^{\otimes 3}|$ and concludes the proof of the strong convergence $\gamma_{\Phi_N}^{(3)} \to |Q_0^{\otimes 3}\rangle \langle Q_0^{\otimes 3}|$ in the trace class, which in turn yields the convergence $\gamma_{\Phi_N}^{(1)} \to |Q_0 \rangle \langle Q_0|$. The limit being rank $1$, this implies the convergence of higher order density matrices to tensor powers of the limiting operator by well-known arguments (see e.g., \cite[Section 2.2]{Rougerie-20b}).

We now come back to the proof of~\eqref{cv:many_body_final_noncritical}.
Defining
\[
	\delta_{N} := \int_{\mathcal{S}_{\overline{P}}} \frac{ \mathcal{E}_{a_N, b_N}^{{\rm H}, N}\!\left( \ell_{N}^{-1/2}\sqrt{\rho_\gamma}(\ell_{N}^{-1}\cdot) \right) - E_{a_N,b_{N}}^{{\rm H}, N} }{ \left| E_{a_N,b_{N}}^{{\rm H}, N} \right| } \di\mu_{\Phi_N}(\gamma) \,,
\]
we have $\delta_{N} \geq 0$, by definition of $E_{a_N,b_{N}}^{{\rm H}, N}$, and $\delta_{N} \to 0$, by~\eqref{cv:energy_measure_scaled_noncritical} and since the sign of $E_{a_N,b_{N}}^{{\rm H}, N}$ is constant for $N$ large enough (by Theorem~\ref{thm:collapse_Hartree} and recall we assume $\zeta \ne 12(s+1)/s$). Let $T_N$ be the set of all positive trace-class self-adjoint operator $\gamma$ on $\overline{P}\mathfrak{H}$ with $\tr\gamma = 1$ satisfying
\begin{equation}\label{behavior_approximate_noncritical}
	0\leq \frac{ \mathcal{E}_{a_N, b_N}^{{\rm H}, N}\!\left( \ell_{N}^{-1/2}\sqrt{\rho_\gamma}(\ell_{N}^{-1}\cdot) \right) - E_{a_N,b_{N}}^{{\rm H}, N} }{ \left| E_{a_N,b_{N}}^{{\rm H}, N} \right| } \leq \sqrt{\delta_N} \,.
\end{equation}

We prove that we must have
\begin{equation}\label{limit:PsiQ_noncritical}
	\lim_{N\to\infty}\inf_{\gamma\in T_N} \left|\pscal{ \sqrt{\rho_{\gamma}},Q_0}\right| = 1 \,.
\end{equation}
If this was not the case, there would exist a sequence $\{\gamma_{N}\} \subset T_N$ such that
\begin{equation}\label{limit:PsiQ_fraud_noncritical}
	\limsup_{N\to\infty} \left|\pscal{ \sqrt{\rho_{\gamma_N}},Q_0 }\right| < 1 \,.
\end{equation}
Since $\gamma_{N} \in T_N$ and $\delta_{N}\to 0$, we would deduce from~\eqref{behavior_approximate_noncritical} that
\[
	\lim_{N\to\infty} \frac{\mathcal{E}_{a_N, b_N}^{{\rm H}, N}\!\left( \ell_{N}^{-1/2}\sqrt{\rho_{\gamma_N}}(\ell_{N}^{-1}\cdot) \right)}{E_{a_N,b_{N}}^{{\rm H}, N}} = 1 \,.
\]
That is, $\{\ell_{N}^{-1/2}\sqrt{\rho_{\gamma_N}}(\ell_{N}^{-1}\cdot)\}_N$ would be a sequence of approximate ground states and Theorem~\ref{thm:collapse_Hartree} would then imply
\[
	\lim_{N\to\infty} \left|\pscal{ \sqrt{\rho_{\gamma_N}},Q_0}\right| = 1 \,,
\]
contradicting~\eqref{limit:PsiQ_fraud_noncritical}. Hence, we have~\eqref{limit:PsiQ_noncritical}.

On the other hand, by the choice of $T_N$, we have
\[
	\frac{ \mathcal{E}_{a_N, b_N}^{{\rm H}, N}\!\left( \ell_{N}^{-1/2}\sqrt{\rho_\gamma}(\ell_{N}^{-1}\cdot) \right) - E_{a_N,b_{N}}^{{\rm H}, N} }{ \left| E_{a_N,b_{N}}^{{\rm H}, N} \right| } \geq \sqrt{\delta_N} \,.
\]
for any $\gamma\in {}^\complement T_N$. Therefore,
\[
	\delta_{N} \geq \int_{{}^\complement T_N} \frac{ \mathcal{E}_{a_N, b_N}^{{\rm H}, N}\!\left( \ell_{N}^{-1/2}\sqrt{\rho_\gamma}(\ell_{N}^{-1}\cdot) \right) - E_{a_N,b_{N}}^{{\rm H}, N} }{ \left| E_{a_N,b_{N}}^{{\rm H}, N} \right| } \di\mu_{\Phi_N}(\gamma) \geq \sqrt{\delta_N} \mu_{\Phi_N}({}^\complement T_N) \,.
\]
Thus, $\mu_{\Phi_N}({}^\complement T_N) \leq \sqrt{\delta_N} \to 0$ and~\eqref{cv:energy_measure_scaled_noncritical} then yields $\mu_{\Phi_N}(T_N) \to 1$. The latter convergence and~\eqref{limit:PsiQ_noncritical} imply that
\[
	\int_{\mathcal{S}_{\overline{P}}} \left|\pscal{ Q_0,\sqrt{\rho_\gamma} }\right|^{3} \di\mu_{\Phi_N} (\gamma) \geq \int_{T_N} \left|\pscal{ Q_0,\sqrt{\rho_\gamma} }\right|^3 \di\mu_{\Phi_N}(u) \geq \mu_{\Phi_N}(T_N)\inf_{\gamma\in T_N} \left|\pscal{ Q_0,\sqrt{\rho_{\gamma}} }\right|^3 \to 1 \,.
\]
Thus~\eqref{cv:many_body_final_noncritical} holds true, concluding the proof of Theorem~\ref{thm:many_body_noncritical}. \qed

\appendix
\renewcommand{\thesection}{\Alph{section}}
\section{A property of the quintic \NLS solution}\label{app:techn_NLS}
\makeatletter
	\renewcommand{\thetheorem}{\thesection.\arabic{theorem}}
	\@addtoreset{theorem}{section}
	
	\renewcommand{\theequation}{\thesection.\arabic{equation}}
	\@addtoreset{equation}{section}
\makeatother
This appendix is dedicated to the proof of the inequality
\begin{equation}\label{ineq:techn_NLS_strict_symmetry}
	\norm{|\cdot-y|^{\frac{s}{2}}Q_0}_{2}^2>\norm{|\cdot|^{\frac{s}{2}}Q_0}_{2}^2, \quad \forall\, y\neq0\,,
\end{equation}
which was used in the proofs of Theorems~\ref{thm:collapse_NLS} and~\ref{thm:collapse_Hartree}, and where $Q_0$ is the (unique) solution of the quintic \NLS equation~\eqref{qNLS_solution}.

Following~\cite[Section 3.3]{LieLos-01}, we say that $f:\R^d\to\R$ is \emph{strictly symmetric-decreasing} (resp. \emph{strictly symmetric-increasing}) if $f(x)=f(y)$ when $|x|=|y|$ and $f(x)>f(y)$ when $|x|<|y|$ (resp. $f(x)<f(y)$ when $|x|<|y|$). We have the following.
\begin{lemma}\label{lem:techn_NLS_strict_symmetry}
	Let $d\geq1$. Let $f$ be strictly symmetric-increasing and $g$ be strictly symmetric-decreasing. Assume $h(0) < +\infty$ for $h:\R^d\to\R$ defined by $h(y) := \int_{\R^d} f(x-y)g(x)\dix$. Then, $h(y)>h(0)$ for any $y\neq0$ such that $h(y) > -\infty$.
\end{lemma}
The inequality~\eqref{ineq:techn_NLS_strict_symmetry} is then a direct application of Lemma~\ref{lem:techn_NLS_strict_symmetry} to $d=1$, $f=|\cdot|^s$ with $s>0$, and $g=Q_0^2$, for which the condition $h(y) < +\infty$ holds for all $y\in\R$ by the exponential decay of $Q_0$. Actually, by the positivity of our $g$ and the nonnegativity of our $f\not\equiv0$ we even have $h(y)\in(0,+\infty)$ for all $y\in\R$. The ideas for the proof of this lemma can be found, e.g., in the proof of a similar result in~\cite{YanYan-17}. We nevertheless give a detailed proof here, for the convenience of the reader.
\begin{proof}
	Let $y\neq0$ with $h(y) > -\infty$. We assume $h(y) < +\infty$ as otherwise the result obviously holds. First, we have
	\[
		\int_{\R^d} f(x-y)g(x)\dix = \int_{\R^d} f(-x-y)g(-x)\dix = \int_{\R^d} f(x+y)g(x)\dix = h(-y)\,,
	\]
	which proves in particular that $h(-y) < +\infty$.
	Second, we compute
	\begin{align*}
		2 [ h(y) - h(0) ] &= h(y) - h(0) + h(-y) - h(0) \\
			&= \int_{\R^d} [f(x-y) - f(x)] g(x)\dix + \int_{\R^d} [f(x+y) - f(x)] g(x)\dix \\
			&= \int_{\R^d} [f(x) - f(x+y)] g(x+y)\dix + \int_{\R^d} [f(x+y) - f(x)] g(x)\dix \\
			&= \int_{\R^d} [f(x) - f(x+y)] [g(x+y) - g(x)]\dix >0\,,
	\end{align*}
	because the last integrand is positive by the properties of $f$ and $g$, and since $y\neq0$.
\end{proof}

\end{document}